\pdfoutput=1
\documentclass[amsmath,amssymb,aps]{revtex4-2}
 
\usepackage{amsmath,amssymb,amsthm}
\usepackage{graphicx}
\usepackage{hyperref}
\usepackage{braket}
\usepackage{accents}
\usepackage{thmtools}
\usepackage{thm-restate}
\usepackage{cleveref}
\usepackage{mathtools}
\usepackage{enumitem}

\DeclarePairedDelimiter{\biggceil}{\bigg\lceil}{\bigg\rceil}

\newcommand{\CSS}{\operatorname{CSS}}

\newcounter{Corollary}
\theoremstyle{plain}
\newtheorem{theorem}{Theorem}
\newtheorem{proposition}{Proposition}
\newtheorem{lemma}{Lemma}
\newtheorem{corollary}[Corollary]{Corollary}
\theoremstyle{definition}
\newtheorem{definition}{Definition}
\newtheorem{construction}{Construction}
\theoremstyle{remark}
\newtheorem*{proof-sketch}{Proof sketch}
\DeclareMathOperator{\Sym}{Sym}
\DeclareMathOperator{\Sp}{Sp}
\DeclareMathOperator{\GL}{GL}
\DeclareMathOperator{\Aut}{Aut}
\DeclareMathOperator{\rowspan}{rowspan}
\DeclareMathOperator{\rank}{rank}

\DeclareMathOperator{\diag}{diag}
\DeclareMathOperator{\oh}{\mathcal{O}}
\DeclareMathOperator{\ftwo}{\mathbb{F}_2}
\DeclareMathOperator{\bA}{\boldsymbol{\mathcal{A}}}
\DeclareMathOperator{\bB}{\boldsymbol{\mathcal{B}}}
\DeclareMathOperator{\bC}{\boldsymbol{\mathcal{C}}}
\DeclareMathOperator{\bD}{\boldsymbol{\mathcal{D}}}

\begin{document}

\title{Quantum Logic Codes: \\Complete Transversal Logical Clifford Instruction Sets \\for High-Rate Stabilizer Quantum Error Correcting Codes}
\author{Adam Holmes}
\email[Corresponding author: ]{adholmes@nvidia.com}
\affiliation{$\mathrm{NVIDIA}$}
\date{\today}

\begin{abstract}
We study the structure and transversal logical capabilities of stabilizer quantum error correcting codes. Among our results, we identify universal lower bounds on circuit depth to generate a full logical Clifford algebra,
and develop novel constructions of logical transversal gates including a new depth-one transversal phase $\mathrm{\overline{S}}$ gate in the rotated surface code and a depth-one intra-block $\mathrm{\overline{CZ}}$ gate in the 2D-toric code that generalizes to all odd distances and all lengths $L\ge3$, respectively. Finally, we construct a high-rate non-LDPC CSS code family with parameters $[[n,\sqrt{n},\Theta({n^{\beta}})]]$
where $\beta \approx 0.2823$ in one demonstrated case, that provably possesses a constant-depth complete 2-local transversal logical Clifford basis instruction set architecture (ISA) composed of all individually targeted $\mathrm{\overline{S}}, \mathrm{\overline{SHS}} = \sqrt{X},$ and $\mathrm{\overline{CZ}}$ gates. This ISA is depth-one for certain subfamilies that we design and generally constant-depth under certain conditions. The code family is built from a small code with parameters $[[n_0, 2, d_0]]$, and is tunable in the standard way: it tiles out to form utility-scale logical qubit counts, and it scales up through concatenation to achieve higher distances and error suppression. We show that this construction preserves the depth-one complete transversal logical Clifford basis ISA when composed with these commuting construction actions, inheriting structure from the core codes so that at scale the complete logical Clifford basis ISA remains depth-one up to depth-two addressable operations between tiled cores. We call these \emph{Quantum Logic Codes}.
\end{abstract}

\maketitle
\tableofcontents

\section{Introduction}
\label{sec:intro}

Building a utility-scale quantum computer requires fault-tolerance techniques of one form or another, and the predominant designs today encode physical qubits into blocks and perform error correction through the use of quantum error correcting codes. A very useful family of codes is the stabilizer code family \cite{gottesman1997stabilizer}, in particular the codespaces formed by the joint eigenspace of an abelian subgroup of the Pauli group. Within this group, the Calderbank-Shor-Steane (CSS) code family is also very useful: those with a stabilizer group decomposable as a direct product of all-$X$ type and all-$Z$ type stabilizer generators. These codes are useful for many reasons, including their natural structure that admits many interesting construction methods to build new codes, clear designs for efficient syndrome extraction circuits, and the ability to perform certain types of transversal logic. We develop a theory of transversal Clifford logic in CSS codes including characterizing the capabilities of physical gates inside CSS codes, and identifying multi-layered constructions of the full logical Clifford group wherein each individual physical gate layer preserves the codespace. The intent with these designs is to enable fault-tolerant logic inside CSS codes in general. A utility-scale system with multi-layered codespace-preserving transversal gate constructions to perform all Clifford logic can avoid the overheads associated with ancillary system gadgetry \cite{webster2025explicit}. This is intended to support full-scale fault tolerant system design along the lines originally posited by Gottesman \cite{gottesman2014constant}, refined and partially realized by more recent architectures including \cite{paetznick2013universal}, then modern \cite{bravyi2024high} and \cite{xu2024constant}.

To this end, we present a theory of transversal Clifford logic in stabilizer codes. Much prior work exists in this area that we build upon, most prominently the recent work establishing multi-layered transversal gate constructions for the entire logical Clifford group in quantum Reed-Muller (QRM) codes \cite{tansuwannont2026construction}, followed by the observation by Chakraborty and Gottesman that transversal gates acting on groups of at least $k$-qubits ($k$-local gates) are required to achieve the full logical Clifford group in codes that encode $k$ logical qubits \cite{chakraborty2026no}. These results are in fact compatible: the former uses strict 1-local transversal gate layers composed with a subset of 2-local transversal gate layers (the so-called \emph{fold-transversal} gates) and permutation automorphisms to build a generating set of the logical Clifford group for QRM codes; the latter can be interpreted as a statement about the power of \emph{single} layers, specifically that there exist logical operators inside the full logical Clifford group of an $[[n,k,d]]$ code that require at least $k$-local physical gates if they are synthesized in a single transversal layer. This is one way to view the result, and a useful one in that it motivates a unified view of transversality focusing on characterizing the logical reach of physical gate layers.

We first use lower bounds on $W$-local transversal layer counts to study the general transversal logical power of stabilizer codes. We show that there are two elementary drivers of depth that are each dominant in distinct code regimes, information transfer from low-weight to high-weight Pauli operators, and Clifford group entropy. We further develop a theory characterizing the logical reach of local transversal gate layers in CSS codes, developing the theory for $W$-local physical gates that preserve the codespace as a single circuit layer, and detail the set of logical operators generated by such layers. Codespace preservation layer-by-layer is not strictly necessary, but it allows for a logical gate construction protocol that is fault tolerant up to some fault density (albeit relaxed from the traditional 1-local transversal gate constructions). A $W$-local gate will propagate a single qubit error into at most a $W$-qubit error, and if the layer preserves the codespace it can be followed by a round of error correction immediately afterward, used to diagnose and analyze these weight-$\leq W$ error patterns which can produce good error suppression in large systems. 

Altogether, this work is a step toward understanding logic in high-rate stabilizer codes. This theory offers constructive routes to synthesis as well by greatly reducing the search space in general, and offering pathways to scalable code families by establishing logical circuit preservation properties under composition and tiling. We make use of all of this in this work, building out the theory of \cite{guyot2025_addressability} for two-locality enabling a fast SAT-solver search to find the logical $S$ gate in the rotated surface code closely inspired as well by \cite{moussa2016transversal}, as well as the pre-existing work \cite{jochymoconnor2014using} inspiring the code construction for \emph{Quantum Logic Codes} that concatenates with self-dual doubly-even codes.

\subsection{Background}
\label{sec:methods}

We approach the study of transversality in stabilizer codes from first principles by studying the space of circuits that can and cannot preserve the codespace, then studying the logical action of such circuits, and by developing a formalism allowing us to generally characterize the logical \emph{reach} of transversal gate layers in general stabilizer codes, defined as the set of logical operations that is realizable with codespace-preserving transversal gates. We further use this formalism to design novel transversal gate constructions for important code families, and to design concrete CSS code families with constant-depth 2-local transversal logical Clifford generator operations.

Our analysis focuses on stabilizer codes that are \emph{indecomposable}. This means that the code cannot be written as a tensor product of smaller codes, or that the code is ``non-splitting'' in the terminology of \cite{guyot2025_addressability}. Decomposable CSS codes are interesting objects themselves, and admit logical transversal blockwise $\mathrm{\overline{CNOT}}$ gates between decomposed blocks. This theory underpins nearly all of the proposed architectures that exist today: a set of physical qubits is partitioned into a set of logical qubits, each encoded in a (perhaps different family of) quantum error correcting code, and logical gadgets are designed to enable logical operations between code blocks in the architecture \cite{cain2026_quasicyclic_lp, webster2026_pinnacle, gidney2025_factor, tripier2026fault} to name a few. Any such system can be written as one large quantum error correction code, with block-diagonal parity check matrices corresponding to each code. \emph{Indecomposable} codes form the atoms of such a theory, and are the focus of our analysis in hopes that they can combine into useful molecules for quantum computing system architectures.

To perform our analysis we study the effect of \emph{transversal} gate layers on stabilizers and logical operators. We follow both \cite{guyot2025_addressability} and \cite{chakraborty2026no} in their definition of transversality. Specifically, we use the following:

\begin{definition}[single-layer transversal W-local gate]
\label{def:wlocal}
A physical qubit circuit $V$ is \emph{single-layer transversal W-local} if it can be written as:

$$ V = \bigotimes_{i} V_i $$ where each $V_i$ is an operator on at most $W$ physical qubits, the $V_i$ operating on disjoint sets of at most $W$ physical qubits that partition $[n]$. 
\end{definition}

We refer to this as a $W$-local transversal gate layer (Figure~\ref{fig:wlocal}). Further, we study the action of \emph{multiple layers} of codespace-preserving transversal gate layers. This allows us to synthesize logical operators that are not contained in the set of logical operations reachable in single layers. We characterize the logical capability that this creates. Throughout this work we refer to circuit \emph{depth}, which we use to mean the number of distinct $W$-local transversal gate layers. This embeds an assumption that a $W$-qubit unitary operator is treated as a depth-one operation. For many real systems this is not the case, however in this work especially for the 2-local cases, this is a good approximation.

\begin{figure*}[h!]
\centering
\includegraphics[width=0.95\linewidth,keepaspectratio]{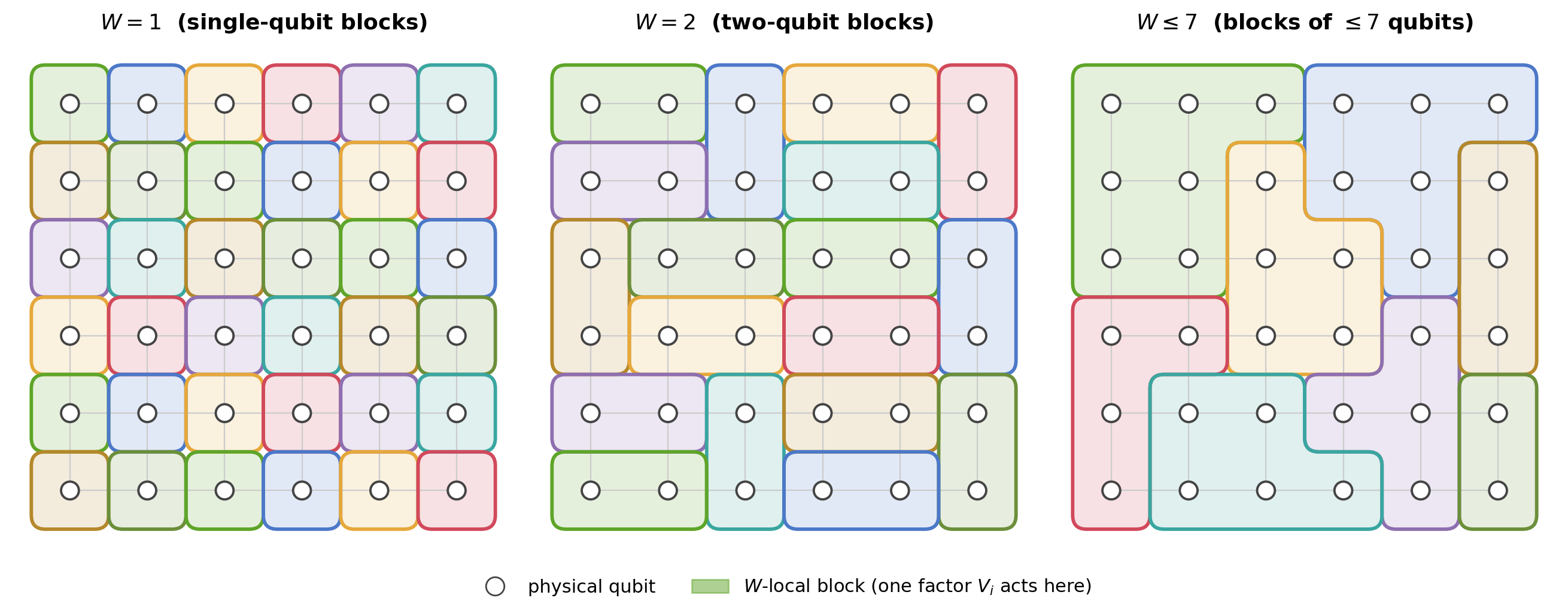}
\caption{Single-layer transversal $W$-local gate layers (Definition~\ref{def:wlocal}) on a shared grid of physical qubits. Such a layer $V=\bigotimes_i V_i$ partitions the qubits into disjoint blocks of at most $W$ qubits and applies one factor $V_i$ within each encircled block. \emph{Left} ($W=1$): every qubit is its own block, so each $V_i$ is a single-qubit gate. \emph{Middle} ($W=2$): blocks of two qubits. \emph{Right} ($W\le 7$): blocks of up to seven qubits. Enlarging $W$ grows the per-block gate space from $\mathrm{Sp}(2,\mathbb{F}_2)$ up to $\mathrm{Sp}(2W,\mathbb{F}_2)$ while preserving the disjoint-block structure that underlies the layer-counting bound.}
\label{fig:wlocal}
\end{figure*}

\subsection{Summary of Results}
\label{sec:results}

First, we show a set of results lower bounding the number of layers required to implement the full logical Clifford group in any CSS stabilizer code. In particular, we show a set of lower bounds on circuit depth $D$ of 2-local gate layers required to implement the full logical Clifford group for a code $Q$ with encoding density $k = f(n)$.

\begin{align}
k = \mathcal{O}(1) &\implies D = \Omega(1) \\
k = \mathcal{O}(\sqrt{n \log n}) &\implies D \geq \Omega\bigg(\log \biggceil{\frac{k}{d \log n}}\bigg)\\
k = n^{\alpha}, 1/2 < \alpha < 1 &\implies D \geq \Omega\bigg(\frac{n^{2\alpha - 1}}{\log n}\bigg)\\
k = \Theta(n) &\implies D \geq \Omega\bigg(\frac{n}{\log n}\bigg)
\end{align}

These all turn out to be special parameter regimes of deeper structure of $W$-local physical gate layers:

\begin{equation}
\label{eq:depth-bound}
D \geq \max \bigg\{\biggceil{\log_W \frac{w^{\star}}{d}}, \biggceil{\log_{N_{n,W}} \frac{|\Sp(2k,\ftwo)|}{|\rho(\Aut_{\mathrm{perm}}(Q))|}} \bigg\} 
\end{equation}

where each ceiling $\lceil\cdot\rceil$ is taken as its positive part $\lceil\cdot\rceil_{+}:=\max\{0,\lceil\cdot\rceil\}$ (Corollary~\ref{corollary:master-bound}), $w^{\star}$ is the \emph{Pauli radius}: the weight of the largest minimum-weight logical Pauli operator over all nontrivial logical classes of the code $Q$, $|\rho(\Aut_{\mathrm{perm}}(Q))|$ is the size of the \emph{logical image} of the permutation automorphism group (the number of distinct logical actions induced by physical permutation automorphisms), and $N_{n,W}$ is the number of codespace-preserving $W$-local gate layers on a code $Q$ over $n$ qubits. This relationship has two terms. The first of these expresses the relationship between the sizes of logical Pauli terms of the code, and captures the requirement that in general a logical Clifford needs to spread information from Paulis to Paulis. The second is an entropy argument calculating the number of bits required to implement all of the logical operators of the Clifford group reduced by the set of code automorphism freedom, and upper bounds the ability of $W$-local codespace preserving gates to supply these operations by counting the number of available layers. We note that an earlier observation was made by \cite{tansuwannont2026construction} where they calculated that $k = \omega(\sqrt{n \log n}) \implies D \geq \Omega\big(\frac{k^2}{n \log n}\big)$. This is the general form of the scaling of depth lower bounds once $k$ grows faster than $\sqrt{n \log n}$, and is a succinct coarse-grained summary of the relationship shown in Equation \ref{eq:depth-bound} in the appropriate parameter regime.

Next, we show a series of results classifying the capabilities of $W$-local transversal gate layers themselves in CSS codes. For an indecomposable quantum CSS code $Q$, and where \emph{addressable} logical gates are precisely those that leave at least one logical qubit invariant under the action:

\begin{itemize}[label=--]

\item To perform a logically addressable gate with a single 1-local physical gate restricts the physical per-qubit gate operators to the set $\{I, S, SHS\}$. This was discovered in \cite{guyot2025_addressability}.
\item The addressable 1-local single-layer logical reach already lies in the $2$-local set $\mathcal{G}_{2L}(K)$ below, with logical blocks $L_X\diag(\vec c)L_X^\top$ and $L_Z\diag(\vec b)L_Z^\top$. They are symmetric but not necessarily diagonal, so even 1-local layers can realize entangling couplings such as $\mathrm{\overline{CZ}}_{ij}$.
\item A logically addressable single 2-local layer cannot enact a global $X\!\leftrightarrow\!Z$ sector exchange, and, under the confinement hypothesis below, its remaining non-trivial-diagonal freedom is a logical basis change (Lemmas~\ref{lem:confine-singular}--\ref{lem:confine-shear}). The \emph{trivial-diagonal} blocks ($A_{V_p}=D_{V_p}=I$) form the phase/$CZ$ set $\mathcal{G}_{\text{diag}}$, built from the gate types $\big\{ I, S, SHS, CZ, CZ^{\sharp}, \mathcal{J} \big\}$ ($|\mathcal{G}_{\text{diag}}|=18$, Prop.~\ref{prop:7b}), where $\mathcal{J}$ is the all-ones entangler and $CZ^{\sharp} = (H \otimes H)CZ(H \otimes H)$. 

\item Up to code automorphisms, and under the confinement hypothesis of Proposition~\ref{prop:8}, the symplectic form of the logical addressable 2-local single-layer reach is contained in the $2$-local set
$$
\mathcal{G}_{2L}(K) := \bigg\{\begin{pmatrix} \mathbf{I} & \mathbf{\mathcal{B}} \\ \mathbf{\mathcal{C}} & \mathbf{I} \end{pmatrix} \quad | \quad \mathbf{\mathcal{B}}, \mathbf{\mathcal{C}} \in \mathrm{Sym}_K, \mathbf{\mathcal{B}} \mathbf{\mathcal{C}} = 0\bigg\}.
$$
\item The logical addressable 2-local single-layer reach is contained in $\GL(K,\mathbb{F}_2)\cdot\mathcal{G}_{2L}(K)$ under the same confinement hypothesis; when the basis changes are realized by permutation automorphisms it is contained in the set composed with the logical automorphism group $\mathrm{Aut}_{\text{perm}}^{\text{log}}(Q)$. The basis changes are contributed by the invertible (automorphism-type) blocks above, and the matching-realizable elements form the subset $\mathcal{S}^{\mathrm{gen}}_{2L+\pi}(Q)$ (Definition~\ref{def:S2Lpi}) of the displayed set:
\begin{widetext}
$$ 
\mathcal{G}_{2L + \pi}(Q) = \bigg \{ \begin{pmatrix} \mathcal{A} & \mathcal{B}\mathcal{A}^{-T} \\ \mathcal{C}\mathcal{A} & \mathcal{A}^{-T} \end{pmatrix} \quad | \quad \mathcal{A} \in \mathrm{Aut}_{\text{perm}}^{\text{log}}(Q), \mathcal{B},\mathcal{C} \in \mathrm{Sym}_K(\mathbb{F}_2), \mathcal{B}\mathcal{C} = 0  \bigg \}.
$$
\end{widetext}

\item For any locality $2 \leq W < n$, the addressable single-layer $W$-local Clifford reach is, under the same hypothesis: $
\mathcal{T}^{WL}_{\mathrm{addr}}(Q) \subseteq \GL(K,\mathbb{F}_2)\cdot\mathcal{G}_{2L}(K),$ with $\mathcal{G}_{2L+\pi}(Q)$ in place of the right-hand side when the basis changes are permutation-automorphism actions.
Enlarging the physical locality to $W$ increases the per-block parameter space, yields more physical circuits that can construct particular addressable logical operations, but does not enable physical circuits to reach new addressable logical operations outside of this set.

\item Algebraically, the multi-layer closure of the $2$-local set is the full logical Clifford group: $\langle \mathcal{G}_{2L}(K) \rangle = \mathrm{Sp}(2K, \mathbb{F}_2)$. Whether a given code physically attains this depends on which set generators it realizes as codespace-preserving $W\le 2$ layers; the Quantum Logic Codes of Section~\ref{sec:qlc} realize a complete such library.

\end{itemize}

The final section of results focuses on novel constructions. In spite of the hardness of the problem in general, interesting code families and gate targets can be found with SAT-solvers. In particular, we present a novel single-layer 2-local transversal $\mathrm{\overline{S}}$ gate in the rotated surface code and a single-layer 2-local transversal intracode $\mathrm{\overline{CZ}}$ in the 2D-toric code. The surface code construction provably preserves code distance for every odd $d$ and is single-fault-tolerant for $d\ge5$, and it generalizes to all odd distances. We also show CSS code instances that each have a concrete \emph{complete transversal logical Clifford Instruction Set Architecture (ISA)}, defined as a set of transversal 2-local circuits implementing a full logical basis of the Clifford group. Here we focus on the basis $\{\mathrm{\overline{S}}_i, \mathrm{\overline{SHS}}_i, \mathrm{\overline{CZ}}_{i,j}\}$ for all $i \neq j \in [k]$. We design a family construction that produces codes of asymptotic parameters $[[n, \sqrt{n}, \Theta(n^{0.2823})]]$
built from self-dual cores $(n_0,d_0)\in\{(4,2),(18,5),(20,6)\}$ by $r$-fold tiling and $\ell$ levels of concatenation with $[[7,1,3]]$ inner codes. The displayed asymptotic form takes the $(4,2)$ core with $r=7^{\ell}$, and the other cores give $k=\Theta(\sqrt{n})$ with the same exponent. This family provably possesses a complete transversal logical Clifford basis ISA at circuit depth independent of $\ell$ and bounded in $r$ (Proposition~\ref{prop:qlc-depth}), and so we call them \emph{Quantum Logic Codes}.

\subsection{Related Work}
\label{sec:related}

Much work has already been done in this area, and the closest related work is listed here first. To start, \cite{chakraborty2026no} proved a limitation on single-layer locality for transversal gates to be capable of producing the full logical Clifford group. This result combines with the earlier result of \cite{tansuwannont2026construction}, wherein they prove that the Quantum Reed Muller code family has a logical Clifford transversal gate instruction set, with layer counts for addressable $\mathrm{\overline{S}}$ gates scaling as  $\oh(\sqrt{n})$. Another earlier work \cite{guyot2025_addressability} studied the requirements of physical transversal gates within CSS codes. This work showed significant restrictions on the gate set of 1-local gates, and showed insightful bounds on encoding rate stemming from indecomposability. We directly build out these results to 2-local and $W$-local transversal layers, for codespace-preserving addressable logical operations. For hypergraph product (HGP) codes, earlier work \cite{patra2025targeted} derived explicit symplectic generators of the logical Clifford group and synthesized these logical operators physically, though they are not fault tolerant. Prior work in this area is \cite{rengaswamy2018_logical_clifford} which established the groundwork for symplectic synthesis as a whole, and greatly inspired this continuation work. Other work \cite{liang2025selfdual} pursued designing self-dual group-algebra two-block code families, specifically bivariate bicycle codes, and wrote out global, fault tolerant logical transversal operations for $\mathrm{\overline{CNOT}}$ between code blocks, global $\mathrm{\overline{H}}$ and $\mathrm{\overline{S}}$, but they do not provide addressability, so these gates cannot generate the full $\Sp(2k,\ftwo)$ group. Code automorphisms are studied in \cite{sayginel2025faulttolerant}, which produces an algorithmic and computational tool that computes and emits the physical circuits and logical action of a code's automorphism group. This result is constructive and quite useful; our sections on automorphisms are informed by their work, and on \cite{grassl2013leveraging}, who performed foundational work in studying code automorphism groups.

There is also much prior work establishing no-go theorems, connecting to locality, and establishing bounds related to synthesis of operations in the Clifford hierarchy. First, the seminal result \cite{eastin2009restrictions} establishes the non-existence of universal transversal gate sets for $W=1$ local transversal gate sets, and \cite{bravyi2013classification} established the connection between circuit locality and the ability of circuits to reach logical gates at high levels of the Clifford hierarchy, which was later strengthened by \cite{pastawski2015faulttolerant} and extended to subsystem codes. Most directly relevant to our depth bounds, \cite{jochymoconnor2018disjointness} introduced the \emph{disjointness} $\Delta(Q)$ of a stabilizer code and used it to bound the level of the Clifford hierarchy reachable by transversal and constant-depth circuits. Our information-transfer bound (Theorem~\ref{theorem:pauli-spread}) is a Clifford-level specialization connected to their support-spreading argument. Other works relevant to the study of logic in quantum error correcting codes include \cite{zeng2011transversality} wherein the connection between doubly-even codes supporting transversal phase gates, and self-dual codes supporting transversal Hadamard gates is established; \cite{anderson2016classification} classifying the diagonal physical gates that yield transversal logical gates within which our bounds occupy a special case; \cite{fu2025nogo} recently extend Bravyi-König to product codes.

On the code construction and logical operation design side, there has been a large body of work designing codes for logic. One of the most notable is the Subsystem Hypergraph Product Simplex codes \cite{malcolm2025_shyps} construction using gauge-fixed subsystem codes. A seminal result \cite{yoder2016universal} reasoned about \emph{pieceable fault-tolerance}, which is a concept that deeply inspires our approach to synthesize logical circuitry out of codespace-preserving logical layers. Recent work \cite{liu2026self} provides a way to explicitly construct self-dual LDPC quantum codes by stacking non-self-dual codes, which adds a greatly enabling symmetry that often enhances logically accessible transversal gates. A sequence of papers starting with \cite{breuckmann2026cups} studies the structure of cochain complex CSS code representations and derives logical operations within that structure. This is a way of revealing structure that our present work deals with both linear algebraically and through formulations amenable to numerical solvers. Further work \cite{breuckmann2024fold} studies the general structure of fold-transversal gates and generalizes them to arbitrary CSS codes with ZX-dualities. A similar construction to our novel transversal phase gate in rotated surface codes is given in \cite{moussa2016transversal}, who derives theirs via a folding of the surface code, while ours operates solely on data qubits of the code block. Concatenation for universality is studied in \cite{jochymoconnor2014using}, and we design a Quantum Logic Code family in a similar fashion. 

\subsection{Framework}

We will work with the set of $n$-qubit Pauli operators, denoted $\mathcal{P}_n$. The set [$n$] denotes the set $\{1, 2, \cdots, n\}$. We follow \cite{chakraborty2026no} and work with the projective Clifford group on $n$ qubits, defined as the Clifford group modulo phases and Paulis: $\mathrm{Proj}\ \mathrm{Cl}_n = \mathrm{Cl}_n / (U(1), \mathcal{P}(n)) $. This group is isomorphic to the classical symplectic group of degree $2n$, denoted $\mathrm{Sp}(2n,\mathbb{F}_2)$.

We will also work simultaneously in the group-theoretic and linear-algebraic representations for our analysis. A stabilizer code on $n$ qubits is in this representation a joint $+1$ eigenspace of an abelian subgroup $\mathcal{S}$ of the $n$-qubit Pauli group $\mathcal{P}_n$, and we refer to $\mathcal{S}$ as the \emph{stabilizer group} of the code. In binary symplectic coordinates $\mathcal{S}$ is represented as an \emph{isotropic} subspace $S\subseteq\mathbb{F}_2^{2n}$. It is isotropic in the symplectic form, which reflects the fact that all stabilizers commute. If $S$ has rank $m$, the code encodes $k = n-m$ logical qubits, and its \emph{codespace} $\mathcal{C}$, the simultaneous $+1$ eigenspace inside $(\mathbb{C}^2)^{\otimes n}$, is a Hilbert subspace of complex dimension $\dim_{\mathbb{C}}\mathcal{C} = 2^{k}$. The code's \emph{distance} $d$ is the minimum Pauli weight of a nontrivial logical operator, or the minimum weight of an element of $N(\mathcal{S})\setminus\mathcal{S}$. Equivalently this is expressed in symplectic coordinates via $S^{\perp}\setminus S$, with $S^{\perp}$ the symplectic dual of $S$.

A CSS code $Q = \mathrm{CSS}(H_X, H_Z)$ on $n$ qubits with parity check matrices $H_X \in \mathbb{F}_2^{m_X \times n}, H_Z \in \mathbb{F}_2^{m_Z \times n}$ satisfying $H_X H_Z^{\top} = 0$ encodes $k = n - \mathrm{rank}(H_X) - \mathrm{rank}(H_Z)$ logical qubits. We refer to the rowspans of $H_X$ and $H_Z$ as $S_X$ and $S_Z$, respectively, so that the stabilizer subspace splits in symplectic coordinates as $S = S_X \oplus S_Z$ (an $X$-type part and a $Z$-type part). The quantum \emph{codespace} $\mathcal{C}\subseteq(\mathbb{C}^2)^{\otimes n}$ is the simultaneous $+1$ eigenspace of $\mathcal{S}$, of complex dimension $2^k$. The \emph{logical representatives} $L_X \in \mathbb{F}_2^{k \times n}$ and $L_Z \in \mathbb{F}_2^{k \times n}$ are chosen as bases of $\mathrm{ker}\,H_Z / \mathrm{rowspan}(H_X)$ and $\mathrm{ker}\,H_X / \mathrm{rowspan}(H_Z)$, respectively, with $L_X L_Z^{\top} = \mathbf{I}_k$, and they label the nontrivial logical $X$- and $Z$-classes. 

A $W$-local single-layer Clifford gate can be written as $V = \bigotimes_{m} V_m$, where $V_m \in \mathrm{Sp}(2W, \mathbb{F}_2)$ act on disjoint physical qubit sets $\{q_{i_1}, q_{i_2}, \cdots, q_{i_W}\}$, partitioning [$n$]. The \emph{partition} $P = \{ (q_{i_1}, q_{i_2}, \cdots, q_{i_W})_i \}$ is a grouping of $n$ qubits into sets of size at-most $W$. We will be focused on $W=2$ primarily throughout our work, where partitions are matchings on data qubits of the code. A useful piece of framework we borrow from \cite{guyot2025_addressability} is that we will assume that physical single-layer 1-local Clifford gates have been compiled down to only contain elements \{$I$, $H$, $S$, $SH$, $HS$, $SHS$\}, followed by Pauli operators (which we mostly omit). 

We will be focused on physical Clifford layers $V$ that are \emph{addressable}.
\begin{definition}[Addressable layer]
\label{def:addressable}
A physical Clifford operator layer $V$ is \emph{addressable} on a code $Q$ if $V \mathcal{S} V^{\dagger} = \mathcal{S}$ (codespace preservation) and the induced logical action $\Phi(V) \in \mathrm{Sp}(2k,\mathbb{F}_2)$ leaves at least one logical qubit fully invariant: there is an index $i$ with $\Phi(V)\,\overline{X}_i = \overline{X}_i$ \emph{and} $\Phi(V)\,\overline{Z}_i = \overline{Z}_i$.
\end{definition}

\begin{definition}[Targeted layer]
\label{def:targeted}
A physical Clifford operator layer is \emph{targeted} if it is a single codespace-preserving physical Clifford layer $V$ realizing a prescribed logical Clifford target (such as a single-qubit $\mathrm{\overline{S}}_i$ or a two-qubit $\mathrm{\overline{CZ}}_{ij}$), with \emph{no} requirement that any logical qubit be left fully invariant. \end{definition}

Throughout our analysis we will study the symplectic form of physical and logical gates. In particular, $W$-local physical gates $V_p$ on partitioned blocks $p$ of qubits with index sets $p \subseteq [n]$ with $|p| = W$ will have the general form:

$$
V_p = \begin{pmatrix} {\mathcal{A}} & {\mathcal{B}} \\ {\mathcal{C}} & {\mathcal{D}} \end{pmatrix},
$$

where each $\mathcal{A}, \mathcal{B}, \mathcal{C}, \mathcal{D}$ are block-matrices of dimension $W \times W$. These combine in direct products over each partition block $p$ of a partition $P$ into full $W$-local gate layers as:

$$
 \mathcal{V}\;=\;\bigoplus_{p=1}^{\ell} g_p
\;=\; \begin{pmatrix} \boldsymbol{\mathcal{A}} = \oplus_p\ { \mathcal{A}_p}, & \boldsymbol{\mathcal{B}} = \oplus_p\ { \mathcal{B}_p} \\ \boldsymbol{\mathcal{C}} = \oplus_p\ { \mathcal{C}_p}, & \boldsymbol{\mathcal{D}} = \oplus_p\ { \mathcal{D}_p} \end{pmatrix}.
$$

The physical blocks $V_p$ and joint layer $\mathcal{V}$ also satisfy the symplectic matrix conditions. Defining the following symplectic forms $\omega$ in the physical block dimension $W$ and $\Omega$ in the full layer dimension $n$ as:

\[
\omega \;=\; \begin{pmatrix} 0 & I_W \\ I_W & 0 \end{pmatrix}\in\mathbb{F}_2^{\,2W\times 2W},
\qquad
\Omega \;=\; \bigoplus_{p=1}^{\ell}\omega \;=\; I_\ell\otimes\omega \;\in\;\mathbb{F}_2^{\,2n\times 2n}.
\]

Here $\omega$ and $\Omega$ are written in per-block (interleaved $X_pZ_p$) coordinates; the global $X\,|\,Z$ ordering used for logical images from Section~\ref{sec:wlocal} onward is obtained by a fixed sorting permutation $P_{\mathrm{sort}}$ that conjugates $\Omega$ to $\left(\begin{smallmatrix}0&I_n\\I_n&0\end{smallmatrix}\right)$ while preserving block-diagonality. We require then that the physical blocks satisfy:
\[
V_p\in\mathbb{F}_2^{\,2W\times 2W},\qquad
V_p^{\!\top}\,\omega\,V_p \;=\; \omega
\quad\Longleftrightarrow\quad
V_p\in\mathrm{Sp}(2W,\mathbb{F}_2),
\qquad p=1,\dots,\ell
\]

and further that the full layer of physical blocks satisfies: 

\[
\mathcal{V}^{\!\top}\,\Omega\,\mathcal{V}
\;=\;\bigoplus_{p=1}^{\ell} V_p^{\!\top}\,\omega\,V_p
\;=\;\bigoplus_{p=1}^{\ell}\omega
\;=\;\Omega
\qquad\Longrightarrow\qquad
\mathcal{V}\in\mathrm{Sp}(2n,\mathbb{F}_2).
\]

We will require that individual $W$-local layers are \emph{codespace-preserving}. To satisfy this property, the physical operators must preserve the stabilizer subspace, which in this section is assumed to split into $X$-type and $Z$-type subspaces for CSS codes. Concretely then, we require that:

$$ \text{for all } a\in S_X,\ b\in S_Z:\qquad \mathcal{V}X^aZ^b\mathcal{V}^{\dagger} = X^{a'}Z^{b'} \implies a' \in S_X,\ b' \in S_Z.$$

\section{Codespace-Preserving Transversal Circuit Depth Bounds}

We start by studying transversal logic information-theoretically by considering the requirements needed to perform the full logical Clifford operation set, and how the set of $W$-local physical gate layers can achieve these requirements. We provide bounds on what can be done logically when any $W$-local physical Clifford operators are allowed. Throughout this section, $\log$ denotes the base-$2$ logarithm, so that $\log|\Sp(2k,\ftwo)|$ counts the bits required to specify a logical Clifford. The following result is a Corollary formed by taking the maximum over two distinct pieces shown in Theorems \ref{theorem:pauli-spread} and \ref{theorem:clifford-entropy}:

\begin{corollary}[$W$-Local Circuit Depth Bound for Stabilizer Codes]
\label{corollary:master-bound}
Let $Q$ denote an $[[n,k,d]]$ stabilizer code with maximum logical Pauli operator weight $w^{\star}$ and permutation automorphism group $\Aut_{\mathrm{perm}}(Q)$, and let $W\ge2$. If the code $Q$ is capable of synthesizing the full logical Clifford group using codespace-preserving $W$-local transversal gate layers, then the circuit depth $D$ required is bounded from below by:

$$ D \geq \max \bigg \{ \biggl\lceil\log_W \frac{w^{\star}}{d}\biggr\rceil_{+}, \biggl\lceil\log_{N_{n,W}} \frac{|\Sp(2k,\ftwo)|}{|\rho(\Aut_{\mathrm{perm}}(Q))|}\biggr\rceil_{+} \bigg \}. $$
where each term is taken as its positive part $\lceil\cdot\rceil_{+}:=\max\{0,\lceil\cdot\rceil\}$.
\end{corollary}

Both of these terms show up for stabilizer codes in different regimes. The latter is a more well-known bound that counts the required specification bits needed to build the full logical Clifford group on $k$ qubits out of $W$-local layers which increases dramatically as a function of $k$, while the former is a count of the size of a light-cone emanating from the $W$-local layers and relates this to the spread of weights of logical operators. For high-rate codes with $k = \omega(\sqrt{n \log n})$, the latter bound dominates as $D \geq \Omega\big(\frac{k^2}{n \log n}\big)$. For low rate codes though, for instance when $k = \oh(1)$, this bound is vacuous, but the depth as a whole may still be bounded by the information transfer bound: $D \geq \log_W \frac{w^{\star}}{d}$. 

We can illustrate these situations with specific examples. Consider first a $k=1$ highly asymmetric code such that $d_x \gg d_z$ or $d_x = \alpha d_z$ for some $\alpha > 1$. Concretely we can consider a rectangular patch of a rotated surface code, as would exist inside the bulk of a lattice surgery operation in a full quantum system. For these codes, the Clifford entropy bound is rendered trivial: $D \geq 1$ taking the appropriate ceiling. The Pauli reach bound is active: with $d_{\mathrm{min}}=\min(d_x,d_z)$ the code distance and $w^{\star}$ the maximum over \emph{all} nontrivial logical classes (including the mixed $\overline Y=\overline X\overline Z$ class) of the class's minimum-weight representative, $w^{\star}\ge\max(d_x,d_z)$, so $D \geq \log_W \frac{w^{\star}}{d_{\mathrm{min}}} \ge \log_W \alpha$ with $\alpha=\max(d_x,d_z)/\min(d_x,d_z)$. For a rectangular rotated surface code of $\alpha = 2^m$ aspect ratio, using $2$-local layers requires a circuit depth then of at least $m$ layers (and more if the $\overline Y$ class is heavier). Conversely, consider a constant rate code with $k = \Theta(n)$ with linear distance $d = \Theta(n)$. Here, since $w^{\star}\ge d$ always while $w^{\star}\le n$ and $d=\Theta(n)$, the ratio is bounded, $w^{\star}/d=\Theta(1)$, so the spread term $\log_W(w^{\star}/d)=\oh(1)$ contributes only a constant; the Clifford entropy bound far exceeds it and dominates as a depth bound at $D \geq \Omega\big(\frac{n}{\log n}\big)$. We can see this as a spectrum, and organize a parameter regime scaling table:

\begin{corollary}[$W$-Local Depth Bound Parameter-Regime Scaling Regimes]
\label{corollary:scaling-table}
\[
D \;\ge\; \max\!\Bigl\{
\underbrace{\bigl\lceil \log_W(w^\star/d)\bigr\rceil_+}_{\text{spread}},\;\;
\underbrace{\Bigl\lceil \log_{N_{n,W}}\!\tfrac{|\mathrm{Sp}(2k,2)|}{|\rho(\mathrm{Aut}_{\mathrm{perm}}(Q))|}\Bigr\rceil_+}_{\text{counting}}
\Bigr\}
\qquad (w^\star=\text{Pauli radius})
\]

\begin{center}
\renewcommand{\arraystretch}{1.8}
\setlength{\tabcolsep}{10pt}
\begin{tabular}{l c c l}
\hline\hline
$k$ &
$\log_W\!\dfrac{w^\star}{d}$ &
$\log_{N_{n,W}}\!\dfrac{|\mathrm{Sp}(2k,2)|}{|\rho(\mathrm{Aut}_{\mathrm{perm}}(Q))|}$ &
$\mathrm{Dominant}\ \mathrm{scaling}$ \\[4pt]
\hline
$\oh(1)$ &
$\oh(\log_W n)$ &
$\Theta\!\bigl(\tfrac{1}{n\log n}\bigr)\!\to\!0$ &
$\oh(\log_W n)$ \quad (spread; counting vacuous) \\

$\Theta\!\bigl(\sqrt{n\log n}\bigr)$ &
$\oh(\log_W n)$ &
$\Theta(1)$ &
$\oh(\log_W n)$ \quad (spread; counting turns on) \\

$n^{\alpha},\;\tfrac12<\alpha<1$ &
$\oh(\log_W n)$ &
$\Theta\!\bigl(\tfrac{n^{2\alpha-1}}{\log n}\bigr)$ &
$\Theta\!\bigl(\tfrac{n^{2\alpha-1}}{\log n}\bigr)$ \quad (counting) \\

$\Theta(n)$ &
$\oh(\log_W n)$ &
$\Theta\!\bigl(\tfrac{n}{\log n}\bigr)$ &
$\Theta\!\bigl(\tfrac{n}{\log n}\bigr)$ \quad (counting) \\
\hline\hline
\end{tabular}
\end{center}

The counting term becomes non-vacuous once $k = \omega(\sqrt{n \log n})$, and overtakes the worst-case spread ceiling $\oh(\log_W n)$ only once $k = \omega(\sqrt{n}\,\log n)$; between these thresholds the dominant term depends on the code's $w^\star/d$ and $|\rho(\Aut_{\mathrm{perm}}(Q))|$ at fixed $W$. Below $\sqrt{n\log n}$, the spread of information into high-weight logical operators is the only potentially non-vacuous term in this regime.
\end{corollary}

\section{CSS Transversality Theory}
In this section we study the power of $W$-local transversal gates on CSS codes. Throughout the structural results we write $K := k$ for the number of logical qubits. We remain focused on \emph{indecomposable} CSS codes, and now restrict our focus further to \emph{addressable} gates (Def.~\ref{def:addressable}), which requires that at least one logical \emph{qubit} remain fully invariant under the logical action of the gate.
Full proofs of propositions are deferred to the Supplementary Material section \ref{sec:supp}.

\subsection{1-Local Transversality Theory}
A $1$-local single-layer Clifford is a tensor product $V = \bigotimes_{i \in [n]} V_i$ of per-qubit single-qubit Cliffords $V_i \in \Sp(2,\mathbb{F}_2)$, drawn from the six-element set $\{I, H, S, SH, HS, SHS\}$ (the compiled single-qubit Clifford set introduced in Sec.~\ref{sec:methods}). The Guyot--Jaques addressability criterion restricts which $V$ can be addressable on an indecomposable CSS code.

\begin{proposition}[Admissible 1-local layers, Guyot--Jaques~\cite{guyot2025_addressability}]
\label{prop:1local-admissible}
Let $Q = \mathrm{CSS}(H_X, H_Z)$ be an indecomposable CSS code and let $V = \bigotimes_i V_i$ be a $1$-local single-layer Clifford. Then $V$ is addressable on $Q$ only if every per-qubit factor lies in the set
$$
\mathcal{G}^{(1)}_{\mathrm{diag}} \;:=\; \{I,\, S,\, SHS\} \;\subset\; \Sp(2,\mathbb{F}_2).
$$
\end{proposition}

\begin{proposition}[1-local reach]
\label{prop:1local-reach}
For an indecomposable CSS code $Q$ encoding $K$ logical qubits, every addressable $1$-local single-layer Clifford is encoded by flag vectors $\vec b,\vec c\in\mathbb{F}_2^n$ ($b_q = 1$ for a physical $SHS_q$, $c_q = 1$ for a physical $S_q$, $\vec b\wedge\vec c = 0$ by Proposition~\ref{prop:1local-admissible}), and its symplectic image takes the form
$$
\Phi(V) \;=\; \begin{pmatrix} \mathbf{I}_K & L_Z\diag(\vec b)L_Z^\top \\ L_X\diag(\vec c)L_X^\top & \mathbf{I}_K \end{pmatrix},
$$
with both off-diagonal blocks symmetric and of product zero, i.e.\ an element of the $2$-local set $\mathcal{G}_{2L}(K)$ of Proposition~\ref{prop:8}. The blocks need not be diagonal: a $1$-local layer can realize entangling couplings such as $\mathrm{\overline{CZ}}_{ij}$. This is the $1$-local specialization of the proof of Proposition~\ref{prop:8} with $B = \diag(\vec b)$, $C = \diag(\vec c)$.
\end{proposition}

\subsection{2-Local Transversality Theory}
A $2$-local single-layer Clifford $V = \bigotimes_p V_p$ has per-block factors $V_p \in \Sp(4,\mathbb{F}_2)$ acting on disjoint pairs $(q_i, q_j)$ that partition $[n]$. Writing $V \in \Sp(4,\mathbb{F}_2)$ in $2\times 2$ blocks,
$$
V \;=\; \begin{pmatrix} A_V & B_V \\ C_V & D_V \end{pmatrix},
$$
the symplectic condition is equivalent to the three identities
\begin{align}
&A_V^\top C_V \;=\; C_V^\top A_V \nonumber \\ &B_V^\top D_V \;=\; D_V^\top B_V, \nonumber \\
&A_V^\top D_V + C_V^\top B_V \;=\; I_2. \label{eq:sp4-identities}
\end{align}

\begin{proposition}[Admissible per-block set $\mathcal{G}_{\mathrm{diag}}$]
\label{prop:7b}
Let $Q$ be an indecomposable CSS code and let $\mathcal{V} = \bigotimes_p V_p$ be a codespace-preserving, addressable (Def.~\ref{def:addressable}) $2$-local single-layer Clifford. Then the layer is not a global sector exchange (not every $A_{V_p}$ is zero), and a block with \emph{trivial diagonal} $A_{V_p}=D_{V_p}=I_2$ lies in the phase/CZ set
\begin{widetext}
$$
\mathcal{G}_{\mathrm{diag}} \;:=\; \Bigl\{\, V \in \Sp(4,\mathbb{F}_2) \;:\; A_V = D_V = I_2,\; B_V,C_V \in \Sym_2(\mathbb{F}_2),\; B_V C_V = 0 \,\Bigr\},\qquad |\mathcal{G}_{\mathrm{diag}}| = 18.
$$
\end{widetext}
\end{proposition}
The eighteen \emph{algebraically allowed} trivial-diagonal phase/CZ-type classes form $\mathcal{G}_{\mathrm{diag}}$: eight $C$-only with $B_g=0$, eight $B$-only with $C_g=0$, three mixed rank-one with both $B_g,C_g$ nonzero, the identity counted in each subset. Codespace-preservation selects which of them is realizable as an actual layer on $Q$, and depends on the stabilizer structure of $Q$ itself. A block with invertible $A_{V_p}\ne I_2$ or $D_{V_p}\ne I_2$ is a sector-preserving \emph{basis change} (e.g.\ a $\mathrm{SWAP}$): it contributes only to the logical $\GL(K,\mathbb{F}_2)$ basis-change factor of Proposition~\ref{prop:8}. In particular, if $\mathcal{V}$ induces no logical basis change then every $V_p\in\mathcal{G}_{\mathrm{diag}}$. Singular diagonal blocks are controlled at the logical level by Lemma~\ref{lem:confine-singular}.

\begin{proposition}[2-local Logical Reach]
\label{prop:8}
For any codespace-preserving, addressable (Def.~\ref{def:addressable}) $2$-local single-layer Clifford $V$ on an indecomposable CSS code $Q$ encoding $K$ logical qubits, suppose that no nonzero element of $\ker H_Z$ or of $\ker H_X$ is supported entirely within the union of the supports of the blocks with singular $A_p$ or with $A_p^{\top}D_p\ne I_2$ (the \emph{confinement hypothesis}; Lemmas~\ref{lem:confine-singular} and~\ref{lem:confine-shear}). Then the symplectic image lies in $\mathcal{A}\cdot\mathcal{G}_{2L}(K)$ for a logical basis change $\mathcal{A}\in\GL(K,\mathbb{F}_2)$; when the basis change is realized by a permutation automorphism, $\mathcal{A}\in\rho(\Aut_{\mathrm{perm}}(Q))$ and the image lies in the automorphism-augmented set $\mathcal{G}_{2L+\pi}(Q) = \rho(\Aut_{\mathrm{perm}}(Q))\cdot\mathcal{G}_{2L}(K)$. Here
$$
\mathcal{G}_{2L}(K) \;:=\; \biggl\{ \begin{pmatrix} \mathbf{I}_K & \boldsymbol{\mathcal{B}} \\ \boldsymbol{\mathcal{C}} & \mathbf{I}_K \end{pmatrix} \biggm|\; \boldsymbol{\mathcal{B}}, \boldsymbol{\mathcal{C}} \in \Sym_K,\; \boldsymbol{\mathcal{B}}\boldsymbol{\mathcal{C}} = 0 \biggr\}
$$
is the $2$-local set and $\rho(\Aut_{\mathrm{perm}}(Q))\le\GL(K,\mathbb{F}_2)$ is the logical automorphism action (Proposition~\ref{prop:pi-logical}). On a canonical basis $L_X, L_Z$, the set $\mathcal{G}_{2L}(K)$ is composed of $\boldsymbol{\mathcal{C}}_{ii} = 1 \Leftrightarrow \mathrm{\overline{S}}_i$, $\boldsymbol{\mathcal{C}}_{ij} = 1 \Leftrightarrow \mathrm{\overline{CZ}}_{ij}$, and dually $\boldsymbol{\mathcal{B}}_{ii} \Leftrightarrow \mathrm{\overline{SHS}}_i$, $\boldsymbol{\mathcal{B}}_{ij} \Leftrightarrow \mathrm{\overline{CZ}}^{\sharp}_{ij}$.
\end{proposition}

\subsection{W-Local and Permutation Automorphism Transversality Theory}
\label{sec:wlocal}

A $W$-local single-layer Clifford partitions qubits into subsets of size at most $W$, with per-block factors in $\Sp(2W, \mathbb{F}_2)$. Raising the locality $W$ enlarges the per-block parameter space, but it does not enlarge the addressable logical reach until $W$ reaches $n$.

\begin{proposition}[$W$-invariance of logical reach, $W \ge 2$]
\label{prop:Winvariance}
Let $Q$ be an indecomposable CSS code on $n$ qubits encoding $K$ logical qubits, and let $2 \le W < n$. Every codespace-preserving, addressable (Def.~\ref{def:addressable}) $W$-local single-layer Clifford $V$ has logical image $\Phi(V)$ given by
$$
\Phi(V)\;\in\;\mathcal{A}\cdot\mathcal{G}_{2L}(K),\qquad \mathcal{A}\in\GL(K,\ftwo),
$$
under the confinement hypothesis of Proposition~\ref{prop:8}, applied to the layer's blocks at width $W$; when the basis change $\mathcal{A}$ is realized by a permutation automorphism, $\mathcal{A}\in\rho(\Aut_{\mathrm{perm}}(Q))$ and $\Phi(V)\in\mathcal{G}_{2L+\pi}(Q)=\rho(\Aut_{\mathrm{perm}}(Q))\cdot\mathcal{G}_{2L}(K)$.
\end{proposition}
This is the $2$-local set of Proposition~\ref{prop:8} composed with the logical automorphism group $\rho(\Aut_{\mathrm{perm}}(Q))\le\GL(K,\ftwo)$ of Proposition~\ref{prop:pi-logical}. The addressable $W$-local single-layer reach lies in $\GL(K,\ftwo)\cdot\mathcal{G}_{2L}(K)$ for every $W\ge 2$ (in $\mathcal{G}_{2L+\pi}(Q)$ when the basis changes are permutation-automorphism actions): raising the locality enlarges the per-block parameter space and supplies more physical circuits for a given logical target, but reaches no logical operation outside this set. In case $A^{\mathrm{log}}=D^{\mathrm{log}}=I_K$ one has $\Phi(V)\in\mathcal{G}_{2L}(K)$ exactly. So we see that a wider block can realize more set elements as physical matchings, enlarging the code-specific reach, but never an element outside the set already specified by $2$-local gates. Strict enlargement beyond the set occurs only at $W=n$.

Permutation automorphisms of $Q$ provide another source of logical actions that is disjoint in character from the 2-Local logical reach of Proposition~\ref{prop:8}: a permutation cannot apply $\mathrm{\overline{S}}$, $\mathrm{\overline{CZ}}$, or any entangling action in a fixed basis, but can realize basis-changing $\GL(K,\mathbb{F}_2)$ elements on the logical block that jointly can realize entangling action.

\begin{definition}[Permutation automorphism]
For a CSS code $Q = \mathrm{CSS}(H_X, H_Z)$, a permutation $\pi \in S_n$ is a \emph{permutation automorphism} of $Q$ if
\begin{align*}
\rowspan(H_X)\, P_\pi \;&=\; \rowspan(H_X),\\
\rowspan(H_Z)\, P_\pi \;&=\; \rowspan(H_Z),
\end{align*}
where $P_\pi$ is the $n\times n$ qubit-permutation matrix with entries $(P_\pi)_{q,\pi(q)}=1$, acting on the right by permuting physical-coordinate columns, consistent with the logical action $L_X P_\pi$ of Proposition~\ref{prop:pi-logical}. The group of such permutations is denoted $\Aut^{\mathrm{phys}}_{\mathrm{perm}}(Q)$; we abbreviate $\Aut_{\mathrm{perm}}(Q):=\Aut^{\mathrm{phys}}_{\mathrm{perm}}(Q)$ and write $\rho(\Aut_{\mathrm{perm}}(Q))$ for its logical image (Proposition~\ref{prop:pi-logical}).
\end{definition}

\begin{proposition}[Logical action of $\pi$,~\cite{guyot2025_addressability}]
\label{prop:pi-logical}
For every $\pi \in \Aut^{\mathrm{phys}}_{\mathrm{perm}}(Q)$ there is a unique $V_\pi \in \GL(K, \mathbb{F}_2)$ such that
$$
L_X P_\pi \;\equiv\; V_\pi L_X \pmod{\rowspan(H_X)}.
$$
The map $\pi \mapsto V_\pi$ is a group homomorphism, and its image
$$
\Aut^{\mathrm{log}}_{\mathrm{perm}}(Q) \;:=\; \rho\big(\Aut^{\mathrm{phys}}_{\mathrm{perm}}(Q)\big) \;=\; \{V_\pi : \pi \in \Aut^{\mathrm{phys}}_{\mathrm{perm}}(Q)\} \;\subseteq\; \GL(K, \mathbb{F}_2)
$$
embeds into the symplectic group via
$$
\Phi(\pi) \;=\; \begin{pmatrix} V_\pi & 0 \\ 0 & V_\pi^{-\top} \end{pmatrix} \;\in\; \Sp(2K, \mathbb{F}_2).
$$
\end{proposition}

\begin{lemma}[Generation Lemma]
\label{lem:generation}
For every $K \ge 1$, the multi-layer closure of the $2$-local set is the full logical Clifford group:
$$
\bigl\langle \mathcal{G}_{2L}(K) \bigr\rangle \;=\; \Sp(2K, \mathbb{F}_2).
$$
\end{lemma}
This is an \emph{algebraic} identity in $\Sp(2K,\ftwo)$, independent of any particular code; physical attainability on a given $Q$ is governed by Corollary~\ref{cor:phys-generation}. Constructive generators (the $\mathrm{\overline{H}}_i$ gadget has depth $3$): $\mathrm{\overline{H}}_i = \mathrm{\overline{S}}_i \cdot \mathrm{\overline{SHS}}_i \cdot \mathrm{\overline{S}}_i$ and $\mathrm{\overline{CNOT}}_{ij} = (\mathbf{I}_i \otimes \mathrm{\overline{H}}_j)\, \mathrm{\overline{CZ}}_{ij}\, (\mathbf{I}_i \otimes \mathrm{\overline{H}}_j)$.

\begin{proof}
The identity $\mathrm{\overline{H}} = \mathrm{\overline{S}}\cdot\mathrm{\overline{SHS}}\cdot\mathrm{\overline{S}}$ on a single logical qubit is verified directly in $\Sp(2, \mathbb{F}_2) \cong S_3$ (it expresses the Bruhat decomposition of $\mathrm{\overline{H}}$ across the cell outside the set). Each factor $\mathrm{\overline{S}}_i, \mathrm{\overline{SHS}}_i, \mathrm{\overline{CZ}}_{ij}$ is a set element by Proposition~\ref{prop:8}. The Aaronson--Gottesman canonical form for stabilizer/symplectic operations~\cite{aaronson2004improved} expresses any $M \in \Sp(2K, \mathbb{F}_2)$ as a bounded sequence of $\mathrm{\overline{H}}$, phase ($\mathrm{\overline{S}}$), and $\mathrm{\overline{CNOT}}$ layers, so $\{\mathrm{\overline{S}}_i, \mathrm{\overline{SHS}}_i, \mathrm{\overline{CZ}}_{ij}\}_{i\ne j}$ generates $\Sp(2K, \mathbb{F}_2)$; depth-$3$ access to $\mathrm{\overline{H}}$ then upgrades any set-generated word to a $\Sp(2K, \mathbb{F}_2)$ word at cost $O(K^2)$ depth.
\end{proof}

\begin{corollary}[Physical generation from a realized library]
\label{cor:phys-generation}
Lemma~\ref{lem:generation} is algebraic. Its physical counterpart is conditional: if an indecomposable CSS code $Q$ realizes every set generator $\mathrm{\overline{S}}_i$, $\mathrm{\overline{SHS}}_i$, and $\mathrm{\overline{CZ}}_{ij}$ ($i\ne j\in[K]$) by at most $D_0$ codespace-preserving $W\le 2$ layers so that $\mathrm{\overline{H}}_i=\mathrm{\overline{S}}_i\,\mathrm{\overline{SHS}}_i\,\mathrm{\overline{S}}_i$ is available at depth $\le 3D_0$, then these layers generate the full projective logical Clifford group $\Sp(2K,\ftwo)$. If only a subset of the set generators is realized as CSP layers, only the subgroup they generate is physically attained. The Quantum Logic Codes of Section~\ref{sec:qlc} are codes that realize a complete such library.
\end{corollary}

\section{Novel Constructions}
\label{sec:constructions}

In this section we detail several new results coming from this theory. First, we show a novel depth-1 2-local transversal $\mathrm{\overline{S}}$ gate in the rotated surface code, only operating on data qubits, that is distance-preserving for every odd distance (a code automorphism) and is fault-tolerant for $d\ge5$, generalizing to all odd distances. The single-layer construction is a \emph{diagonal} circuit composed of products of $S$ and $CZ$ gates that fixes every $Z$ operator while performing the phase gate logically. Next, we construct a depth-1 2-local intra-block $\mathrm{\overline{CZ}}$ gate in the 2D-toric code which generalizes to all code block sizes $L > 2$. 
Then, we show a code family that possesses a full logical Clifford instruction set using only codespace-preserving 2-local transversal gates and permutation automorphisms, with parameters scaling as $[[\,r n_0 7^\ell,\,2r,\,d_0 3^\ell\,]]$ for self-dual cores $(n_0,d_0)\in\{(4,2),(18,5),(20,6)\}$. We call these \emph{Quantum Logic Codes}.

\subsection{Rotated Surface Code Transversal Phase Gate}
\label{sec:rotated-S}

The first construction is a single-layer $2$-local transversal $\mathrm{\overline{S}}$ gate on the rotated surface code, defined for every odd distance $d \ge 3$ (Fig.~\ref{fig:surface_d5_d7}). This can be viewed in contrast to the construction in \cite{chen2026transversal}, wherein the mid-cycle state \cite{mcewen2023relaxing} of the surface code is used to land the logical $\mathrm{\overline{S}}$ gate, leveraging the joint state formed with ancillary qubits during syndrome extraction. This construction only uses data qubits. Qubits of the $d\times d$ rotated surface code are indexed by $q = rd + c$ for $(r, c) \in \{0,1,\dots,d-1\}^2$ with row $r$ and column $c$ (rows and columns $0$-indexed).
\begin{figure*}[h!]
\centering
\includegraphics[width=0.98\linewidth,keepaspectratio]{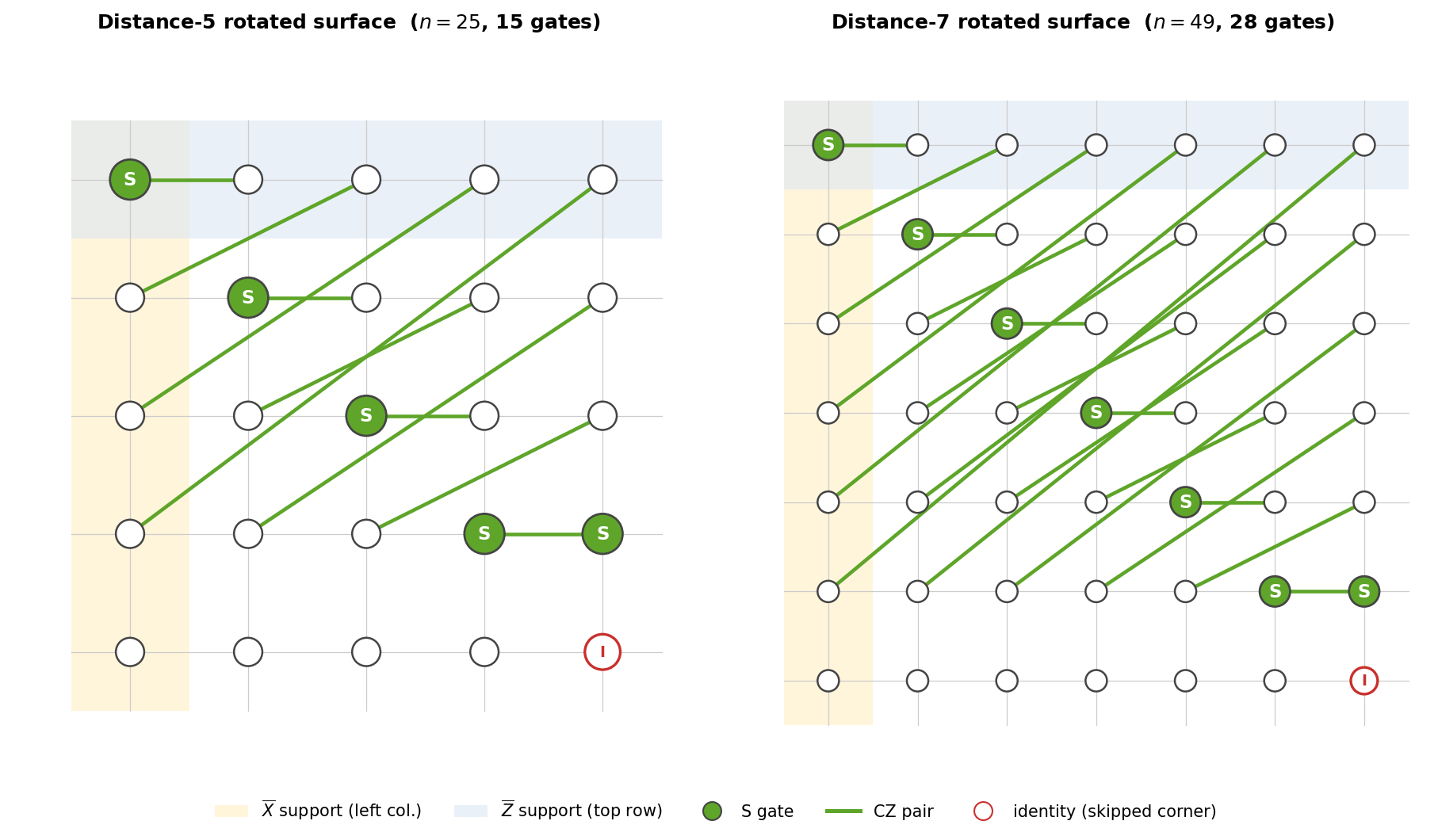}
\caption{Construction~\ref{con:rotated-S} at $d=5$ (left) and $d=7$ (right): the same closed-form $d$ diagonal phase gates and $\binom{d}{2}$ rotated-diagonal $CZ$ pairs scale to all odd distances.}
\label{fig:surface_d5_d7}
\end{figure*}

\begin{construction}[Rotated-pair $\mathrm{\overline{S}}$ on the $d\times d$ rotated surface code]
\label{con:rotated-S}
Take the standard logical representatives $\overline X = \{(r,0) : 0 \le r < d\}$ (the left column) and $\overline Z = \{(0,c) : 0 \le c < d\}$ (the top row), which meet only at the corner $(0,0)$. For every odd $d \ge 3$, apply
\begin{enumerate}[label=(\roman*)]
\item $d$ phase gates $S_q$ on the diagonal qubits $(r,r)$ for $0 \le r \le d-2$ together with $(d-2,\,d-1)$;
\item $\binom{d}{2}$ controlled-$Z$ gates on the rotated pairs $\{(c,r),\,(r,c+1)\}$ for $0 \le r \le c \le d-2$, i.e.\ each on/below-diagonal qubit $(c,r)$ is paired with its reflection across the main diagonal shifted one column right.
\end{enumerate}
Total gate count $d(d+1)/2$ ($d$ phase gates plus $\binom{d}{2}$ $CZ$ pairs). The $\binom{d}{2}$ rotated pairs are pairwise disjoint, so the off-diagonal support is a matching and the layer is $2$-local.
\end{construction}

\begin{theorem}[Closed-form rotated-surface $\mathrm{\overline{S}}$]
\label{thm:rotated-S}
For every odd $d \ge 3$, Construction~\ref{con:rotated-S} realizes the logical operator $\mathrm{\overline{S}}$ exactly, as a single targeted $2$-local Clifford layer on the $d \times d$ rotated surface code. The layer is CSS-preserving for every odd $d$ and it preserves the code distance. As a depth-1 gadget its circuit-fault distance is at most $d-1$ and at least $\lceil d/2\rceil$ (exactly $d-1$ for $d=3,5$), making it single-fault-tolerant for $d\ge5$.
\end{theorem}

\subsection{2D-Toric Code Transversal Intra-block Controlled-Z Gate}
\label{sec:toric-CZ}

The second construction is a single-layer $2$-local transversal intra-block $\mathrm{\overline{CZ}}_{0,1}$ on the $L\times L$ 2D toric code, defined for every $L \ge 3$ (Fig.~\ref{fig:toric_L3_L5}). Qubits sit on edges of the torus, indexed by $(r, c, t) \in [L]\times [L]\times \{0,1\}$ with $t = 0$ horizontal and $t = 1$ vertical; total $n = 2L^2$. Standard logical operators are $\overline{X}_0 = \prod_r h(r, 0)$, $\overline{X}_1 = \prod_c v(0, c)$, and their $Z$-duals.
\begin{figure*}[h!]
\centering
\includegraphics[width=0.98\linewidth,keepaspectratio]{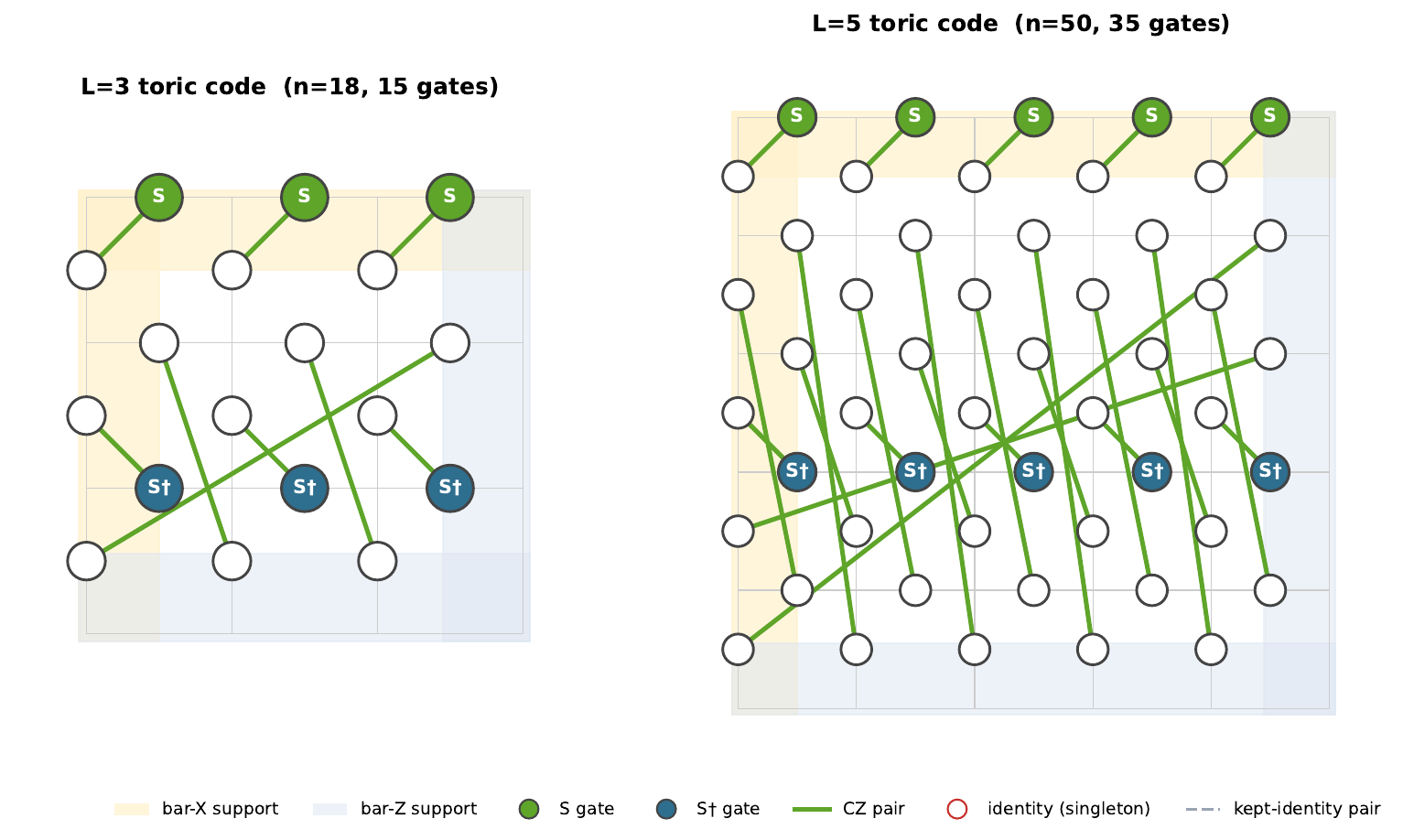}
\caption{Construction~\ref{con:toric-CZ} at $L = 3$ (left) and $L = 5$ (right). Both layers have row-$0$ $CZ$ pairs $h(0,c)\leftrightarrow v(0,c)$ combined with row-$0$ $S$ gates on the $h(0,c)$ qubits, two cross-pair families per tier $r = 1, \ldots, \lfloor(L-1)/2\rfloor$ joining rows $a = r$ and $b = L-r$, and inner-row $S^{\dagger}$ gates on row $s = \lceil L/2\rceil$. Total weights $w(3) = 15$ and $w(5) = 35$; the long diagonals visible in the $L = 5$ panel correspond to cross pairs wrapping across the torus boundary.}
\label{fig:toric_L3_L5}
\end{figure*}

\begin{construction}[Closed-form intra-block $\mathrm{\overline{CZ}}_{0,1}$ on the $L \times L$ toric code, $L\ge3$]
\label{con:toric-CZ}
\begin{enumerate}[label=(\roman*)]
\item Row $0$: for every $c$, pair $h(0, c) \leftrightarrow v(0, c)$ with a $CZ$, and apply $S_q$ on every $h(0, c)$.
\item For each $r = 1, \ldots, \lfloor (L-1)/2 \rfloor$, pair rows $a = r$ and $b = L - r$ by two cross-pair families: a skewed family $h(a, c) \leftrightarrow v(b, (c+1) \bmod L)$ and an aligned family $v(a, c) \leftrightarrow h(b, c)$, for $c \in [L]$.
\item For even $L$, on the middle row apply a self-pair $h(L/2, c) \leftrightarrow v(L/2, (c+1) \bmod L)$.
\item On the inner row $s = \lceil L/2 \rceil$, apply $S^{\dagger}_q$ on every $h(s, c)$ ($S^{\dagger} = ZS$, the same symplectic action as $S$).
\end{enumerate}
Total gate count $w(L) = L^2 + 2L$.
\end{construction}

\begin{theorem}[Closed-form toric $\mathrm{\overline{CZ}}_{0,1}$]
\label{thm:toric-CZ}
For every $L \ge 3$, Construction~\ref{con:toric-CZ} realizes the logical operator $\mathrm{\overline{CZ}}_{0,1}$ exactly, as a single targeted $2$-local Clifford layer on the $L \times L$ 2D toric code.
\end{theorem}

\subsection{Quantum Logic Codes}
\label{sec:qlc}

The third construction is a constructive family of CSS codes, the \emph{Quantum Logic Codes}, each of which carries a \emph{complete} transversal logical-Clifford basis instruction set: a fixed library of codespace-preserving (CSP) $2$-local circuits together with permutation automorphisms whose logical symplectic images generate the entire logical Clifford group $\Sp(2k, \mathbb{F}_2)$. The family is generated from three small self-dual \emph{core} codes by two composition operations (Fig.~\ref{fig:qlc-construction}): block tiling, which grows $k$, and concatenation with inner self-dual doubly-even codes which we choose to be the $[[7,1,3]]$ code following \cite{jochymoconnor2014using}, which grows $d$. Each of these composition operations provably preserves the complete ISA, and in fact preserves the depth of the transversal gates that form the logical basis set. For two of the three core codes, the complete transversal logical Clifford basis ISA is depth-one, which is carried through the construction composition operations and remains depth-one up to addressable operations between tiled code cores, which become depth-two. Combining the two yields the code family of parameters
\begin{equation}
\label{eq:qlc-family}
[[n_0r7^{\ell}, \;\; 2r, \;\; d_0\,3^{\ell} \,]],
\qquad (n_0, d_0) \in \{(4,2),\,(18,5),\,(20,6)\},
\quad r \ge 1,\ \ell \ge 0,
\end{equation}
with transversally achievable logical Clifford group $\Sp(4r, \mathbb{F}_2)$ and rate $2/(n_0\,7^{\ell})$. Here $r$ counts the tiled core copies and $\ell$ the concatenation levels; the base data $(n_0,d_0)$ selects the core. We can fix parameters to fix subfamilies of this code construction. For example, set $r = 7^{\ell}$, connecting the number of tiled code cores with the concatenation level, and focus on the $[[4,2,2]]$ core code. This then produces the following asymptotic parameters:
\begin{equation}
\label{eq:qlc-family-fixed-tiling}
[[4 \cdot 49^{\ell}, \;\; 2 \cdot 7^{\ell}, \;\; 2\cdot3^{\ell} \,]] = [[n, \sqrt{n}, \Theta(n^{0.2823})]].
\end{equation}

We first define the instruction set, then prove the composition theorem that generates the family, tabulate its members, and finally exhibit the complete logical Clifford basis ISA of the three cores and of their first concatenation level with explicit circuits.

\begin{figure*}[t]
\centering
\includegraphics[width=0.96\linewidth,keepaspectratio]{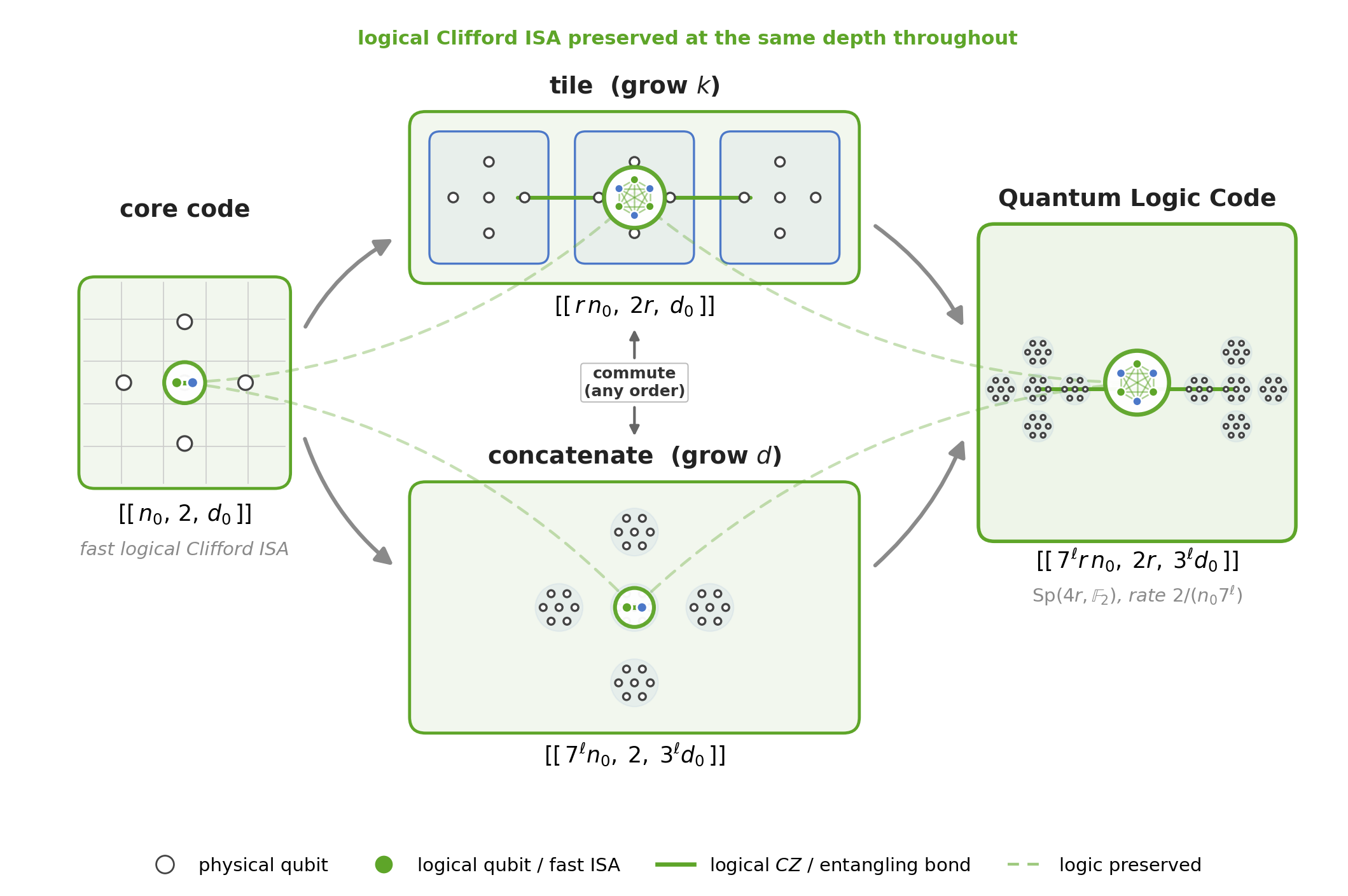}
\caption{Construction intuition for the Quantum Logic Codes (schematic; the literal transversal gate layers are not shown). A small \emph{core} code (left) carries a complete logical-Clifford basis ISA at low circuit depth. Two composition operations grow the code parameters while preserving that ISA: \emph{tiling} ($\times r$, top) places $r$ copies side by side, raising the logical-qubit count to $2r$ at fixed distance and supplying inter-copy addressable $\mathrm{\overline{CZ}}$; \emph{concatenation} with the Steane $[[7,1,3]]$ code ($\times\ell$, bottom) replaces each physical qubit by an inner block, tripling the distance per level at fixed $k$. Combining the two yields the family $[[\,7^{\ell}r\,n_0,\,2r,\,3^{\ell}d_0\,]]$ with logical Clifford group $\Sp(4r,\mathbb{F}_2)$ and rate $2/(n_0 7^{\ell})$ (right). The fixed emblem (a bonded logical pair) marks the logical ISA: it is invariant under both operations, and so is the circuit depth that generates it (Theorem~\ref{thm:qlc-compose}).}
\label{fig:qlc-construction}
\end{figure*}

Each core is a single-block self-dual group-algebra code $\CSS(C,C)$ whose check $H_X=H_Z=H$ is a full-rank row basis of the rowspan of the $|G|\times|G|$ regular-representation matrix $M(c)$ (entry conventions given in Constructions~\ref{con:qlc-2026} and~\ref{con:qlc-1825}) of an element $c \in \mathbb{F}_2[G]$ that is self-orthogonal ($cc^\ast = 0$, so $H_X H_Z^\top = 0$) and doubly-even; self-duality ($S_X = S_Z$) is automatic from $H_X=H_Z$. The smallest, $[[4,2,2]]$, has $G = \mathbb{Z}_4$ and $c = 1+x+x^2+x^3$; its complete depth-1 ISA is given explicitly in Construction~\ref{con:qlc-422} and Fig.~\ref{fig:qlc-c422}.

\subsubsection{The complete logical Clifford basis ISA}
\label{sec:qlc-isa-def}

\begin{definition}[Complete transversal logical Clifford basis ISA]
\label{def:qlc-isa}
Let $Q = \CSS(H_X, H_Z)$ encode $k$ logical qubits with canonical representatives $(L_X, L_Z)$, $L_X L_Z^\top = \mathbf{I}_k$, and stabilizer rowspans $S_X, S_Z$. A \emph{complete transversal logical Clifford basis ISA} for $Q$ is a finite library of circuits whose logical symplectic images generate $\Sp(2k, \mathbb{F}_2)$, in which every layer is one of:
\begin{enumerate}[label=(\roman*)]
\item a \emph{$2$-local codespace-preserving layer}: a tensor product of admissible per-block factors $g_p \in \mathcal{G}_{\mathrm{diag}}$ (Prop.~\ref{prop:7b}) on disjoint qubit pairs;
\item a \emph{permutation automorphism} $\pi \in S_n$ with $\pi(S_X) = S_X$ and $\pi(S_Z) = S_Z$, contributing its logical action $\mathcal{A} \in \Aut_{\mathrm{perm}}^{\mathrm{log}}(Q)$;
\item the global transversal Hadamard $\mathrm{\overline{H}} = H^{\otimes n}$.
\end{enumerate}
We fix the generating set $\{\mathrm{\overline{S}}_i,\, \mathrm{\overline{SHS}}_i,\, \mathrm{\overline{CZ}}_{i,j}\}_{i\ne j\in[k]}$ together with the logical $\mathrm{\overline{H}}$ and the logical-qubit permutations. The ISA is depth-1 if every $\mathrm{\overline{S}}_i$, $\mathrm{\overline{SHS}}_i$, and $\mathrm{\overline{CZ}}_{i,j}$ is realized by a \emph{single} layer of type (i), and more generally bounded-depth $D_0$ if every diagonal/$CZ$ generator $\mathrm{\overline{S}}_i,\mathrm{\overline{SHS}}_i,\mathrm{\overline{CZ}}_{i,j}$ is realized by at most $D_0$ layers of type (i); the code is \emph{self-dual} when $S_X = S_Z$.
\end{definition}

\subsubsection{Bounded-depth self-dual composition}
\label{sec:qlc-compose}

The family of Eq.~\eqref{eq:qlc-family} is generated by the following theorem, which gives sufficient conditions under which the two composition operations preserve a complete ISA.

\begin{theorem}[Bounded-depth self-dual composition]
\label{thm:qlc-compose}
Let $Q=[[n,k,d]]$ carry a complete transversal logical Clifford basis ISA (Def.~\ref{def:qlc-isa}) that is \emph{bounded-depth ($D_0$)} and \emph{self-dual} ($S_X=S_Z$). The operations
\begin{enumerate}[label=(\roman*)]
\item \textbf{tiling}: $r$ disjoint copies $Q^{\oplus r}$ equipped with the permutation automorphisms and the inter-block transversal $CZ$, and
\item \textbf{concatenation}: $Q$ as outer code with the inner Steane code $[[7,1,3]]$ or any self-dual doubly-even $[[n_i,1,d_i]]$ admitting transversal $S$, $SHS$, $CZ$, $H$, in which case the parameter map below becomes $[[n,k,d]]\mapsto[[n_i n,k,d_i d]]$,
\end{enumerate}
each preserve the completeness of the ISA. Concatenation preserves every per-generator layer count exactly; tiling preserves the intra-copy layer counts and supplies the inter-copy generators at constant depth (Proposition~\ref{prop:crosscopy}), so the per-generator depth of the complete ISA is bounded by a constant independent of $r$ and $\ell$, equal to $D_0$ for the intra-copy generators. Tiling gives $[[n,k,d]]\mapsto[[rn, rk, d]]$ with logical group $\Sp(2rk, \mathbb{F}_2)$, and concatenation gives $[[n,k,d]]\mapsto[[7n, k, 3d]]$. The two operations commute. 
\end{theorem}

\begin{proof}
\emph{Concatenation (sufficiency).} Identify outer qubit $p$ with the inner block $B_p = \{n_i p,\dots,n_i p + n_i - 1\}$ where $n_i$ is the number of qubits of the inner code. The inner code is assumed to be self-dual and doubly-even, and admits transversal $S$, $SHS$, $CZ$ and $H$, each a depth-1 CSP operation realizing the corresponding inner logical. Lift a single outer type-(i) layer generator-wise: replace each outer $S_p$ (resp.\ $SHS_p$) by inner transversal $S$ (resp.\ $SHS$) on $B_p$, and each outer $CZ_{pq}$ (resp.\ $CZ^{\sharp}_{pq}$) by the inner transversal $CZ$ (resp.\ $CZ^{\sharp}$) between $B_p$ and $B_q$, i.e.\ the $n_i$ physical $CZ$ matching qubit $j$ of $B_p$ to qubit $j$ of $B_q$. Because the outer layer is a matching, its touched qubits are disjoint, so the touched inner blocks are disjoint and the lifted gates are again a matching with $W\le 2$. The lifted layer is codespace-preserving: the inner transversal gates fix every inner-block stabilizer, and they map the lifted outer stabilizers within $S_X\oplus S_Z$ precisely by the outer codespace-preserving property $C\,S_X\subseteq S_Z$, $B\,S_Z\subseteq S_X$ of the outer physical $C$- and $B$-matrices. The encoded logical action equals the outer one, since each inner transversal gate realizes its inner logical up to a logical Pauli (e.g.\ transversal $S$ on the Steane code realizes $\mathrm{\overline{S}}^\dagger = \overline{Z}\,\mathrm{\overline{S}}$, the same symplectic element). Exact stabilizer signs are preserved as well: every inner logical $\overline{X}$ representative has odd weight, fixed mod $4$ across representatives (inner stabilizers are doubly-even and meet each representative evenly), so each outer $S$ gate contributes its phase of Lemma~\ref{lem:diag-action}(i) with multiplicity of this odd weight, and each outer $CZ$ pair with multiplicity of the overlap of the two blocks' inner representatives, odd since the representatives have odd weight and meet inner stabilizers evenly; since a valid outer layer applies an even number of $S$ gates to every stabilizer, the lifted phase equals the outer phase, and inner stabilizers acquire phase $+1$ because they are doubly-even and single-block. Thus each outer layer lifts to exactly one inner layer of equal width; applying this to each of the $D_g$ layers of a generator realized in $D_g$ layers preserves its layer count $D_g$ exactly, so the per-generator depth bound $D_0=\max_g D_g$ is invariant. Completeness carries over, and the distance multiplies exactly: $d(Q\circ\mathrm{inner}) \ge d\times d_{\mathrm{inner}}$ by the standard concatenation bound, while replacing each qubit of a minimum-weight outer logical representative by a minimum-weight inner representative (self-duality supplies one of each type at weight $d_{\mathrm{inner}}$) exhibits a nontrivial logical operator of weight $d\times d_{\mathrm{inner}}$, so $d(Q\circ\mathrm{inner}) = d\times d_{\mathrm{inner}}$.

\emph{Tiling (sufficiency).} As a code, the tiling is the $r$-fold direct sum $Q^{\oplus r}$: its checks are block-diagonal ($H_X^{\oplus r},\,H_Z^{\oplus r}$), it carries no cross-code stabilizers, and its codespace is the tensor product of the $r$ copies. The tiling adds two transversal operations, both enabled by the multiplicity of identical copies. These are the permutation automorphisms $S_r$ and the addressable inter-code transversal $CZ$. The per-copy ISA realizes $\Sp(2k)^{\times r}$. Each permutation $\tau\in S_r$ permutes identical stabilizer blocks and is a code automorphism. Its logical action permutes the $r$ logical blocks accordingly. The full inter-block transversal $CZ$ between copies $a,b$ is the qubit-wise $CZ$ across their identical supports, a $W{=}2$ matching whose logical image is the off-diagonal coupling $\bC = \bigl(\begin{smallmatrix} 0 & P \\ P^\top & 0\end{smallmatrix}\bigr)$ with $P = L_X(L_X)^\top$ so that qubit $j$ of copy $a$ paired with qubit $j$ of copy $b$. By \eqref{eq:2local-css-pres-x} the coupling is codespace-preserving, with exact phase $+1$ on every stabilizer (each tiled stabilizer is supported on a single copy, so no aggregate $CZ$ pair lies inside its support, and the layer has no $S$ gates), and self-duality ($S_X=S_Z$) makes the term $P=L_XL_X^\top$ invertible.
Together these operations generate $\Sp(2rk, \mathbb{F}_2)$ (Proposition~\ref{prop:crosscopy} together with the Generation Lemma~\ref{lem:generation}). We show in Proposition \ref{prop:qlc-depth} that tiling also provides the addressable inter-block $CZ$ gate at constant depth.

\emph{Commutativity.} Both operations require and preserve the bounded-depth and self-dual properties and act on the copy index versus the inner block, so $\mathrm{tile}_r\circ\mathrm{concat}_\ell = \mathrm{concat}_\ell\circ\mathrm{tile}_r$, giving the two-parameter lattice of Eq.~\eqref{eq:qlc-family}.
\end{proof}

\subsubsection{The combined family}
\label{sec:qlc-family}

Theorem~\ref{thm:qlc-compose} generates Eq.~\eqref{eq:qlc-family} from each core: $\ell$ $[[7,1,3]]$ levels multiply $(n,d)$ by $(7,3)$ per level at fixed $k$ and fixed depth, while $r$-fold tiling multiplies $(n,k)$ by $r$ at fixed $d$, so $k = 2r$ and $d = d_0 3^{\ell}$ are independently tunable and the rate $2/(n_0 7^{\ell})$ depends only on the core and the concatenation depth. Table~\ref{tab:qlc-family} lists representative members; every entry carries the complete logical Clifford basis ISA $\Sp(4r, \mathbb{F}_2)$.

\begin{table}[ht]
\centering
\renewcommand{\arraystretch}{1.5}
\setlength{\tabcolsep}{6pt}
\begin{tabular}{c c c c c c}
\hline\hline
core $(n_0,d_0)$ & group $G$ & seed $[[n_0,2,d_0]]$ & tiled $[[r n_0, 2r, d_0]]$ & $\circ$ $[[7,1,3]]$ $[[7n_0,2,3d_0]]$ & tiled $\circ$ $[[7,1,3]]$\\
\hline
$(4,2)$  & $\mathbb{Z}_4$            & $[[4,2,2]]$  & $[[4r,2r,2]]$  & $[[28,2,6]]$   & $[[28r,2r,6]]$ \\
$(18,5)$ & $D_9$                     & $[[18,2,5]]$ & $[[18r,2r,5]]$ & $[[126,2,15]]$ & $[[126r,2r,15]]$ \\
$(20,6)$ & $\mathrm{AGL}(1,5)$       & $[[20,2,6]]$ & $[[20r,2r,6]]$ & $[[140,2,18]]$ & $[[140r,2r,18]]$ \\
\hline\hline
\end{tabular}
\caption{Members of the Quantum Logic Code family $[[n_0 r 7^{\ell},2r,d_03^{\ell}]]$ at $\ell\in\{0,1\}$. Concatenated distances are exact (Theorem~\ref{thm:qlc-compose}). All carry the complete logical Clifford basis ISA and can achieve $\Sp(4r,\mathbb{F}_2)$ transversally. They have rate $=2/(n_0 7^{\ell})$.}
\label{tab:qlc-family}
\end{table}
\subsubsection{First concatenation level}
\label{sec:qlc-concat}

By Theorem~\ref{thm:qlc-compose} the first $[[7,1,3]]$ level of each core, $[[28,2,6]] = [[4,2,2]]\circ [[7,1,3]]$, $[[126,2,15]] = [[18,2,5]]\circ [[7,1,3]]$, and $[[140,2,18]] = [[20,2,6]]\circ [[7,1,3]]$ carries the complete logical-Clifford basis ISA at the \emph{same} per-generator layer count as its core: each outer $W\le2$ layer lifts to exactly one inner $W\le2$ layer, so the distance triples while depth is unchanged. Table~\ref{tab:qlc-depths} records the depth-optimized layer counts. The depth-1 cores stay depth-1, and the $D_9$ profile $\{1,2\}$ carries to $[[126,2,15]]$. Figure~\ref{fig:qlc-c2826} shows the lifted $\mathrm{\overline{S}}_0$ on $[[28,2,6]]$ as a single depth-1 CSP layer: $[[7,1,3]]$ blocks $0$ and $2$ carry transversal $S$, joined by seven transversal $CZ$.

\begin{table}[ht]
\centering
\renewcommand{\arraystretch}{1.4}
\setlength{\tabcolsep}{9pt}
\begin{tabular}{l c c c}
\hline\hline
generator & $[[4,2,2]]\!\to\![[28,2,6]]$ & $[[18,2,5]]\!\to\![[126,2,15]]$ & $[[20,2,6]]\!\to\![[140,2,18]]$ \\
\hline
$\mathrm{\overline{S}}_0$    & $1\to1$ & $2\to2$ & $1\to1$ \\
$\mathrm{\overline{S}}_1$    & $1\to1$ & $2\to2$ & $1\to1$ \\
$\mathrm{\overline{CZ}}_{0,1}$ & $1\to1$ & $1\to1$ & $1\to1$ \\
$\mathrm{\overline{SHS}}_0$  & $1\to1$ & $2\to2$ & $1\to1$ \\
$\mathrm{\overline{SHS}}_1$  & $1\to1$ & $2\to2$ & $1\to1$ \\
$\mathrm{\overline{H}}$  & $2\to2$ & $1\to1$ & $2\to2$ \\
$\mathrm{\overline{SWAP}}$   & $1\to1$ & $1\to1$ & $1\to1$ \\
\hline
\multicolumn{4}{l}{\emph{$r{=}2$ tiling (on $[[2n_0,4,d_0]]$):}}\\
inter-copy aggregate $\mathrm{\overline{CZ}}_{a,b}$ & $1\to1$ & $1\to1$ & $1\to1$ \\
addressable inter-copy $\mathrm{\overline{CZ}}_{i^{(a)},j^{(b)}}$ & $2\to2$ & $\le 18\to\le 18$ & $2\to2$ \\
block-shift $\overline{\sigma}$ & $1\to1$ & $1\to1$ & $1\to1$ \\
\hline\hline
\end{tabular}
\caption{Depth-optimized layer counts (core $\to$ first concatenation level). Concatenation is layer-by-layer, so every per-generator depth is preserved while $(n,d)$ scales by $(7,3)$. Tiling ($r$-fold) preserves these intra-copy depths and adds the depth-$1$ aggregate inter-copy $\mathrm{\overline{CZ}}_{a,b}$ which couples all logicals of the two copies, and the copy-permutation automorphisms (every $\tau\in S_r$ is itself a depth-$1$ automorphism). The individually \emph{addressable} inter-copy $\mathrm{\overline{CZ}}_{i^{(a)},j^{(b)}}$ is realized at small constant depth (Prop.~\ref{prop:crosscopy}): depth $2$ for the depth-1 core families, and $\le 18$ layers for the $D_9$ core family.}
\label{tab:qlc-depths}
\end{table}

The two composition axes leave the per-generator circuit depth flat, and only the
number of generators grows. We
record this as a depth bound in $r$ and $\ell$.

\begin{proposition}[Depth bound for the Quantum Logic Code family]
\label{prop:qlc-depth}
Let $Q_0=[[n_0,2,d_0]]$ be a bounded-depth self-dual core whose diagonal/$CZ$ generators have maximum layer count $D_0$ (so $D_0=1$ for $\mathbb{Z}_4$ and $\mathrm{AGL}(1,5)$, and $D_0=2$ for $D_9$). Then every family member $Q_{r,\ell}=[[n_0 r 7^{\ell},2r,d_0 3^{\ell}]]$ of Theorem~\ref{thm:qlc-compose} carries a complete logical Clifford basis ISA in which
\begin{enumerate}[label=(\roman*)]
\item every intra-copy diagonal generator $\mathrm{\overline{S}}_i,\mathrm{\overline{SHS}}_i$ and intra-copy $\mathrm{\overline{CZ}}_{ij}$ is realized by at most $D_0$ codespace-preserving $W\le 2$ layers; the inter-copy coupling between two copies is realized by a single qubit-wise transversal $CZ$ of Theorem~\ref{thm:qlc-compose}, whose logical image is the \emph{aggregate} $\bigl(\begin{smallmatrix}0&P\\P^\top&0\end{smallmatrix}\bigr)$ with $P=L_XL_X^\top$ invertible, and each \emph{individual addressable} cross-copy generator $\mathrm{\overline{CZ}}_{i^{(a)},j^{(b)}}$ is realized at \emph{constant} depth: depth $2$ for the $\mathbb{Z}_4$ and $\mathrm{AGL}(1,5)$ cores and a constant $\le 18$ for $D_9$, synthesized from the diagonal generators, by Proposition~\ref{prop:crosscopy}, in all cases independent of $r$ and $\ell$; the global $\mathrm{\overline{H}}$ and the permutation generators (every copy permutation $\tau\in S_r$, and the in-core swap, is itself a depth-$1$ automorphism) have fixed depth $\le 2$, independent of both $r$ and $\ell$, while an arbitrary permutation of the $2r$ logical qubits requires $O(r^{2})$ such generators; and
\item an arbitrary logical Clifford in $\Sp(4r,\mathbb{F}_2)$ is realized by $O(r^{2})\,D_0=O(r^{2})$ such layers, again independent of $\ell$.
\end{enumerate}
\end{proposition}
\begin{proof}
\emph{(i)} Concatenation is layer-by-layer (Theorem~\ref{thm:qlc-compose}): each outer $W\le2$ layer lifts to exactly one inner $W\le2$ layer, so every generator's layer count is unchanged under $\ell\mapsto\ell+1$. The depth is independent of $\ell$ while the distance triples. Tiling replicates the per-copy ISA $\Sp(4)^{\times r}$ at the unchanged core depths and introduces exactly two new generator types: the inter-copy $\mathrm{\overline{CZ}}_{a,b}$, a single $W{=}2$ matching of transversal $CZ$ across the identical supports of copies $a,b$ (depth $1$), and the copy-permutation automorphisms (each $\tau\in S_r$ a qubit permutation, depth $1$). Hence every diagonal/$CZ$ generator of $Q_{r,\ell}$ has layer count $\le D_0$, and the permutation generators $\le 2$, independent of $r$ and $\ell$.

\emph{(ii)} By the Generation Lemma~\ref{lem:generation} and the Aaronson--Gottesman normal form \cite{aaronson2004improved}, any $M\in\Sp(4r,\mathbb{F}_2)$ factors into $O((2r)^2)=O(r^2)$ generators from $\{\mathrm{\overline{S}}_i,\mathrm{\overline{SHS}}_i,\mathrm{\overline{CZ}}_{i,j},\mathrm{\overline{H}}_i\}$; substituting the realizations of (i) (each diagonal/$CZ$ generator at $\le D_0$ layers and $\mathrm{\overline{H}}_i=\mathrm{\overline{S}}_i\,\mathrm{\overline{SHS}}_i\,\mathrm{\overline{S}}_i$ at $\le 3D_0$ layers) yields a codespace-preserving circuit of depth $O(r^2)\,D_0$. The bound carries no dependence on $\ell$, so arbitrarily large distance $d_0 3^{\ell}$ is obtained at fixed depth.
\end{proof}

\begin{proposition}[Addressable inter-copy entanglers are constant-depth]
\label{prop:crosscopy}
Let $Q_0=[[n_0,k,d_0]]$ be a self-dual core ($S_X=S_Z$, $k\ge2$) carrying a complete bounded-depth-$D_0$ transversal logical Clifford basis ISA (Def.~\ref{def:qlc-isa}). On any tiled family member $Q_{r,\ell}$ of Theorem~\ref{thm:qlc-compose}, for every pair of copies $a\ne b$ and all $i,j\in[k]$, the individually addressable $\mathrm{\overline{CZ}}_{i^{(a)},j^{(b)}}$ is realized by a codespace-preserving circuit that uses \emph{exactly two} inter-copy aggregate-$CZ$ matchings together with intra-copy ISA layers, of total depth at most $2+3\,c_{\mathrm{GL}}(Q_0)$. Here $c_{\mathrm{GL}}(Q_0)$ is the intra-copy ISA depth needed to synthesize a logical $\GL(k,\ftwo)$ basis change; it is bounded by a function of $(k,D_0)$ alone, so the depth is a constant independent of $r$ and $\ell$.
\end{proposition}
\begin{proof}
The gate touches only the $2k$ logical qubits of copies $a,b$; index them $0,\dots,k{-}1$ (copy $a$) and $k,\dots,2k{-}1$ (copy $b$). By Lemma~\ref{lem:diag-action}(ii) a codespace-preserving phase/$CZ$ (``$C$-type'') layer has logical image $\bigl(\begin{smallmatrix}\mathbf I&0\\\bC&\mathbf I\end{smallmatrix}\bigr)$ with $\bC\in\Sym_{2k}(\ftwo)$, and such layers compose by \emph{addition} of their $\bC$-blocks. Write $\bC=\bigl(\begin{smallmatrix}\bC_a&X\\X^\top&\bC_b\end{smallmatrix}\bigr)$ with cross-block $X\in\ftwo^{k\times k}$. The aggregate inter-copy $CZ$ of Theorem~\ref{thm:qlc-compose} is the $C$-type layer $A_{\mathrm{agg}}$ with $\bC=\bigl(\begin{smallmatrix}0&P\\P^\top&0\end{smallmatrix}\bigr)$, $P=L_XL_X^\top$ \emph{invertible} (self-duality), at depth $1$; the target $\mathrm{\overline{CZ}}_{i^{(a)},j^{(b)}}$ is the $C$-type layer with cross-block $E_{ij}$, i.e.\ $\bC=\bigl(\begin{smallmatrix}0&E_{ij}\\E_{ji}&0\end{smallmatrix}\bigr)$ (Lemma~\ref{lem:diag-action}(ii)).

\emph{Step 1 (basis change acts by congruence).} A per-copy logical basis change is an element $g=\bigl(\begin{smallmatrix}A&0\\0&A^{-\top}\end{smallmatrix}\bigr)$, $A\in\GL(k,\ftwo)$ (Prop.~\ref{prop:pi-logical}); it conjugates a $C$-type layer to a $C$-type layer, $\bC\mapsto A^{-\top}\bC A^{-1}$. Taking $g$ block-diagonal across the two copies, with copy-$a$ part $A_a$ and copy-$b$ part $A_b$, sends the aggregate's cross-block $P\mapsto A_a^{-\top}PA_b^{-1}=UPV^\top$ ($U,V\in\GL(k,\ftwo)$), leaving the intra-copy blocks zero. Completeness of the ISA (Cor.~\ref{cor:phys-generation}) places the entire subgroup $\GL(k,\ftwo)\hookrightarrow\Sp(2k,\ftwo)$ in the realized group, so each such $g$ is a codespace-preserving circuit of depth $\le c_{\mathrm{GL}}(Q_0)$. As $\Sp(2k,\ftwo)$ is a fixed finite group, its diameter in the ISA generators is a function of $k$ alone, and each generator costs $\le 3D_0$ layers (Lem.~\ref{lem:generation}); hence $c_{\mathrm{GL}}(Q_0)$ is bounded by a function of $(k,D_0)$, a synthesis cost on $k$ logical qubits only. \emph{When $A\in\rho(\Aut_{\mathrm{perm}}(Q_0))$ is a permutation automorphism, the conjugation is a relabeling of the matching wiring and adds no depth.}

\emph{Step 2 (every invertible cross-block is one aggregate).} Each conjugated aggregate $g\,A_{\mathrm{agg}}\,g^{-1}$ is $C$-type with cross-block $UPV^\top$. Since $P\in\GL(k,\ftwo)$, already $\{UP:U\in\GL(k,\ftwo)\}=\GL(k,\ftwo)$: every \emph{invertible} cross-block is realized by a single aggregate preceded by one copy-$a$ basis change.

\emph{Step 3 (matrix units are sums of two invertibles).} For $k\ge2$, every matrix unit $E_{ij}\in\ftwo^{k\times k}$ is a sum of two invertibles. If $i\ne j$, $E_{ij}=\mathbf I+(\mathbf I+E_{ij})$ and $\mathbf I+E_{ij}$ is unipotent ($\det=1$). If $i=j$, fix a transposition $\tau=(i\,i')$ with $i'\ne i$ (possible as $k\ge2$); then $E_{ii}=\tau+(\tau+E_{ii})$, where $\tau+E_{ii}$ restricts to $\bigl(\begin{smallmatrix}1&1\\1&0\end{smallmatrix}\bigr)$ on the $\{i,i'\}$ block and to the identity elsewhere, hence is invertible. (For $k=1$, $P$ is the scalar $1$ and $A_{\mathrm{agg}}$ already \emph{is} the addressable gate, depth $1$.)

\emph{Step 4 (assembly).} Write $E_{ij}=W_1+W_2$ with $W_1,W_2\in\GL(k,\ftwo)$ (Step 3), and let $g_t$ be the copy-$a$ basis change with $U_t=W_tP^{-1}$ (so the $t$-th conjugated aggregate has cross-block $U_tP=W_t$). Composing the two,
\[
\bigl(g_1A_{\mathrm{agg}}g_1^{-1}\bigr)\bigl(g_2A_{\mathrm{agg}}g_2^{-1}\bigr)
=\begin{pmatrix}\mathbf I&0\\[2pt] \bigl(\begin{smallmatrix}0&\,W_1+W_2\\(W_1+W_2)^\top&\,0\end{smallmatrix}\bigr)&\mathbf I\end{pmatrix}
=\begin{pmatrix}\mathbf I&0\\[2pt] \bigl(\begin{smallmatrix}0&E_{ij}\\E_{ji}&0\end{smallmatrix}\bigr)&\mathbf I\end{pmatrix},
\]
which is exactly $\mathrm{\overline{CZ}}_{i^{(a)},j^{(b)}}$. Every layer is codespace-preserving (the aggregates by Theorem~\ref{thm:qlc-compose}, the $g_t$ as ISA circuits). Merging $g_1^{-1}g_2$, the circuit is two aggregate matchings interleaved with three intra-copy basis-change blocks: total depth $\le 2+3\,c_{\mathrm{GL}}(Q_0)$, independent of $r$ and $\ell$. (Independence of $\ell$ also follows because every layer used is itself concatenation-invariant by Theorem~\ref{thm:qlc-compose}.)
\end{proof}

The constant $c_{\mathrm{GL}}(Q_0)$ is the only core-dependent quantity, and it specializes to the validated entries of Table~\ref{tab:qlc-depths}. For $\mathbb{Z}_4$ ($[[4,2,2]]$) and $\mathrm{AGL}(1,5)$ ($[[20,2,6]]$) the cross-copy term $P$ is the off-diagonal swap and the requisite basis changes are permutation automorphisms (Prop.~\ref{prop:pi-logical}): for $[[4,2,2]]$ the transposition $(0\,1)$ and for $[[20,2,6]]$ the permutation $(0\,4)(1\,3)(5\,10\,6\,15\,7\,17)(8\,19\,13)(9\,11\,16\,18\,14\,12)$ (cycles mapping each entry to its successor) act logically as $\bigl(\begin{smallmatrix}1&0\\1&1\end{smallmatrix}\bigr)$ and $\bigl(\begin{smallmatrix}1&1\\1&0\end{smallmatrix}\bigr)$ respectively, which together with the swap (a transposition for $[[4,2,2]]$, right translation by $(2,0)$ for $[[20,2,6]]$) generate all of $\GL(2,\ftwo)$. So by the last sentence of Step~1 they are absorbed into the matching wiring and $\mathrm{\overline{CZ}}_{i^{(a)},j^{(b)}}$ is realized in \textbf{depth $2$}, with two $C$-side matchings whose cross-blocks sum to $E_{ij}$, the explicit witness for Step~3 being, e.g.\ $E_{01}=\mathbf I+\bigl(\begin{smallmatrix}1&1\\0&1\end{smallmatrix}\bigr)$ (both summands invertible). For $D_9$ ($[[18,2,5]]$), $P=\mathbf I$ but the permutation automorphisms realize only $\{\mathbf I,\mathrm{swap}\}$ logically: the coset weight enumerator of the class of $\overline X_0\overline X_1$ ($\{6{:}75,\,10{:}171,\,14{:}9,\,18{:}1\}$) differs from that of $\overline X_0$ and $\overline X_1$ (both $\{5{:}36,\,9{:}184,\,13{:}36\}$), automorphisms preserve coset weight enumerators, so none maps $\overline X_0$ to $\overline X_0\overline X_1$ which any subgroup of $\GL(2,\ftwo)$ larger than $\{\mathbf I,\mathrm{swap}\}$ would require. So the $\GL(2,\ftwo)$ basis changes must be synthesized from the diagonal generators: the length-four word $\bC{=}E_{11}$, $\bB{=}E_{01}{+}E_{10}$, $\bC{=}E_{11}$, $\bB{=}E_{00}{+}E_{01}{+}E_{10}$ realizes the required transvection at layer cost $2+1+2+3$, so $c_{\mathrm{GL}}\le8$. Since $P=\mathbf I$ (or a free swap) serves as one of the two summands of Step~3, only one basis change must be synthesized, giving the constant $\le 2+2\,c_{\mathrm{GL}}\le18$. In all cases the per-generator depth is independent of $r$ and $\ell$, and only the count of generators in a \emph{generic} logical Clifford grows, as $O(r^2)$ (Prop.~\ref{prop:qlc-depth}(ii)).

This proof is perhaps best understood intuitively by an example. Consider two core codes (copies) $a, b$, each of parameters $[[n_0, 2, d_0]]$, and label the four logical qubits $a_0, a_1, b_0, b_1$. Suppose we wish to synthesize the addressable cross-block entangler $CZ(a_0, b_1)$. For clarity we illustrate with the aggregate coupling matched indices ($P=\mathbf I$); the same argument applies verbatim to the swap coupling realized by the depth-$2$ cores. First, perform the aggregate blockwise $CZ$, coupling the logical qubits as $CZ(a_0, b_0)\cdot CZ(a_1, b_1)$ and costing one layer of gate depth. Next, perform a logical basis change on copy $a$, sending $a_0 \rightarrow a_0$ and $a_1 \rightarrow a_1 + a_0$; this is an element of $\GL(2,\ftwo)$, which the ISA realizes by Corollary~\ref{cor:phys-generation} at constant depth $c_{\mathrm{GL}}(Q_0)$ and at zero depth, as a permutation automorphism, for the $\mathbb{Z}_4$ and $\mathrm{AGL}(1,5)$ cores. Apply one final aggregate cross-block entangling gate, now coupling the logical qubits as $CZ(a_0, b_0)\cdot CZ(a_1 + a_0, b_1)$, followed by inversion of the basis change. The net entangling operation is then:

$$ \big[\,CZ(a_0, b_0)\, CZ(a_1, b_1)\,\big]\cdot\big[\,CZ(a_0, b_0)\, CZ(a_1, b_1)\, CZ(a_0, b_1)\,\big] = CZ(a_0, b_1), $$

where we expanded $CZ(a_1 + a_0, b_1) = CZ(a_1, b_1)\cdot CZ(a_0, b_1)$ so that each of $CZ(a_0, b_0)$ and $CZ(a_1, b_1)$ appears twice and cancels, synthesizing the proper gate. The two aggregate layers thus realize the addressable $CZ(a_0, b_1)$ at total depth $\le 2 + 3\,c_{\mathrm{GL}}(Q_0)$ in general, and $2 + 2\,c_{\mathrm{GL}}(Q_0)$ here since the first aggregate needs no basis change, depth $2$ exactly for the $\mathbb{Z}_4$ and $\mathrm{AGL}(1,5)$ cores. This sequence exists for all Quantum Logic Codes by the arguments above establishing the preservation of the logical ISA across tiling and concatenation composition operations.

\section{Conclusion}

This work presents a detailed theory of transversal logic in high-rate stabilizer codes and a refined analysis of such logic in CSS codes. It presents a unified lower bound on circuit depths required for any stabilizer code to reach the full logical Clifford group, combining bounds on information-transfer limitations with entropic operator-specification requirements. It further outlines the concrete requirements for $W$-local gates that preserve codespaces in CSS codes and that participate in addressable logic, which forms the basis for a ladder of logical reach. This theory allows us to construct novel gates more efficiently, and we present a novel construction in the rotated surface code that was found leveraging the reductions in search complexity given by these requirements. 

We also design and present a family of codes we call \emph{Quantum Logic Codes} that is designed specifically to support fast transversal Clifford logic. The construction is relatively flexible, and incorporates the selection of an appropriate \emph{core} code, and composes via tiling and concatenation with a self-dual doubly-even inner code to produce large instances. The family has asymptotic parameters $[[n, \sqrt{n}, \Theta(n^{\beta})]]$ where $\beta$ is shown in one instance to be $\beta \approx 0.2823$.
The code construction can expand to form large-scale encoded qubit systems while simultaneously able to reach high code distances, all while preserving a constant-depth 2-local transversal basis that generates the full logical Clifford group. It does this through tiling and concatenating a dense logical core code, and we prove how this family construction procedure preserves the logical ISA of any such core. 

\begingroup
\renewcommand{\addcontentsline}[3]{}
\begin{acknowledgements}
We would like to thank Rui Chao for in-depth discussions and for encouraging the development of the 2D-toric code gates. We also want to thank many others in the NVIDIA quantum error correction group for collaborations. We thank Krysta Svore for creating a research environment that made this work possible. Throughout this work, AI tools were used to facilitate development, aid brainstorming, and mechanically verify and revise the proof structures. All of this work was completed while at NVIDIA, relying solely on information from the cited references and material learned during this affiliation.
\end{acknowledgements}
\endgroup

\bibliographystyle{unsrtnat}
\bibliography{refs}
\newpage
\makeatletter
\let\orig@addcontentsline\addcontentsline
\renewcommand{\addcontentsline}[3]{%
  \def\@tempa{#2}\def\@tempb{subsection}%
  \ifx\@tempa\@tempb\else\orig@addcontentsline{#1}{#2}{#3}\fi}
\makeatother
\section{Supplementary Material}
\label{sec:supp}

In this Supplementary Material, we provide formal rigorous proofs of propositions, lemmas, theorems, and constructions present in the main text. We go through this in a sequence matching the presentation of the main material, beginning with the scaling laws dictating minimum circuit depths for the synthesis of the full logical Clifford group in stabilizer codes.

\subsection*{Stabilizer Code Transversal Clifford Circuit Depth Lower Bounds}

\begin{lemma}[Pauli radius lower bound]
\label{lemma:pauli-weight}
For a stabilizer code $Q$ of parameters $[[n, k, d]]$, the \emph{Pauli radius} $w^{\star} := \max_{\overline{P}\ne\mathrm{\overline{I}}}\,\min_{P\in\overline{P}}\mathrm{wt}(P)$, the largest, over all nontrivial logical Pauli classes $\overline{P}$, of the class's minimum-weight representative, satisfies 
$$ w^{\star} \;\ge\; \Omega\bigg(\frac{k}{\log{n}}\bigg). $$
\end{lemma}
\begin{proof}
By definition of $w^{\star}$ every one of the $2^{2k}$ logical Pauli classes has a representative of weight at most $w^{\star}$, and representatives of distinct classes are distinct Pauli operators. So the ball of Pauli operators of weight at most $w^{\star}$ must hold at least $2^{2k}$ of them. There are exactly $\binom{n}{j}3^{j}$ Paulis of weight $j$, which we can see by first choosing the $j$ nonidentity positions, then one of $X,Y,Z$ on each. So: 
$$ \sum_{j=0}^{w^{\star}}\binom{n}{j} 3^j \;\geq\; 2^{2k}. $$

Since $3^{j}\le 3^{w}$ for $0\le j\le w$, and using the standard bound $\sum_{j=0}^{w}\binom{n}{j}\le (en/w)^{w}$ (for $1\le w\le n$), we have
$$ \sum_{j=0}^{w}\binom{n}{j} 3^j \;\le\; 3^{w}\sum_{j=0}^{w}\binom{n}{j} \;\le\; \Big(\frac{3en}{w}\Big)^{w}. $$
At $w=w^{\star}$ this gives $(3en/w^{\star})^{w^{\star}}\ge 2^{2k}$, hence $w^{\star}\log(3en/w^{\star})\ge 2k$. Using $\log(3en/w^{\star})\le\log 3en$ for $w^{\star}\ge1$,
$$  w^{\star} \;\geq\; \frac{2k}{\log{3ne}} \;=\; \Omega\bigg(\frac{k}{\log n}\bigg). $$

\end{proof}

Now we use the heavy Pauli weights to prove a limit on the speed of information transfer through a stabilizer code.

\begin{theorem}[Information Transfer Limit]
\label{theorem:pauli-spread}
For a stabilizer code $Q$ of parameters $[[n, k, d]]$, with Pauli radius $w^{\star} = \Omega\bigg(\frac{k}{\log{n}}\bigg)$, any $W \ge 2$-local circuit implementing a logical Clifford that maps a minimum-weight logical Pauli class of weight $d$ to a heaviest logical Pauli class with lightest representative weight $w^{\star}$ has depth:

$$ D \geq \log_W{\frac{w^{\star}}{d}}. $$

\end{theorem}
\begin{proof}
The projective logical Clifford group $\Sp(2k,\ftwo)$ acts transitively on the nonzero logical Pauli classes, so it contains an operation $\overline{U}$ carrying a minimum-weight $d$ logical class to the heaviest logical class, whose lightest representative has weight $w^{\star} = \Omega(k/\log n)$ by lemma~\ref{lemma:pauli-weight}. A depth-$D$ layer of $W$-local gates acts by growing the support of any Pauli by a factor of at most $W$ per layer, so it maps a weight-$d$ representative to a representative of $\overline{U}$'s image of weight at most $d \cdot W^D$. Since every representative of the heaviest class has weight at least $w^{\star}$, realizing $\overline{U}$ forces $d \cdot W^D \geq w^{\star}$, i.e. 

$$ D \geq  \log_W \frac{w^{\star}}{d} \;\ge\; \log_W \frac{c\,k}{d \log n}, $$
the second inequality by Lemma~\ref{lemma:pauli-weight} ($w^{\star}\ge c\,k/\log n$).
\end{proof}

The invariants in Theorem~\ref{theorem:pauli-spread} coincide with the \emph{disjointness} framework in~\cite{jochymoconnor2018disjointness}: the Pauli radius $w^{\star} = \max_{\overline P}\min\{|g| : g\in\overline P\}$ is referred to in \cite{jochymoconnor2018disjointness} as \emph{maximum distance} $d_\uparrow$ (the largest, over logical classes $\overline P$, of the class's minimum-weight representative), and the code distance $d$ is their \emph{minimum distance} $d_\downarrow$. The bound above is the level-$1$ (Pauli-to-Pauli) specialization of their support-spreading lemma, and the ratio $w^{\star}/d = d_\uparrow/d_\downarrow$ is precisely the quantity that governs which levels of the Clifford hierarchy a shallow transversal circuit can reach. Their disjointness $\Delta(Q)\geq 1$ obeys the structural bound
$$ \Delta(Q) \;\leq\; \min\Bigl\{\,d,\ \tfrac{n}{w^{\star}}\,\Bigr\}, $$
so $w^{\star} \leq n/\Delta(Q) \leq n$ and the spread term satisfies $\log_W(w^{\star}/d) \leq \log_W(n/d)$, a ceiling that is independent of $k$.

This relationship captures one key element of the structure of logic inside stabilizer codes, albeit a relatively weak one. $W$-local layers have the ability to connect qubit blocks in size exponential in the layer count, which can then quickly cover all Pauli operators. This turns out to only be a bounding behavior in cases of low-rate codes, as another relationship begins to dominate once $k = \omega(\sqrt{n \log n})$. 

\begin{lemma}[Codespace-Preserving $W$-Local Layer Bound]
\label{lemma:csp-bound}
For any $[[n,k,d]]$ stabilizer code, assuming $W\mid n$ (padded otherwise), the number of codespace-preserving $W$-local layers is bounded above by the total number of $W$-local layers,
$$ N_{n,W} \;\leq\; P_{n,W}\,|\Sp(2W,\mathbb{F}_2)|^{n/W}, \qquad P_{n,W}=\frac{n!}{(W!)^{n/W}(n/W)!}, $$
the number of partitions of the $n$ qubits into $n/W$ blocks of size $W$, times the number $|\Sp(2W,\mathbb{F}_2)|$ of $W$-qubit Cliffords available in each block; layers whose blocks have fewer than $W$ qubits contribute at most a further factor polynomial in $n$ (one term per block-size profile), absorbed in the $\oh(n)$ term below. In particular,
$$ \log N_{n,W} \;\leq\; \log P_{n,W} + \frac{n}{W}\log|\Sp(2W,\mathbb{F}_2)| \;=\; \frac{W-1}{W}\,n\log n + \oh(n) \;=\; \oh(n\log n). $$
\end{lemma}
\begin{proof}
A $W$-local layer is specified by a partition of the $n$ qubits into $n/W$ blocks of size $W$, of which there are $P_{n,W}$, together with one $W$-qubit Clifford modulo Pauli on each block, i.e.\ an element of $\Sp(2W,\mathbb{F}_2)$. This gives $P_{n,W}\,|\Sp(2W,\mathbb{F}_2)|^{n/W}$ layers in total; layers with smaller blocks are counted by the same expression for their block-size profile, each profile contributes less (per qubit, $|\Sp(2j,\mathbb{F}_2)|^{1/j}/\,(j!)^{1/j}$ increases with $j$), and there are only polynomially many profiles; the codespace-preserving layers are a subset, so their number is at most this product up to the stated polynomial factor. The entropy estimate follows from $\log P_{n,W}=\frac{W-1}{W}n\log n+\oh(n)$.
\end{proof}

\begin{theorem}[Clifford Entropy Bound]
\label{theorem:clifford-entropy}
For a stabilizer code $Q$ with permutation automorphism group $\Aut_{\mathrm{perm}}(Q)$, write $\rho(\Aut_{\mathrm{perm}}(Q))\le\Sp(2k,\ftwo)$ for its \emph{logical image} (the group of distinct logical Cliffords induced by physical permutation automorphisms. For CSS codes this is the embedding $V_\pi\mapsto\diag(V_\pi,V_\pi^{-\top})$ of Proposition~\ref{prop:pi-logical}). Let $N_{n,W}$ denote the number of distinct codespace-preserving $W$-local gate layers for $W \geq 2$ supported on the code. Then, for fixed $W$ and codes with $\log|\rho(\Aut_{\mathrm{perm}}(Q))| = o(k^2)$, the circuit depth $D$ required to implement the full logical Clifford group with codespace-preserving layers is at least (for unrestricted layers the same bound holds with $N_{n,W}$ replaced by the total layer count of Lemma~\ref{lemma:csp-bound}, with the same leading order):

$$ D \geq \log_{N_{n,W}} \frac{|\Sp(2k,\mathbb{F}_2)|}{|\rho(\Aut_{\mathrm{perm}}(Q))|} \geq \frac{2k^2 + k + c_k - \log |\rho(\Aut_{\mathrm{perm}}(Q))| }{\frac{W-1}{W} n \log n + 2n\left(W+\tfrac12 + \frac{c_W}{2W}\right) - \beta_W n + \oh(\log n)} = \Omega\bigg(\frac{k^2}{n \log n} \bigg).$$

\end{theorem}
\begin{proof}

The full Clifford group $\Sp(2k,\ftwo)$ has order $|\Sp(2k,\ftwo)| = 2^{k^2} \prod_{i=1}^{k} (2^{2i}-1)$ (the order of the finite symplectic group~\cite{taylor1992classical,dickson1901linear}). This group then requires $\log |\Sp(2k,\ftwo)| = 2k^2 + k + c_k$ for $c_k = \sum_{i=1}^k \log(1-2^{-2i})$, a bounded negative constant ($c_k\to-0.538$ as $k\to\infty$), total bits to represent. The \emph{logical image} $\rho(\Aut_{\mathrm{perm}}(Q))$, the set of distinct logical Cliffords induced by physical permutation automorphisms, reduces the number of bits required to specify all operations of the logical Clifford group by carving them into equivalence classes. Distinct physical automorphisms that induce the same logical action are counted only once.

Each codespace-preserving $W$-local layer is one of at most $N_{n,W}$ distinct choices (Lemma~\ref{lemma:csp-bound}), so a depth-$D$ circuit is a length-$D$ word over an alphabet of size $\le N_{n,W}$: there are at most $N_{n,W}^{D}$ such circuits, hence at most $N_{n,W}^{D}$ distinct logical Cliffords reachable with $D$ layers. We treat permutation automorphisms as having no $W$-local layer cost, as they are a relabeling of the physical qubits and can be absorbed into classical control. A permutation automorphism $\pi \in \Aut_{\mathrm{perm}}(Q)$ preserves the stabilizer group and acts on the $k$ logical qubits as some logical Clifford $g_\pi \in \Sp(2k,\ftwo)$; the set of all such $g_\pi$ is exactly the logical image $\rho(\Aut_{\mathrm{perm}}(Q))$. Composing a depth-$D$ circuit realizing $g$ with each automorphism realizes $g_\pi g$ at the same depth, so this free action enlarges the reachable set by at most a factor $|\rho(\Aut_{\mathrm{perm}}(Q))|$. (Automorphisms interleaved between layers reduce to this case: conjugating a codespace-preserving $W$-local layer by a permutation automorphism yields another such layer, so any interleaved word collapses to $D$ layers followed by a single automorphism.) Synthesizing the entire logical Clifford group therefore requires
$$ |\rho(\Aut_{\mathrm{perm}}(Q))| \cdot N_{n,W}^{D} \;\ge\; |\Sp(2k,\ftwo)|. $$
Taking logarithms and solving for $D$ gives
$$ D \;\ge\; \frac{\log|\Sp(2k,\ftwo)| - \log|\rho(\Aut_{\mathrm{perm}}(Q))|}{\log N_{n,W}} \;=\; \log_{N_{n,W}} \frac{|\Sp(2k,\ftwo)|}{|\rho(\Aut_{\mathrm{perm}}(Q))|}. $$
For the explicit form, substitute the order count $\log|\Sp(2k,\ftwo)| = 2k^2 + k + c_k$, and expand $\log N_{n,W}$ via Lemma~\ref{lemma:csp-bound}: its leading $\tfrac{W-1}{W} n\log n$ is the Stirling estimate of $\log P_{n,W}$, while the linear terms collect the per-block Clifford factor $\tfrac{n}{W}\log|\Sp(2W,\ftwo)| = \tfrac{n}{W}(2W^2 + W + c_W)$ with $c_W = \sum_{i=1}^{W}\log(1-2^{-2i})$. Combine this together with the Stirling correction $-\beta_W n$ of $\log P_{n,W}$. Under the hypothesis $\log|\rho(\Aut_{\mathrm{perm}}(Q))|=o(k^2)$ we arrive at
$$ D \;\ge\; \Omega\!\left(\frac{k^2}{n\log n}\right). $$
\end{proof}

Corollary~\ref{corollary:master-bound} follows by taking the maximum of the bounds of Theorems~\ref{theorem:pauli-spread} and~\ref{theorem:clifford-entropy}.

\subsection*{$W$-Layer Transversal Gate Theory}

Now we proceed to the $W$-layer transversal gate theory. First we prove Proposition~\ref{prop:7b}

\textbf{Proposition~\ref{prop:7b}} (Admissible per-block set $\mathcal{G}_{\mathrm{diag}}$).
\textit{Let $Q$ be an indecomposable CSS code and let $\mathcal{V} = \bigotimes_p V_p$ be a codespace-preserving, addressable (Def.~\ref{def:addressable}) $2$-local single-layer Clifford. Then the layer is not a global sector exchange (not every $A_{V_p}$ is zero), and a block with \emph{trivial diagonal} $A_{V_p}=D_{V_p}=I_2$ lies in the phase/CZ set}
\begin{widetext}
$$
\mathcal{G}_{\mathrm{diag}} \;:=\; \Bigl\{\, V \in \Sp(4,\mathbb{F}_2) \;:\; A_V = D_V = I_2,\; B_V,C_V \in \Sym_2(\mathbb{F}_2),\; B_V C_V = 0 \,\Bigr\},\qquad |\mathcal{G}_{\mathrm{diag}}| = 18.
$$
\end{widetext}

\begin{proof}[Proof of Proposition~\ref{prop:7b}]
Let
\[
\mathcal{V}
\;=\;
\begin{pmatrix}
\bA & \bB\\
\bC & \bD
\end{pmatrix}
\in \Sp(2n,\mathbb{F}_2)
\]
be the block-diagonal symplectic matrix induced by the physical layer, where
the $p$th diagonal block is
\[
V_p
\;=\;
\begin{pmatrix}
A_p & B_p\\
C_p & D_p
\end{pmatrix}
\in \Sp(4,\mathbb{F}_2).
\]
For $a\in S_X$ and $b\in S_Z$, direct conjugation gives
\begin{align}
\mathcal{V}X^a\mathcal{V}^{\dagger}
  &\sim X^{\bA a}Z^{\bC a}, \label{eq:2local-x-stab-conj}\\
\mathcal{V}Z^b\mathcal{V}^{\dagger}
  &\sim X^{\bB b}Z^{\bD b}, \label{eq:2local-z-stab-conj}\\
\mathcal{V}X^aZ^b\mathcal{V}^{\dagger}
  &\sim X^{\bA a+\bB b}Z^{\bC a+\bD b}. \label{eq:2local-css-stab-conj}
\end{align}
Since the code is CSS, the $X$ and $Z$ components of a stabilizer must
lie in $S_X$ and $S_Z$, respectively.  Codespace preservation is
equivalent to
\begin{align}
\bA S_X &\subseteq S_X, &
\bC S_X &\subseteq S_Z, \label{eq:2local-css-pres-x}\\
\bB S_Z &\subseteq S_X, &
\bD S_Z &\subseteq S_Z. \label{eq:2local-css-pres-z}
\end{align}
Because the input $X^a$ and $Z^b$ were already stabilizers, the diagonal
conditions may be written in reduced form:
\begin{align}
(\bA-I_n)S_X &\subseteq S_X, &
\bC S_X &\subseteq S_Z, \label{eq:2local-reduced-x}\\
\bB S_Z &\subseteq S_X, &
(\bD-I_n)S_Z &\subseteq S_Z. \label{eq:2local-reduced-z}
\end{align}
Blockwise then, for every $a\in S_X$ and $b\in S_Z$, and where we write $a^{(p)}$ to be the restriction of vector $a$ to the $p$-th block coordinates (and $b$ likewise), we have:
\begin{align}
\bigoplus_p (A_p-I_2)a^{(p)} &\in S_X, &
\bigoplus_p C_p a^{(p)} &\in S_Z, \label{eq:2local-block-x}\\
\bigoplus_p B_p b^{(p)} &\in S_X, &
\bigoplus_p (D_p-I_2)b^{(p)} &\in S_Z. \label{eq:2local-block-z}
\end{align}
These are the two-local analog of the four Guyot--Jaques one-local support
conditions: the reduced diagonal blocks $A_p-I_2$ and $D_p-I_2$ are precisely
the Hadamard-type part of the physical Clifford action.

Singular diagonal blocks are excluded or controlled as follows. A layer in which every $A_p=0$ (a genuine sector exchange on every block, e.g.\ the global Hadamard) is excluded by Def.~\ref{def:addressable}: the image of $X^x$ for any representative $x$ of a fixed $\overline{X}_i$ has $X$-part $\bA x=0$, so fixing $\overline{X}_i$ forces $x\in S_X$, contradicting nontriviality. General singular blocks are controlled at the logical level by Lemma~\ref{lem:confine-singular} below.

For a block with \emph{trivial diagonal} $A_p=D_p=I_2$, substituting into the block symplectic identities
\eqref{eq:sp4-identities} gives
\begin{align}
B_p &= B_p^{\top}, &
C_p &= C_p^{\top}, &
B_p C_p &= 0. \label{eq:2local-admissible-symp}
\end{align}
Since $B_p$ and $C_p$ are symmetric, $C_pB_p=0$ is equivalent to
$B_pC_p=0$.  Therefore every trivial-diagonal admissible block lies in
$\mathcal{G}_{\mathrm{diag}}$.

Finally, count the elements.  A symmetric $2\times2$ matrix over
$\mathbb{F}_2$ has the form
\[
\begin{pmatrix}
x & y\\
y & z
\end{pmatrix},
\qquad x,y,z\in\mathbb{F}_2,
\]
so there are $8$ choices.  If $B_p=0$, then $C_p$ is arbitrary, giving $8$
classes.  If $\rank B_p=2$, then $B_pC_p=0$ forces $C_p=0$, and there are $4$
invertible symmetric choices for $B_p$.  If $\rank B_p=1$, there are $3$
choices for $B_p$, and for each such choice the condition $B_pC_p=0$ permits
$C_p=0$ and exactly one nonzero rank-one symmetric $C_p$ supported on
$\ker B_p$, giving $3\cdot 2=6$ classes.  Thus
\[
|\mathcal{G}_{\mathrm{diag}}|
\;=\;
8+4+6
\;=\;
18.
\]
A trivial-diagonal 2-local physical gate can be chosen from the following: a tensor product of single qubit gates $I, S, $ or $SHS$, a $CZ$ gate composed with some product of $I$ or $S$, a $CZ^{\sharp} = (H \otimes H) \cdot CZ \cdot (H \otimes H)$ type entangler composed with some product of $I$ or $SHS$, or the all-ones entangler $\mathcal{J}$. These are the algebraically allowed trivial-diagonal blocks, and which of them preserve the codespace of a given CSS code $Q$ is code-specific. These elements can be written as the set: $\{I, S, SHS\}^{\otimes 2} \cup CZ\cdot\{I, S\}^{\otimes 2} \cup CZ^{\sharp}\cdot\{I, SHS\}^{\otimes 2} \cup \{\mathcal{J}\}$.
\end{proof}

Blocks with invertible $A_p\ne I_2$ are not exchanges but sector-preserving \emph{basis changes}: the $\mathrm{SWAP}$ block $A_p=D_p=\bigl(\begin{smallmatrix}0&1\\1&0\end{smallmatrix}\bigr)$, $B_p=C_p=0$ is admissible. On $[[4,2,2]]$ the layer $\mathrm{SWAP}(1,2)\otimes I_{3,4}$ preserves the code, inducing a logical $\GL(2,\ftwo)$ relabeling of the two logical qubits. Such blocks supply the logical basis-change factor $\mathcal{A}$ of Proposition~\ref{prop:8}.

The two lemmas below regarding singular diagonal blocks and blocks whose diagonals are not inverse-transposes of one another control the remaining freedom at the logical level, for any codespace-preserving $W$-local layer on a CSS code. Both rest on the support structure of block-diagonal maps: an operator built per block vanishes outside the blocks on which it is nontrivial.

\begin{lemma}[Confinement of singular diagonal blocks]
\label{lem:confine-singular}
Let $\mathcal{V}=\bigoplus_p V_p$ be a codespace-preserving $W$-local layer on a CSS code $Q$, and let $H$ be the union of the supports of the blocks with singular $A_p$. If the logical block $A^{\mathrm{log}}=L_Z\bA L_X^{\top}$ is singular, then some nonzero element of $\ker H_Z$ is supported entirely within $H$. In particular $A^{\mathrm{log}}\in\GL(K,\ftwo)$ whenever no such element is confined to $H$.
\end{lemma}
\begin{proof}
Codespace preservation gives $\bA S_X\subseteq S_X$, and conjugation carries normalizer elements to normalizer elements, so $\bA\ker H_Z\subseteq\ker H_Z$; thus $A^{\mathrm{log}}$ is the induced map on $\ker H_Z/S_X$, and it is singular exactly when some $u\in\ker H_Z\setminus S_X$ has $\bA u\in S_X$. The kernel $\ker\bA=\bigoplus_p\ker A_p$ vanishes on every block with invertible $A_p$, so $\ker\bA\subseteq\mathbb{F}_2^{H}$. If $\ker\bA\cap S_X\ne0$, a nonzero stabilizer element is confined to $H$ and we are done. Otherwise $\bA|_{S_X}$ is injective, so $\bA S_X=S_X$ by dimension count; choose $s\in S_X$ with $\bA s=\bA u$. Then $u'=u+s$ is nonzero (its logical class equals that of $u$), lies in $\ker H_Z$ and in $\ker\bA\subseteq\mathbb{F}_2^{H}$: a nontrivial logical $X$ representative confined to $H$, of weight at least $d$.
\end{proof}

\begin{lemma}[Confinement of shear blocks]
\label{lem:confine-shear}
Let $\mathcal{V}$ be as above and let $H'$ be the union of the supports of the blocks with $A_p^{\top}D_p\ne I$. Then
$$ (C^{\mathrm{log}})^{\top}B^{\mathrm{log}} \;=\; L_X\,\bigl(\bA^{\top}\bD+I_n\bigr)\,L_Z^{\top}, $$
and if $(C^{\mathrm{log}})^{\top}B^{\mathrm{log}}\ne0$ then a nontrivial logical $Z$ representative is supported entirely within $H'$; in particular $|H'|\ge d$. Consequently $(C^{\mathrm{log}})^{\top}B^{\mathrm{log}}=0$ (equivalently $D^{\mathrm{log}}=(A^{\mathrm{log}})^{-\top}$ when $A^{\mathrm{log}}$ is invertible) whenever no nontrivial logical $Z$ representative is confined to $H'$.
\end{lemma}
\begin{proof}
The inverse layer $\mathcal{V}^{-1}$ is codespace-preserving and block-diagonal with blocks $V_p^{-1}=\bigl(\begin{smallmatrix}D_p^{\top}&B_p^{\top}\\C_p^{\top}&A_p^{\top}\end{smallmatrix}\bigr)$, so $\bA^{\top}S_Z\subseteq S_Z$; further $\bA^{\top}\ker H_X\subseteq\ker H_X$ follows from $\bA S_X\subseteq S_X$ by duality, and $\bD\ker H_X\subseteq\ker H_X$ is normalizer preservation. For the identity, $(A^{\mathrm{log}})^{\top}D^{\mathrm{log}}=L_X\bA^{\top}L_Z^{\top}\,L_X\bD L_Z^{\top}=L_X\bA^{\top}\bD L_Z^{\top}$: each column of $\bD L_Z^{\top}$ lies in $\ker H_X$, on which $L_Z^{\top}L_X$ acts as the identity modulo $S_Z$, and $L_X\bA^{\top}$ annihilates $S_Z$ (since $\bA^{\top}S_Z\subseteq S_Z$ and $L_XH_Z^{\top}=0$); the identity then follows from the logical symplectic identity $(A^{\mathrm{log}})^{\top}D^{\mathrm{log}}+(C^{\mathrm{log}})^{\top}B^{\mathrm{log}}=I_K$ together with $L_XL_Z^{\top}=I_K$. Now suppose $(C^{\mathrm{log}})^{\top}B^{\mathrm{log}}\ne0$. Then for some row $z$ of $L_Z$ (a $Z$ representative in $\ker H_X$) the vector $y:=(\bA^{\top}\bD+I)z$ has $L_Xy\ne0$, so $y\notin S_Z$; and $y\in\ker H_X$ by the invariances above, so $y$ is a nontrivial logical $Z$ representative. Finally $\bA^{\top}\bD+I=\bigoplus_p(A_p^{\top}D_p+I)$ vanishes on every block with $A_p^{\top}D_p=I$, so $y$ is supported within $H'$ and $\mathrm{wt}(y)\ge d$.
\end{proof}

We now prove Proposition~\ref{prop:8}, describing the logical actions reachable by admissible 2-local physical gates that can produce addressable logical gates.

\textbf{Proposition~\ref{prop:8}} (2-local logical reach). \textit{For any codespace-preserving, addressable (Def.~\ref{def:addressable}) $2$-local single-layer Clifford $V$ on an indecomposable CSS code $Q$ encoding $K$ logical qubits, suppose that no nonzero element of $\ker H_Z$ or of $\ker H_X$ is supported entirely within the union of the supports of the blocks with singular $A_p$ or with $A_p^{\top}D_p\ne I_2$ (the \emph{confinement hypothesis}; Lemmas~\ref{lem:confine-singular} and~\ref{lem:confine-shear}). Then the symplectic image lies in $\mathcal{A}\cdot\mathcal{G}_{2L}(K)$ for a logical basis change $\mathcal{A}\in\GL(K,\mathbb{F}_2)$; when the basis change is realized by a permutation automorphism, $\mathcal{A}\in\rho(\Aut_{\mathrm{perm}}(Q))$ and the image lies in the automorphism-augmented set $\mathcal{G}_{2L+\pi}(Q) = \rho(\Aut_{\mathrm{perm}}(Q))\cdot\mathcal{G}_{2L}(K)$. Here
$$
\mathcal{G}_{2L}(K) \;:=\; \biggl\{ \begin{pmatrix} \mathbf{I}_K & \boldsymbol{\mathcal{B}} \\ \boldsymbol{\mathcal{C}} & \mathbf{I}_K \end{pmatrix} \biggm|\; \boldsymbol{\mathcal{B}}, \boldsymbol{\mathcal{C}} \in \Sym_K,\; \boldsymbol{\mathcal{B}}\boldsymbol{\mathcal{C}} = 0 \biggr\}
$$
is the $2$-local set and $\rho(\Aut_{\mathrm{perm}}(Q))\le\GL(K,\mathbb{F}_2)$ is the logical automorphism action (Proposition~\ref{prop:pi-logical}). On a canonical basis $L_X, L_Z$, the set $\mathcal{G}_{2L}(K)$ is composed of $\boldsymbol{\mathcal{C}}_{ii} = 1 \Leftrightarrow \mathrm{\overline{S}}_i$, $\boldsymbol{\mathcal{C}}_{ij} = 1 \Leftrightarrow \mathrm{\overline{CZ}}_{ij}$, and dually $\boldsymbol{\mathcal{B}}_{ii} \Leftrightarrow \mathrm{\overline{SHS}}_i$, $\boldsymbol{\mathcal{B}}_{ij} \Leftrightarrow \mathrm{\overline{CZ}}^{\sharp}_{ij}$.
}

\begin{proof}[Proof of Proposition~\ref{prop:8}]
Suppose first that every block has trivial physical diagonal $A_p=D_p=I_2$, with $B_p,C_p\in\Sym_2(\mathbb{F}_2)$ (the phase/CZ blocks of Proposition~\ref{prop:7b}). Then $\mathcal{V}$ induces no logical basis change ($A^{\mathrm{log}}=D^{\mathrm{log}}=I_K$). The general case is treated at the end of the proof via Lemmas~\ref{lem:confine-singular} and~\ref{lem:confine-shear}. Assembling the blocks, the diagonal
sectors collapse to the identity, $\bigoplus_p A_p = \bigoplus_p D_p = I_n$,
so in symplectic form
\[
\mathcal{V}
\;=\;
\begin{pmatrix} I_n & B \\ C & I_n \end{pmatrix},
\qquad
B := \bigoplus_p B_p, \quad C := \bigoplus_p C_p \ \in\ \Sym_n(\mathbb{F}_2),
\]
with $B$ and $C$ inheriting symmetry from the block decomposition.

A logical operator with symplectic coordinates splitting as the $X$-part and $Z$-part
$(\bar{u}, \bar{v}) \in \mathbb{F}_2^K \times \mathbb{F}_2^K$ has physical
representative $(L_X^{\top}\bar{u},\, L_Z^{\top}\bar{v})$, and the logical
coordinates of any normalizer element $(a, b)$ are read back by pairing
against the dual representatives, $(a, b) \mapsto (L_Z a,\, L_X b)$. This
readout annihilates the stabilizer ambiguity, since
$L_Z H_X^{\top} = L_X H_Z^{\top} = 0$, and inverts the embedding, since
$L_Z L_X^{\top} = L_X L_Z^{\top} = I_K$. As $\mathcal{V}$ preserves the
stabilizer group it carries normalizer elements to normalizer elements, so
the induced logical action is the conjugation
\[
\Phi(\mathcal{V})
\;=\;
\begin{pmatrix} L_Z & 0 \\ 0 & L_X \end{pmatrix}
\begin{pmatrix} I_n & B \\ C & I_n \end{pmatrix}
\begin{pmatrix} L_X^{\top} & 0 \\ 0 & L_Z^{\top} \end{pmatrix}
\;=\;
\begin{pmatrix} L_Z L_X^{\top} & L_Z B L_Z^{\top} \\[2pt]
                L_X C L_X^{\top} & L_X L_Z^{\top} \end{pmatrix}
\;=\;
\begin{pmatrix} I_K & \boldsymbol{\mathcal{B}} \\
                \boldsymbol{\mathcal{C}} & I_K \end{pmatrix},
\]
where $\boldsymbol{\mathcal{B}} := L_Z B L_Z^{\top}$ and
$\boldsymbol{\mathcal{C}} := L_X C L_X^{\top}$ are the logical blocks.

Both blocks are symmetric:
$\boldsymbol{\mathcal{B}}^{\top} = L_Z B^{\top} L_Z^{\top} = \boldsymbol{\mathcal{B}}$,
and likewise $\boldsymbol{\mathcal{C}}^{\top} = \boldsymbol{\mathcal{C}}$, since
$B$ and $C$ are symmetric. Since $\mathcal{V} \in \Sp(2n, \mathbb{F}_2)$ preserves the CSS code, its
induced action preserves logical commutation relations, so
$\Phi(\mathcal{V}) \in \Sp(2K, \mathbb{F}_2)$. Reading the third symplectic
identity of~\eqref{eq:sp4-identities} in dimension $2K$, namely
$A^{\top} D + C^{\top} B = I_K$, and substituting $A = D = I_K$ gives
$\boldsymbol{\mathcal{C}}^{\top} \boldsymbol{\mathcal{B}} = 0$; transposing and
using symmetry, $\boldsymbol{\mathcal{B}} \boldsymbol{\mathcal{C}} =
(\boldsymbol{\mathcal{C}}^{\top} \boldsymbol{\mathcal{B}})^{\top} = 0$. Hence
$\Phi(\mathcal{V}) \in \mathcal{G}_{2L}(K)$ in the no-basis-change case.

In general, the confinement hypothesis gives $A^{\mathrm{log}}=L_Z\bA L_X^{\top}\in\GL(K,\mathbb{F}_2)$ by Lemma~\ref{lem:confine-singular}, and $(C^{\mathrm{log}})^{\top}B^{\mathrm{log}}=0$ by Lemma~\ref{lem:confine-shear}, hence $D^{\mathrm{log}}=(A^{\mathrm{log}})^{-\top}$ by the symplectic identity. Factoring $\Phi(\mathcal{V})=\diag(A^{\mathrm{log}},(A^{\mathrm{log}})^{-\top})\cdot M$ leaves $M$ symplectic with identity diagonal blocks, so $M\in\mathcal{G}_{2L}(K)$ by the same identities as in the trivial-diagonal case; the basis change is the automorphism action of Proposition~\ref{prop:pi-logical} when it is realized by a permutation automorphism. Therefore $\Phi(\mathcal{V})\in\mathcal{A}\cdot\mathcal{G}_{2L}(K)$, lying in $\mathcal{G}_{2L}(K)$ modulo the basis change.

The logical operation set is read off the off-diagonal blocks: a diagonal
entry $\boldsymbol{\mathcal{C}}_{ii} = 1$ sends
$\overline{X}_i \mapsto \overline{X}_i \overline{Z}_i$ ($\mathrm{\overline{S}}_i$); an
off-diagonal pair
$\boldsymbol{\mathcal{C}}_{ij} = \boldsymbol{\mathcal{C}}_{ji} = 1$ couples
$\overline{X}_i \leftrightarrow \overline{Z}_j$ and
$\overline{X}_j \leftrightarrow \overline{Z}_i$ ($\mathrm{\overline{CZ}}_{ij}$); and
dually for $\boldsymbol{\mathcal{B}}$ under $X \leftrightarrow Z$
($\mathrm{\overline{SHS}}_i$ and $\mathrm{\overline{CZ}}^{\sharp}_{ij}$).
\end{proof}

We now move to proving the invariance of the size of the set of logical operators with respect to transversal layer locality, Proposition~\ref{prop:Winvariance}.

\textbf{Proposition~\ref{prop:Winvariance}} ($W$-invariance of logical reach, $W \ge 2$) \textit{Let $Q$ be an indecomposable CSS code on $n$ qubits encoding $K$ logical qubits, and let $2 \le W < n$. Every codespace-preserving, addressable (Def.~\ref{def:addressable}) $W$-local single-layer Clifford $V$ has logical image $\Phi(V)$ given by
$$
\Phi(V)\;\in\;\mathcal{A}\cdot\mathcal{G}_{2L}(K),\qquad \mathcal{A}\in\GL(K,\ftwo),
$$
under the confinement hypothesis of Proposition~\ref{prop:8}, applied to the layer's blocks at width $W$; when the basis change $\mathcal{A}$ is realized by a permutation automorphism, $\Phi(V)\in\mathcal{G}_{2L+\pi}(Q)=\rho(\Aut_{\mathrm{perm}}(Q))\cdot\mathcal{G}_{2L}(K)$.
}
\begin{proof}[Proof of Proposition~\ref{prop:Winvariance}]
Write $V=\bigotimes_k g_k$, for $g_k\in\Sp(2W,\ftwo)$ acting on a disjoint block $\sigma_k\subseteq[n]$ with $|\sigma_k|\le W$, and assemble the layer $\mathcal{V}=\bigl(\begin{smallmatrix}\bA&\bB\\\bC&\bD\end{smallmatrix}\bigr)$ with $\bA=\bigoplus_k A_{g_k}$ and likewise for $\bB,\bC,\bD$.

\emph{Step 1 (lifted block identities).} Each $g_k$ obeys $g_k^\top\omega\,g_k=\omega$ in $W\times W$ blocks, giving $A_{g_k}^\top C_{g_k},\,B_{g_k}^\top D_{g_k}\in\Sym_W$ and $A_{g_k}^\top D_{g_k}+C_{g_k}^\top B_{g_k}=I_W$. This is the dimension-$W$ form of \eqref{eq:sp4-identities}. These are the only constraints $W$-locality imposes, and their form is independent of $W$.

\emph{Step 2 (trivial diagonal $\Rightarrow$ set, independent of $W$).} By Proposition~\ref{prop:8} the logical image is $\Phi(V)=\bigl(\begin{smallmatrix}A^{\mathrm{log}}&B^{\mathrm{log}}\\C^{\mathrm{log}}&D^{\mathrm{log}}\end{smallmatrix}\bigr)\in\Sp(2K,\ftwo)$, with $A^{\mathrm{log}}=L_Z\bA L_X^\top$, $D^{\mathrm{log}}=L_X\bD L_Z^\top$, $B^{\mathrm{log}}=L_Z\bB L_Z^\top$, $C^{\mathrm{log}}=L_X\bC L_X^\top$. If $A_{g_k}=D_{g_k}=I_W$ for every block then $\bA=\bD=I_n$, so $A^{\mathrm{log}}=D^{\mathrm{log}}=L_ZL_X^{\top}=I_K$. The dimension-$2K$ symplectic identity then forces $C^{\mathrm{log}\top}B^{\mathrm{log}}=0$, while $A^{\mathrm{log}\top}C^{\mathrm{log}},\,B^{\mathrm{log}\top}D^{\mathrm{log}}\in\Sym_K$ give $B^{\mathrm{log}},C^{\mathrm{log}}\in\Sym_K$, so $B^{\mathrm{log}}C^{\mathrm{log}}=0$ and $\Phi(V)\in\mathcal{G}_{2L}(K)$. The set $\mathcal{G}_{2L}(K)$ carries \emph{no} dependence on $W$, and its generators $\mathrm{\overline{S}}_i,\mathrm{\overline{CZ}}_{ij},\mathrm{\overline{SHS}}_i,\mathrm{\overline{CZ}}^\sharp_{ij}$ are already realized by $W{=}2$ blocks. A wider block only re-factorises the same logical target. This is the $W$-invariance: the ambient set is the same at every locality $2\le W<n$.

\emph{Step 3 (controlling the diagonal).} It remains to bound $A^{\mathrm{log}},D^{\mathrm{log}}$. Lemmas~\ref{lem:confine-singular} and~\ref{lem:confine-shear} hold for blocks of any width: under the confinement hypothesis, $A^{\mathrm{log}}=L_Z\bA L_X^\top\in\GL(K,\ftwo)$ and $(C^{\mathrm{log}})^{\top}B^{\mathrm{log}}=0$, so $D^{\mathrm{log}}=(A^{\mathrm{log}})^{-\top}$ by the symplectic identity. Factoring $\Phi(V)=\diag(A^{\mathrm{log}},(A^{\mathrm{log}})^{-\top})\cdot M$ leaves $M$ with identity diagonal blocks, which resides in $\mathcal{G}_{2L}(K)$ by Step~2. The basis change lies in the $\GL(K,\ftwo)$ automorphism action of Proposition~\ref{prop:pi-logical} whenever it is realized by a permutation automorphism, giving $\Phi(V)\in\mathcal{G}_{2L+\pi}(Q)$. Finally $W=n$ places the whole encoded Clifford in one block, so the single-layer reach enlarges strictly beyond the set and can hit all of $\Sp(2K,\ftwo)$.
\end{proof}

\subsection*{The Ladder of Logical Reach}
\label{sec:ladder}

Combining $1$-local, $2$-local, $W$-local ($W>2$), code automorphisms, and multi-layer composition organises the reachable logical actions into a single nested ladder ordered by physical resource. Throughout, $\mathcal{T}^{WL}_{\mathrm{addr}}(Q):=\{\Phi(V): V\text{ a codespace-preserving addressable }W\text{-local single-layer Clifford on }Q\}$ denotes the code-specific single-layer logical reach at locality $W$. $\mathcal{T}^{2L}_{\mathrm{addr}}(Q)$ is the $W=2$ case. $\mathcal{G}^{1L}_{\mathrm{diag}}(Q)$ and $\mathcal{T}^{1L+\pi}_{\mathrm{addr}}(Q)$ denote the addressable $1$-local diagonal reach and its enlargement by a permutation automorphism. The summit of the ladder is the full logical Clifford group, reached by multi-layer composition of $2$-local set elements.

We define two useful terms to study this comprehensively: a \emph{matching-realizable} logical action, and the single-layer 2-local automorphism realizability set. The former captures a sufficient set of constraints required of any 2-local transversal gate layer to guarantee that it preserves the codespace, can perform required logical operations, and is a proper 2-local transversal gate. These constraints are useful for understanding the complexity of this problem: CSS preservation and Reach are both treatable as one linear system with an affine solution space, while the Matching constraint introduces quadratic constraints that add combinatorial complexity to the problem space. The latter definition folds in the logical automorphisms to the reachable 2-local logical operations in a combined statement.

\begin{definition}[Matching-realizable]
\label{def:matching-realizable}
Fix a CSS code $Q=\CSS(H_X,H_Z)$ with logical basis $L=(L_X,L_Z)$. Encode a $2$-local diagonal $C$-side layer by $\alpha=(\vec c,\vec e)\in\mathbb{F}_2^{n}\times\mathbb{F}_2^{\binom n2}$, with $c_p$ denoting a physical $S$ on qubit $p$, and $e_{pq}$ a physical $CZ$ on $\{p,q\}$. This corresponds to the symmetric matrix $C(\alpha)\in\Sym_n(\mathbb{F}_2)$, $C(\alpha)_{pp}=c_p$, $C(\alpha)_{pq}=e_{pq}$. Encode the $X$-dual ($B$-side) layer dually by $\beta=(\vec c',\vec f)$ with matrix $C(\beta)$. A target $(\boldsymbol{\mathcal{B}},\boldsymbol{\mathcal{C}})\in\Sym_K\times\Sym_K$, optionally composed with $\mathcal{A}\in\Aut^{\mathrm{log}}_{\mathrm{perm}}(Q)$, is \emph{matching-realizable on $Q$} if there exist $(\alpha,\beta,\pi)$ with $\pi\in\Aut^{\mathrm{phys}}_{\mathrm{perm}}(Q)$ applied before the diagonal layer, and realizable on a \emph{single} block partition, satisfying:
\begin{enumerate}[label=\textbf{(\Alph*)},leftmargin=2.4em]
\item \textbf{Matching (M):} the off-diagonal supports of $\vec e$ and $\vec f$ are matchings on $[n]$ with each qubit in at most one $CZ$, resp.\ $CZ^{\sharp}$;
\item \textbf{CSS preservation (K):} $C(\alpha)\,S_X\subseteq S_Z$ and $C(\beta)\,S_Z\subseteq S_X$;
\item \textbf{Reach (R):} $L_X\,C(\alpha)\,L_X^\top=\boldsymbol{\mathcal{C}}$, $L_Z\,C(\beta)\,L_Z^\top=\boldsymbol{\mathcal{B}}$, and $V_\pi=\mathcal{A}$ (Prop.~\ref{prop:pi-logical}).
\end{enumerate}
\end{definition}

\begin{definition}[Single-layer 2-local $\rtimes$ automorphism realizability set]
\label{def:S2Lpi}
For a CSS code $Q$ with $K$ logical qubits, the set of symplectic images reachable by one addressable $2$-local single-layer Clifford $V$ composed with one $\pi \in \Aut^{\mathrm{phys}}_{\mathrm{perm}}(Q)$ is
\begin{widetext}
$$
\mathcal{S}^{\mathrm{gen}}_{2L+\pi}(Q) \;:=\; \biggl\{\, \begin{pmatrix} \mathcal{A} & \boldsymbol{\mathcal{B}}\, \mathcal{A}^{-\top} \\ \boldsymbol{\mathcal{C}}\, \mathcal{A} & \mathcal{A}^{-\top} \end{pmatrix} \biggm|\; \mathcal{A} \in \Aut^{\mathrm{log}}_{\mathrm{perm}}(Q),\; \boldsymbol{\mathcal{B}}, \boldsymbol{\mathcal{C}} \in \Sym_K(\mathbb{F}_2),\; \boldsymbol{\mathcal{B}}\boldsymbol{\mathcal{C}} = 0,\; \text{matching-realizable on } Q \,\biggr\}.
$$
\end{widetext}
This is the matching-realizable subset of the automorphism-augmented set $\mathcal{G}_{2L+\pi}(Q)$ of the summary of results. In particular $\mathcal{S}^{\mathrm{gen}}_{2L+\pi}(Q) \subseteq \mathcal{G}_{2L+\pi}(Q)$, and a single addressable $2$-local layer is matching-realizable through its per-block structure (Proposition~\ref{prop:7b}), so $\mathcal{T}^{2L}_{\mathrm{addr}}(Q) \subseteq \mathcal{S}^{\mathrm{gen}}_{2L+\pi}(Q)$, with equality to $\mathcal{G}_{2L+\pi}(Q)$ when $\Aut^{\mathrm{log}}_{\mathrm{perm}}(Q)$ acts matching-realizably.
\end{definition}

\begin{theorem}[Ladder of reachable logical sets]
\label{thm:ladder}
For an indecomposable CSS code $Q$ encoding $K$ logical qubits, and under the confinement hypothesis of Proposition~\ref{prop:8}, the reachable logical sets nest as
\begin{align*}
\{\mathrm{\overline{I}}\} \;&\subseteq\; \mathcal{G}^{1L}_{\mathrm{diag}}(Q) \;\subseteq\; \mathcal{T}^{2L}_{\mathrm{addr}}(Q) \\
\;&\subseteq\; \mathcal{T}^{WL}_{\mathrm{addr}}(Q) \;\subseteq\; \GL(K,\ftwo)\cdot\mathcal{G}_{2L}(K) \quad (2 \le W < n) \\
\;&\subsetneq\; \bigl\langle\mathcal{G}_{2L}(K)\bigr\rangle \;=\; \Sp(2K, \mathbb{F}_2),
\end{align*}
with $\mathcal{T}^{1L+\pi}_{\mathrm{addr}}(Q) \subseteq \mathcal{S}^{\mathrm{gen}}_{2L+\pi}(Q) \subseteq \mathcal{G}_{2L+\pi}(Q) \subseteq \GL(K,\ftwo)\cdot\mathcal{G}_{2L}(K)$ alongside.
\end{theorem}
\begin{proof-sketch}
A $1$-local layer is a $2$-local layer with empty matching, giving the first containments; by Proposition~\ref{prop:1local-reach} the $1$-local reach already includes entangling couplings on some codes, so this step need not be strict. The $\subsetneq$ is strict because $\mathrm{\overline{H}}_i$ never lies inside the set: every element of $\GL(K,\ftwo)\cdot\mathcal{G}_{2L}(K)$ has invertible diagonal blocks. The middle containments $\mathcal{T}^{2L}_{\mathrm{addr}}(Q)\subseteq\mathcal{T}^{WL}_{\mathrm{addr}}(Q)\subseteq\GL(K,\ftwo)\cdot\mathcal{G}_{2L}(K)$ are given by Proposition~\ref{prop:Winvariance} and carry its confinement hypothesis. The single-layer-plus-automorphism realizable set $\mathcal{S}^{\mathrm{gen}}_{2L+\pi}(Q)$ of Definition~\ref{def:S2Lpi} sits between, $\mathcal{T}^{2L}_{\mathrm{addr}}(Q)\subseteq\mathcal{S}^{\mathrm{gen}}_{2L+\pi}(Q)\subseteq\mathcal{G}_{2L+\pi}(Q)$, and contains $\mathcal{T}^{1L+\pi}_{\mathrm{addr}}(Q)$; the final equality is the Generation Lemma~\ref{lem:generation}.
\end{proof-sketch}

\subsection*{Novel Construction Proofs}
Now we move to the novel constructions. We perform the following analysis to prove the stated properties of the rotated surface code gadget. We make use of the following definition of a diagonal layer and the reach map to present a full proof of the construction validity. 

\begin{definition}[Diagonal layer and reach map]
\label{def:reachmap}
Fix a CSS code $Q$ with stabilizer rowspans $S_X = \rowspan(H_X)$, $S_Z = \rowspan(H_Z)$ and canonical logical representatives $L_X, L_Z$ such that $L_X L_Z^\top = \mathbf{I}_K$. A \emph{diagonal $2$-local layer} is a circuit
\[
V \;=\; \prod_{q\,:\,c_q = 1} S_q \;\;\prod_{q<q'\,:\,e_{qq'} = 1} CZ_{q,q'},
\]
specified by a phase vector $c \in \mathbb{F}_2^n$ and a symmetric zero-diagonal $e \in \mathbb{F}_2^{n\times n}$ whose support $\{\{q,q'\} : e_{qq'} = 1\}$ is a \emph{matching} involving each qubit in at most one $CZ$. Its physical $C$-matrix is the symmetric matrix $C = \diag(c) + e \in \Sym_n(\mathbb{F}_2)$, with the $S$ gates on the diagonal and the $CZ$ adjacency off it. Writing $R_q^X \in \mathbb{F}_2^{K}$ for column $q$ of $L_X$ to denote the \emph{$X$-signature} of qubit $q$, recording which logical $\overline X_i$ contain $q$, the \emph{reach map} of $V$ is
\begin{equation}
\label{eq:qlc-reach}
\bC(V) \;:=\; L_X\, C\, L_X^\top
\;=\; \sum_q c_q\, R_q^X (R_q^X)^\top
\;+\; \sum_{q<q'\,:\,e_{qq'}=1}\!\bigl(R_q^X (R_{q'}^X)^\top + R_{q'}^X (R_q^X)^\top\bigr)
\;\in\; \Sym_K(\mathbb{F}_2).
\end{equation}
\end{definition}

\begin{lemma}[Action of a diagonal layer]
\label{lem:diag-action}
A diagonal $2$-local layer $V$ with physical $C$-matrix $C$ fixes every $Z$ operator and sends $X^a \mapsto X^a Z^{Ca}$ for all $a \in \mathbb{F}_2^n$. These layers have the following properties: 
\begin{enumerate}[label=(\roman*)]
\item \emph{(CSS preservation, projective.)} $V$ preserves the stabilizer group \emph{up to Pauli signs} if and only if $C\,a \in S_Z$ for every $a \in S_X$, i.e.\ $C\,S_X \subseteq S_Z$; the $Z$-stabilizers are fixed automatically. This binary support condition is \emph{necessary} but not sufficient for exact $+1$-codespace preservation: a diagonal layer carries the exact phase
\[
V X^a V^\dagger \;=\; i^{\sum_q c_q a_q}\,(-1)^{\sum_{q<q'} e_{qq'} a_q a_{q'}}\,X^a Z^{Ca},
\]
so exact preservation additionally requires this phase to equal $+1$ on every $a\in S_X$ (e.g.\ $S\otimes S$ on the CSS code $\langle XX,ZZ\rangle$ satisfies $C\,S_X\subseteq S_Z$ yet sends $XX\mapsto -XX\,ZZ$). The phase is $+1$ for all explicit constructions in this work.
\item \emph{(Logical action.)} If $V$ is CSS-preserving, its logical symplectic is $\left(\begin{smallmatrix}\mathbf{I}_K & 0\\ \bC(V) & \mathbf{I}_K\end{smallmatrix}\right)$: it fixes every $\overline Z_i$ and sends $\overline X_i \mapsto \overline X_i \prod_j \overline Z_j^{\,\bC(V)_{ji}}$. In particular $V = \mathrm{\overline{S}}_i \iff \bC(V) = E_{ii}$, and $V = \mathrm{\overline{CZ}}_{ij} \iff \bC(V) = E_{ij} + E_{ji}$.
\end{enumerate}
\end{lemma}

\begin{proof}
Every gate is diagonal in the $Z$ basis, so $V Z^b V^\dagger = Z^b$. The single-gate rules $S_q X_q S_q^\dagger = X_q Z_q$ and $CZ_{q,q'} X_q\, CZ_{q,q'} = X_q Z_{q'}$ multiply over the layer to $V X_q V^\dagger = X_q Z^{C e_q}$, with $C = \diag(c) + e$ symmetric, so we have $V X^a V^\dagger = X^a Z^{Ca}$.

(i) For $a \in S_X$ the operator $X^a$ is already a stabilizer, so $X^a Z^{Ca}$ is a stabilizer projectively (up to sign) iff $Z^{Ca}$ is also a stabilizer up to sign. Exact $+1$ preservation further requires the diagonal-layer phase be trivial on $a$.

(ii) When $C\,S_X \subseteq S_Z$, $V$ normalizes the stabilizer group and is therefore a logical Clifford. Acting on the logical $\overline X_i$ with representative $\ell_i$ (row $i$ of $L_X$) gives $V \overline X_i V^\dagger = \overline X_i Z^{C\ell_i}$. The $\overline Z_j$-content of the $Z$-string $C\ell_i$ is its symplectic pairing with $\overline X_j$, namely $\ell_j \cdot (C\ell_i) = (L_X C L_X^\top)_{ji} = \bC(V)_{ji}$, where we used $L_X L_Z^\top = \mathbf{I}_K$ so that $\overline X_j$ reads off exactly the $\overline Z_j$-coefficient. The identifications follow from $\mathrm{\overline{S}}_i : \overline X_i \mapsto \overline X_i \overline Z_i$ and $\mathrm{\overline{CZ}}_{ij} : \overline X_i \mapsto \overline X_i \overline Z_j,\ \overline X_j \mapsto \overline X_j \overline Z_i$. The reach-map identity $\bC(V) = L_X C L_X^\top = \sum_q c_q R_q^X(R_q^X)^\top + \sum_{q<q'} e_{qq'}(R_q^X(R_{q'}^X)^\top + R_{q'}^X(R_q^X)^\top)$ is the expansion of $L_X C L_X^\top$ into the diagonal ($c_q$) and off-diagonal ($e_{qq'}$) entries of $C$, since $R_q^X = L_X e_q$.
\end{proof}

With the diagonal layer action tool we now prove Theorem~\ref{thm:rotated-S}: that the construction has the right logical action, preserves the codespace, and has good circuit distance.

\begin{proof}[Proof of Theorem~\ref{thm:rotated-S}]
Let $V$ be the circuit shown in Construction~\ref{con:rotated-S}. It is a diagonal $2$-local layer with $c_q = 1$ on the $d$ phase qubits of (i) and $e_{qq'} = 1$ on the rotated pairs of (ii). We apply Lemma~\ref{lem:diag-action}.

\emph{Logical action.} The rotated surface code has $K = 1$, so each $X$-signature $R_q^X \in \mathbb{F}_2^{1}$ is simply $1$ if qubit $q$ lies on the logical-$X$ representative $\overline X$ and $0$ otherwise. In the reach map, every off-diagonal term is $R_q^X (R_{q'}^X)^\top + R_{q'}^X (R_q^X)^\top = 2\,R_q^X R_{q'}^X = 0$ over $\mathbb{F}_2$, so the $CZ$ gates contribute nothing to the logical action, and
\[
\bC(V) \;=\; \sum_q c_q\, R_q^X \;=\; \#\{\text{phase gates lying on } \overline X\} \pmod 2 .
\]
The phase qubits are the diagonal $(r,r)$ with $0\le r\le d-2$ and the qubit $(d-2,d-1)$. Of these, the only one on the left column $\overline X = \{(r,0)\}$ is the corner $(0,0)$. Hence $\bC(V) = 1$, and by Lemma~\ref{lem:diag-action}(ii) $V$ realizes $\mathrm{\overline{S}}$ exactly.

\emph{CSS preservation.} By Lemma~\ref{lem:diag-action}(i) it suffices that $C\,a \in S_Z$ for every $X$-stabilizer. We use the standard $[[d^2,1,d]]$ rotated-surface layout in which checks are indexed by a \emph{dual vertex} $(i,j)$ centered at $(i-\tfrac12,\,j-\tfrac12)$, touching the in-range data qubits among $\{(i{-}1,j{-}1),(i{-}1,j),(i,j{-}1),(i,j)\}$: an $X$-stabilizer $a_{(i,j)}$ sits at every $(i,j)$ with $i+j$ even, $0\le i\le d$ and $1\le j\le d-1$ (weight $4$ in the bulk, weight $2$ on the top/bottom rows $i\in\{0,d\}$), and a $Z$-plaquette $Z_{(i,j)}$ at every $(i,j)$ with $i+j$ odd, $1\le i\le d-1$ and $0\le j\le d$ (weight $4$ in the bulk, weight $2$ on the left/right columns $j\in\{0,d\}$); the left/right boundaries are $X$-rough, terminating the vertical $\overline X=\{(r,0)\}$, and the top/bottom are $Z$-rough, terminating the horizontal $\overline Z=\{(0,c)\}$. Using the closed forms of Construction~\ref{con:rotated-S}, we have the phase set $(r,c)\in P\Leftrightarrow (r{=}c\le d{-}2)$ or $(r,c){=}(d{-}2,d{-}1)$, and $CZ$-partner $\pi(r,c)=(c,r{+}1)$ for $c\le r\le d{-}2$ and $\pi(r,c)=(c{-}1,r)$ for $r<c\le d{-}1$. So the byproduct of each $X$-stabilizer is, \emph{exactly} (pointwise over $\mathbb{F}_2$),
\[
C\,a_{(i,j)} \;=\;
\begin{cases}
Z_{(j-1,\,i)}, & i<j,\\[2pt]
Z_{(j,\,i+1)}, & i\ge j,\ (i,j)\neq(d{-}1,d{-}1),\ i\neq d,\\[2pt]
0, & (i,j)=(d-1,d-1)\ \text{or}\ i=d,
\end{cases}
\]
Each nonzero byproduct is a single $Z$-plaquette: the diagonal reflection $(i,j)\mapsto(j,i)$, with the first coordinate decremented when $i<j$ and the second incremented when $i\ge j$ (at $i=j$ these agree, $Z_{(j,i+1)}=Z_{(i,j+1)}$). For $d\ge5$ the case $i\ge j$ supplies the weight-$2$ right-boundary plaquettes $Z_{(j,d)}$ at $i=d-1$, so no separate boundary case is needed. The two vanishing cases are genuinely $0$: the bottom-row $X$-stabilizers $i=d$ meet no phase qubit or $CZ$-partner, and at the corner $(d{-}1,d{-}1)$ the two diagonal phase byproducts cancel through their mutual $CZ$. With the index range stated above ($i+j$ even, $0\le i\le d$, $1\le j\le d-1$, so the corner $(0,0)$ is a $Z$- not an $X$-vertex), these cases are mutually exclusive and exhaustive, and form a genuine partition of the $X$-stabilizers carrying no $d$-specific dependence, so $C\,a\in S_Z$ throughout and $V$ preserves the codespace for all odd $d\ge3$. The exact phase of Lemma~\ref{lem:diag-action}(i) is $+1$ on every $X$-stabilizer: the diagonal checks (including the corner $(d{-}1,d{-}1)$) contain exactly two phase qubits and one $CZ$ pair internal to their support, contributing $i^2\cdot(-1)=+1$, and every other check contains neither.

\emph{Distance preservation.} Since $C\,a\in S_Z$ for every $X$-stabilizer and $V$ fixes every $Z$, we have $V\mathcal{S}V^\dagger=\mathcal{S}$ exactly: $V$ is a \emph{code automorphism}, so the conjugated code equals $Q$ and its static distance is exactly $d$ for every odd $d$. For fault tolerance the relevant notion is the circuit-fault distance, so we analyze this here as well. Each gate fixes $Z$ and spreads $X$ to weight $\le2$, and no single gate fault is undetectable as every $S$- or $CZ$-byproduct anticommutes with some $X$-stabilizer. The top-row $\overline{Z}$ representative is aligned with the logical operator representative, which meets exactly one $CZ$ pair $\{(0,0),(0,1)\}$. A single $CZ$ fault there covers two of its $d$ qubits, while the remaining $d-2$ each require a separate fault, giving $d_{\mathrm{circ}}\le d-1$. For the matching lower bound, an undetectable fault pattern on $f$ gate or idle locations produces a nontrivial logical Pauli $E$ whose support is covered by those locations, so $f\ge\mathrm{wt}(E)-P(E)$ with $P(E)$ the number of $CZ$ pairs contained in the support of $E$; since the pairs are disjoint, $f\ge\mathrm{wt}(E)/2\ge\lceil d/2\rceil$ unconditionally, so the gadget is single-fault-tolerant for every $d\ge5$. The exact value is governed by $P(E)$. A minimum-weight $Z$-class representative meets every column exactly once and changes row by at most one per column (else some $X$-check would overlap it oddly), so the quantity row-plus-column is nondecreasing along it, growing by at most two per column; a contained pair $\{(c,r),(r,c+1)\}$ pins this quantity to $c+r$ at column $r$ and to $c+r+1$ at column $c+1$. Two contained pairs are incompatible with these pins: disjoint column intervals would require the impossible growth $(c_1-r_1)+(c_2-r_2)\le-1$, and overlapping intervals contradict either the pins or the parity rule selecting which unit steps the checks allow. Hence $P(E)\le1$ on minimum-weight $Z$ representatives. A minimum-weight $X$-class representative meets every row exactly once, and the same pinning runs against the monotone direction, so $P(E)=0$ there. For the remaining cases (the mixed class and non-minimal representatives) we verify $f\ge d-1$ exhaustively for $d=3,5$, giving $d_{\mathrm{circ}}=d-1$ there and $d_{\mathrm{circ}}\in[\lceil d/2\rceil,\,d-1]$ in general. Thus $V$ preserves the code distance as an automorphism for every odd $d$, and is single-fault-tolerant for $d\ge5$. 
\end{proof}

\begin{proof}[Proof of Theorem~\ref{thm:toric-CZ}]
Let $V$ be the circuit shown in Construction~\ref{con:toric-CZ}. It is a diagonal $2$-local layer (Def.~\ref{def:reachmap}), and we apply Lemma~\ref{lem:diag-action}. The code has $K = 2$ with chosen representatives $\overline X_0 = \{h(r,0)\}_r$ (column-$0$ horizontal edges) and $\overline X_1 = \{v(0,c)\}_c$ (row-$0$ vertical edges), so the $X$-signatures are
\[
R_{h(r,0)}^X = (1,0)^\top,\qquad R_{v(0,c)}^X = (0,1)^\top,\qquad R_q^X = 0 \ \text{otherwise.}
\]
The target is the logical controlled-$Z$, $\bC(V) = E_{01} + E_{10}$ (Lemma~\ref{lem:diag-action}). We sum the reach map of Def.~\ref{def:reachmap} over the gate families of Construction~\ref{con:toric-CZ}.

\emph{Row-$0$ pairs $h(0,c) \leftrightarrow v(0,c)$.} Here $R_{v(0,c)}^X = (0,1)^\top$ for every $c$, while $R_{h(0,c)}^X = (1,0)^\top$ only for the corner $c = 0$ (which lies on $\overline X_0$) and is $0$ otherwise. Hence only $c = 0$ contributes, giving $E_{01} + E_{10}$.

\emph{Phase gates: row $0$ on $h(0,c)$, and inner row $s = \lceil L/2 \rceil$ on $h(s,c)$.} A phase gate ($S^{\dagger}$ has the same $C$-matrix as $S$) contributes the diagonal term $R_q^X (R_q^X)^\top$, nonzero only on the column-$0$ qubits $h(0,0)$ and $h(s,0)$, each equal to $(1,0)(1,0)^\top = E_{00}$. The two families therefore contribute $E_{00} + E_{00} = 0$.

\emph{Cross pairs (tiers $r = 1,\dots,\lfloor(L-1)/2\rfloor$) and, for even $L$, the middle-row pairs.} Every endpoint lies in a row $\ge 1$, hence off $\overline X_1$; and whenever one endpoint is a column-$0$ horizontal edge $h(r,0) \in \overline X_0$, its partner has $R^X = 0$. Each such pair contributes $0$.

Summing, $\bC(V) = (E_{01} + E_{10}) + 0 + 0 = E_{01} + E_{10}$, so by Lemma~\ref{lem:diag-action}(ii) $V$ realizes the logical $\mathrm{\overline{CZ}}_{0,1}$ exactly. The construction is now transparent: the corner pair supplies the off-diagonal coupling, the two equal phase-gate families cancel the spurious diagonal $\mathrm{\overline{S}}_0$ that each would otherwise induce, and the cross- and middle-row pairs are logically inert.

\emph{CSS preservation.} By Lemma~\ref{lem:diag-action}(i) it suffices that $C\,A\in S_Z$ for every $X$-star $A$ (the weight-$4$ $X$-stabilizers); we show the stronger fact that $V$ sends each star to a single $Z$-plaquette. Write $A_{r,c}$ for the star at vertex $(r,c)$ and $B_{r,c}$ for the $Z$-plaquette. Each of the four edges of $A_{r,c}$ is touched by at most one $CZ$ (the off-diagonal support is a matching) and, on the two phase rows, possibly also by a diagonal $S$; the byproduct $C\,A_{r,c}$ collects both contributions. The cross-tier pairing reflects the row $r\mapsto(-r)\bmod L$; the skewed family $h(a,c)\leftrightarrow v(b,(c{+}1)\bmod L)$ and the aligned family $v(a,c)\leftrightarrow h(b,c)$ fix the plaquette column up to a shift of $0$ or $-1$. The row-$0$ and inner-row ($s=\lceil L/2\rceil$) phase gates supply the two horizontal edges that complete the otherwise-broken plaquette on the two extremal reflected rows, the same cancellation that, at the logical level, kills the spurious $\mathrm{\overline{S}}_0$. (For even $L$ the middle-row self-pair handles the fixed row $L/2$ with no column shift.) Tracing these, the byproduct is exactly one plaquette,
\[
C\,A_{r,c}\;=\;B_{r',c'}\in S_Z,\qquad r'=(-r)\bmod L,\quad c'\in\{c,\,c-1\}\ (\bmod\ L),
\]
the column shift being $-1$ exactly when the reflected row $r'$ lies in the first half $\{0,\dots,\lfloor(L-1)/2\rfloor\}$. The row reflection $r\mapsto r'$ is a bijection of rows, and for each image row the column map $c\mapsto c'$ is a bijection of $\mathbb{Z}_L$; hence $A_{r,c}\mapsto B_{r',c'}$ is a bijection of the $L^2$ stabilizers, $C\,S_X=S_Z$ exactly, and the layer is CSS-preserving for every $L\ge 3$.

\emph{Phases.} By the exact phase of Lemma~\ref{lem:diag-action}(i), with an $S^{\dagger}$ contributing $-i$ per supported qubit where an $S$ contributes $i$, a star on the row-$0$ (inner) phase row contains exactly two $S$ ($S^{\dagger}$) qubits and one $CZ$ pair internal to its support, contributing $i^{2}\cdot(-1)=+1$ ($(-i)^{2}\cdot(-1)=+1$); every other star contains neither, so the phase is $+1$ throughout and the stabilizer group is preserved exactly. The logical signs are exact as well: the representative of $\overline X_0$ meets one $S$ (at $h(0,0)$) and one $S^{\dagger}$ (at $h(s,0)$), contributing $i\cdot(-i)=+1$, while $\overline X_1$ meets no phase gate and no internal pair. With $S$ on both phase rows the layer would instead realize $\mathrm{\overline{CZ}}_{0,1}\cdot\overline{Z}_0$, the same symplectic element; the $S^{\dagger}$ row removes the logical Pauli byproduct.
\end{proof}

\subsection{The Core Codes and their Logical ISA}
\label{sec:qlc-cores}

Throughout this subsection each listed layer realizes its stated generator at the symplectic level with exact $+1$ stabilizer signs; logical Pauli byproducts (e.g.\ $\mathrm{\overline{S}}$ versus $\mathrm{\overline{S}}^{\dagger}$) may remain and are tracked in the Pauli frame.

\begin{construction}[{$[[4,2,2]]$ depth-1 logical-Clifford basis ISA}]
\label{con:qlc-422}
With $H_X = H_Z = (1\,1\,1\,1)$ and canonical logicals $\overline{X}_0 = X_0X_1$, $\overline{X}_1 = X_0X_2$, $\overline{Z}_0 = Z_0Z_2$, $\overline{Z}_1 = Z_0Z_1$, the generators are the single CSP $2$-local layers
\[
\begin{array}{lll}
\mathrm{\overline{S}}_0:\ S_0 S_2,\ CZ_{0,2}; &
\mathrm{\overline{S}}_1:\ S_0 S_1,\ CZ_{0,1}; &
\mathrm{\overline{CZ}}_{0,1}:\ CZ_{0,1},\ CZ_{2,3};\\[2pt]
\mathrm{\overline{SHS}}_0:\ \sqrt{X}_0\,\sqrt{X}_1,\ CZ^{\sharp}_{0,1}; &
\mathrm{\overline{SHS}}_1:\ \sqrt{X}_0\,\sqrt{X}_2,\ CZ^{\sharp}_{0,2}; &
\end{array}
\]
together with the logical $\mathrm{\overline{SWAP}} = \pi$ with $\pi = (1\,2)$, and the $\mathrm{\overline{H}} = \pi\cdot H^{\otimes 4}$ (depth $2$).
\end{construction}

\begin{figure*}[t]
\centering
\includegraphics[width=0.86\linewidth,keepaspectratio]{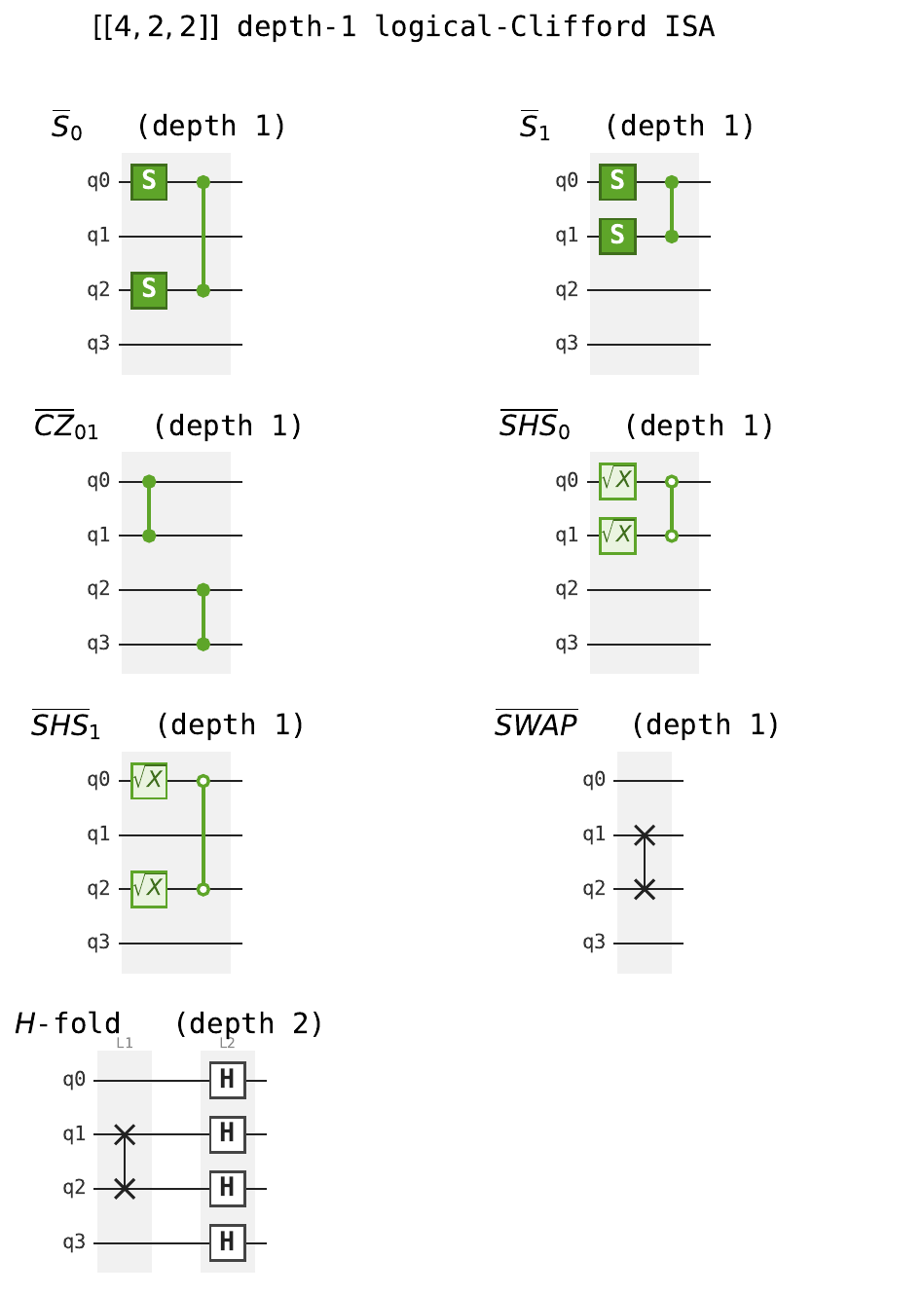}
\caption{Complete logical-Clifford basis ISA of the $[[4,2,2]]$ core (Construction~\ref{con:qlc-422}), with green $S$ boxes and $CZ$ connectors, hollow-green $\sqrt{X}$ ($=SHS$) and dual $CZ^{\sharp}$, qubit permutations as crossing wires. Every $\mathrm{\overline{S}}_i$, $\mathrm{\overline{CZ}}_{i,j}$, $\mathrm{\overline{SHS}}_i$ is a single codespace-preserving $W\le 2$ layer (depth $1$); the $\mathrm{\overline{H}}$ is a routing permutation followed by transversal $H$ (depth $2$).}
\label{fig:qlc-c422}
\end{figure*}

\begin{construction}[{$[[20,2,6]]$ depth-1 logical-Clifford basis ISA, $G=\mathrm{AGL}(1,5)$}]
\label{con:qlc-2026}
The self-dual core $\CSS(C,C)$ with $c=\{0,1,2,4,6,7,10,16\}\in\mathbb{F}_2[\mathbb{Z}_5\!\rtimes\!\mathbb{Z}_4]$ (elements $(a,b)\colon x\mapsto ax+b$ indexed $4b+j$ with $a=2^j \bmod 5$; $M(c)_{g,h}=1$ iff $gh^{-1}\in c$; $H_X=H_Z=H$, the $9\times20$ full-rank row basis of $\rowspan M(c)$) and canonical logicals
$L_X=\begin{psmallmatrix}1&0&1&0&0&1&1&0&1&0&1&0&0&0&0&0&0&0&0&0\\0&1&0&1&1&1&0&0&0&1&0&1&0&0&0&0&0&0&0&0\end{psmallmatrix}$,
$L_Z=\begin{psmallmatrix}0&1&0&1&1&1&0&0&0&1&0&1&0&0&0&0&0&0&0&0\\1&0&1&0&0&1&1&0&1&0&1&0&0&0&0&0&0&0&0&0\end{psmallmatrix}$
carries every generator as a single depth-1 CSP $2$-local layer (qubits $0$-indexed):
\[
\begin{aligned}
\mathrm{\overline{S}}_0:&\ S_{\{0,1,4,6,9,10,11,12,13,15,16,17,18,19\}},\ CZ_{0,1}CZ_{4,6}CZ_{9,16}CZ_{10,12}CZ_{11,17}CZ_{13,19}CZ_{15,18};\\
\mathrm{\overline{S}}_1:&\ S_{\{2,3,4,5,8,9,10,11,12,13,14,15,16,19\}},\ CZ_{2,3}CZ_{4,5}CZ_{8,9}CZ_{10,15}CZ_{11,12}CZ_{13,14}CZ_{16,19};\\
\mathrm{\overline{CZ}}_{0,1}:&\ S_{\{0,1,\dots,19\}}\ \text{(transversal }S,\text{ no }CZ);\\
\mathrm{\overline{SHS}}_0:&\ \sqrt{X}_{\{0,1,2,3,4,5,6,7,8,9,10,15,17,18\}},\ CZ^{\sharp}_{0,4}CZ^{\sharp}_{1,18}CZ^{\sharp}_{2,8}CZ^{\sharp}_{3,9}CZ^{\sharp}_{5,17}CZ^{\sharp}_{6,15}CZ^{\sharp}_{7,10};\\
\mathrm{\overline{SHS}}_1:&\ \sqrt{X}_{\{0,1,2,3,4,5,6,7,9,10,11,12,18,19\}},\ CZ^{\sharp}_{0,10}CZ^{\sharp}_{1,5}CZ^{\sharp}_{2,19}CZ^{\sharp}_{3,9}CZ^{\sharp}_{4,11}CZ^{\sharp}_{6,18}CZ^{\sharp}_{7,12}.
\end{aligned}
\]
The $\mathrm{\overline{H}}=\pi\cdot H^{\otimes 20}$ is depth $2$, with $\pi$ the right translation by the group element $(a,b)=(2,0)$, whose logical action is the swap. All five $C/B$-side generators are depth-1; with $\mathrm{\overline{H}}$ they generate $\Sp(4,\mathbb{F}_2)$.
\end{construction}

\begin{construction}[{$[[18,2,5]]$ logical-Clifford basis ISA, $G=D_9$}]
\label{con:qlc-1825}
The self-dual core $\CSS(C,C)$ with $c=\{0,2,5,7,11,12,13,17\}\in\mathbb{F}_2[D_9]$ (elements ordered $r^0,\dots,r^8,\,sr^0,\dots,sr^8$ for $D_9=\langle r,s\mid r^9=s^2=1,\ srs=r^{-1}\rangle$; $M(c)_{g,h}=1$ iff $h^{-1}g\in c$; $H_X=H_Z=H$, the $8\times18$ full-rank row basis of $\rowspan M(c)$) and canonical logicals
$L_X=L_Z=\begin{psmallmatrix}1&1&1&1&1&1&1&1&1&0&0&0&0&0&0&0&0&0\\0&1&0&0&1&1&0&1&0&1&0&0&0&0&0&0&0&0\end{psmallmatrix}$
carries $\mathrm{\overline{CZ}}_{0,1}$ as a single depth-1 CSP layer and $\mathrm{\overline{S}}_i$ (hence by the $X\!\leftrightarrow\!Z$ mirror $\mathrm{\overline{SHS}}_i$) as depth-2 sequences (qubits $0$-indexed):
\[
\begin{aligned}
\mathrm{\overline{CZ}}_{0,1}:&\ S_{\{2,5,9,12\}},\ CZ_{0,1}CZ_{3,4}CZ_{6,16}CZ_{7,17}CZ_{8,15}CZ_{10,11}CZ_{13,14};\\
\mathrm{\overline{S}}_0:&\ \mathrm{L1}:\ S_{\{5,8,9,15\}},\,CZ_{0,13}CZ_{1,14}CZ_{2,12}CZ_{3,16}CZ_{4,17}CZ_{6,10}CZ_{7,11};\\
&\ \mathrm{L2}:\ S_{\{2,5,8,9,11,12,14,15,17\}},\,CZ_{0,17}CZ_{1,15}CZ_{2,4}CZ_{3,14}CZ_{5,16}CZ_{6,11}CZ_{7,9}CZ_{8,13}CZ_{10,12};\\
\mathrm{\overline{S}}_1:&\ \mathrm{L1}:\ S_{\{5,8,9,15\}},\,CZ_{0,13}CZ_{1,14}CZ_{2,12}CZ_{3,16}CZ_{4,17}CZ_{6,10}CZ_{7,11};\\
&\ \mathrm{L2}:\ S_{\{1,2,4,8,9,11,14,16,17\}},\,CZ_{0,11}CZ_{1,10}CZ_{2,15}CZ_{3,17}CZ_{4,7}CZ_{5,8}CZ_{6,14}CZ_{9,12}CZ_{13,16}.
\end{aligned}
\]
Each $\mathrm{L}j$ is a codespace-preserving $W\le2$ matching, and the composite realizes the diagonal target in $\Sp(4,\mathbb{F}_2)$; the $\mathrm{\overline{SHS}}_i$ use the same patterns with $S\!\to\!\sqrt{X}$, $CZ\!\to\!CZ^{\sharp}$. An exhaustive SAT search over a complete encoding of the matching, CSS-preservation, and reach constraints of Definition~\ref{def:matching-realizable} found no depth-one solution (UNSAT), so depth-two is optimal for these generators. The $\mathrm{\overline{H}}$ is depth-1 on this core. The $D_9$ core attains $d=5$ at the cost of $\le 2$ layers per generator.
\end{construction}

\begin{proposition}[Cores carry $\Sp(4,\mathbb{F}_2)$]
\label{prop:qlc-cores}
Each of $[[4,2,2]]$, $[[18,2,5]]$ ($G=D_9$) and $[[20,2,6]]$ ($G=\mathrm{AGL}(1,5)$) carries a complete logical Clifford basis ISA realizing $\Sp(4,\mathbb{F}_2)$ (order $720$). The ISAs of $[[4,2,2]]$ (Construction~\ref{con:qlc-422}) and $[[20,2,6]]$ (Construction~\ref{con:qlc-2026}) are depth-1 (the $\mathrm{\overline{H}}$ being depth $2$); on $[[18,2,5]]$ (Construction~\ref{con:qlc-1825}) $\mathrm{\overline{CZ}}_{0,1}$ is depth-1 and the diagonal generators $\mathrm{\overline{S}}_i,\mathrm{\overline{SHS}}_i$ use $2$ layers (depth-1 is SAT-infeasible for them, so depth-2 is optimal).
\end{proposition}

\begin{proof}
We verify the $[[4,2,2]]$ generator $\mathrm{\overline{S}}_0$ by Eq.~\eqref{eq:qlc-reach}; the rest follow by the same computation. The columns of $L_X = \bigl(\begin{smallmatrix}1&1&0&0\\1&0&1&0\end{smallmatrix}\bigr)$ are $R_0^X = (1,1)^\top$, $R_1^X = (1,0)^\top$, $R_2^X = (0,1)^\top$, $R_3^X = 0$. The layer $S_0S_2,\,CZ_{0,2}$ has $c_0=c_2=1$, $e_{0,2}=1$, so
\[
\bC = R_0^X(R_0^X)^\top + R_2^X(R_2^X)^\top + \bigl(R_0^X(R_2^X)^\top + R_2^X(R_0^X)^\top\bigr)
= \begin{psmallmatrix}1&1\\1&1\end{psmallmatrix} + \begin{psmallmatrix}0&0\\0&1\end{psmallmatrix} + \begin{psmallmatrix}0&1\\1&0\end{psmallmatrix}
= \begin{psmallmatrix}1&0\\0&0\end{psmallmatrix} = E_{00},
\]
the rank-1 target of $\mathrm{\overline{S}}_0$, and $\bB = 0$. Codespace preservation: conjugating the single $X$-stabilizer $a = (1\,1\,1\,1)$ produces the $Z$-string $C a$ with $(C a)_q = c_q a_q + \sum_{q'} e_{qq'} a_{q'}$; the $S$ gates contribute $Z_0 Z_2$ and $CZ_{0,2}$ contributes $Z_0 Z_2$, summing to $0 \in S_Z$, while the $Z$-stabilizer is fixed by the diagonal layer. The diagonal-layer phase (Lemma~\ref{lem:diag-action}) on $a=(1\,1\,1\,1)$ is $i^{2}(-1)^{1}=+1$, so the stabilizer group is preserved exactly (sign included). By self-duality the $X\!\leftrightarrow\!Z$ mirror gives $\mathrm{\overline{SHS}}_0$ with $\bB = E_{00}$, and the seven generators close to $\Sp(4,\mathbb{F}_2)$. For $[[18,2,5]]$ and $[[20,2,6]]$ (Constructions~\ref{con:qlc-1825},\ref{con:qlc-2026}; Figs.~\ref{fig:qlc-d1825},\ref{fig:qlc-z2026}) the layer sequences are validated mechanically in the same fashion.
\end{proof}

\begin{figure*}[t]
\centering
\includegraphics[width=0.99\linewidth,keepaspectratio]{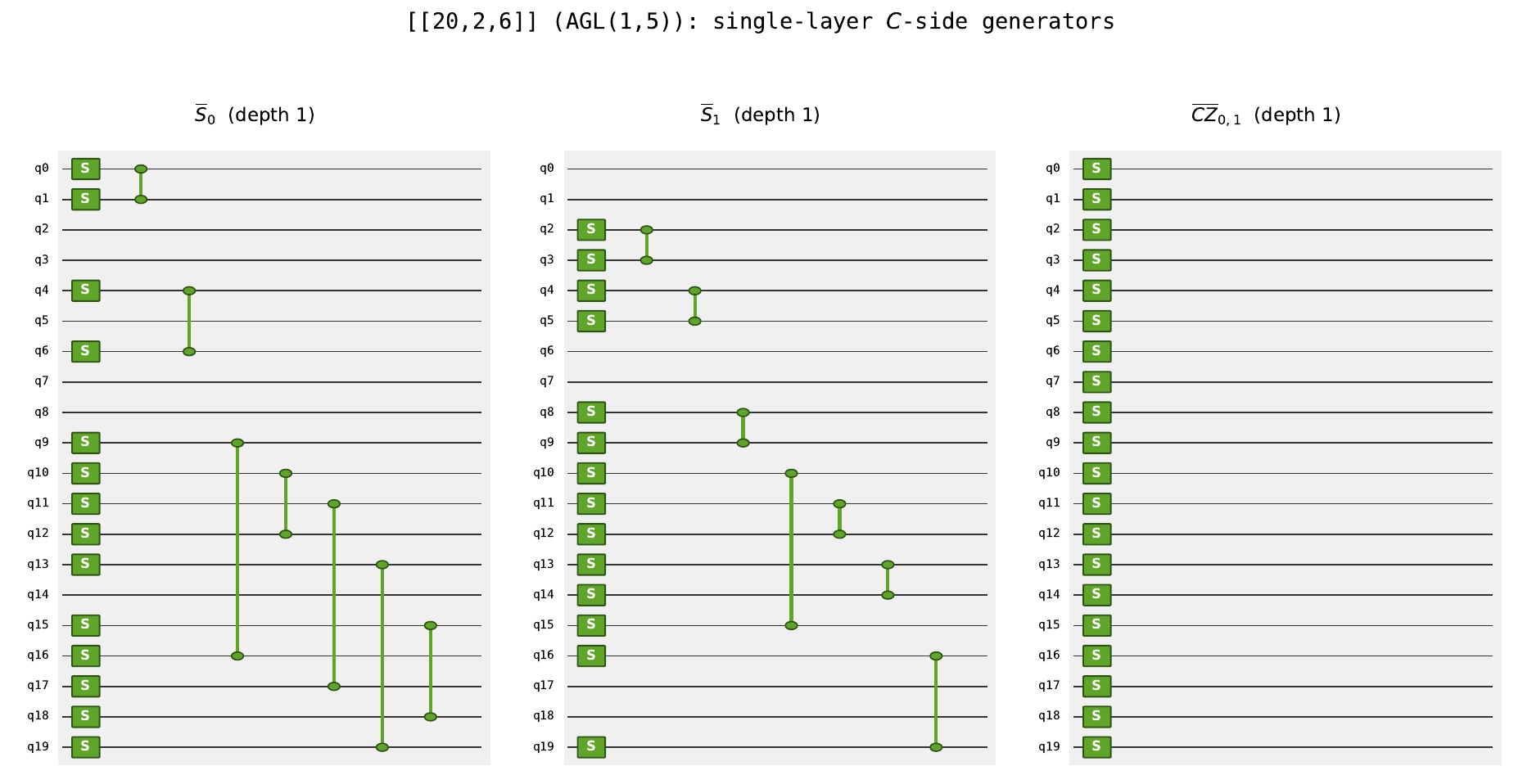}
\caption{Single-layer (depth-1) $C$-side generators of the $[[20,2,6]]$ core ($G=\mathrm{AGL}(1,5)$): $\mathrm{\overline{S}}_0$, $\mathrm{\overline{S}}_1$, $\mathrm{\overline{CZ}}_{0,1}$. Each is one codespace-preserving $W\le2$ matching of $S$ gates and $CZ$ pairs (shaded band); the $\mathrm{\overline{SHS}}_i$ are the $X$-duals. Same palette as Fig.~\ref{fig:qlc-c422}.}
\label{fig:qlc-z2026}
\end{figure*}

\begin{figure*}[t]
\centering
\includegraphics[width=0.95\linewidth,keepaspectratio]{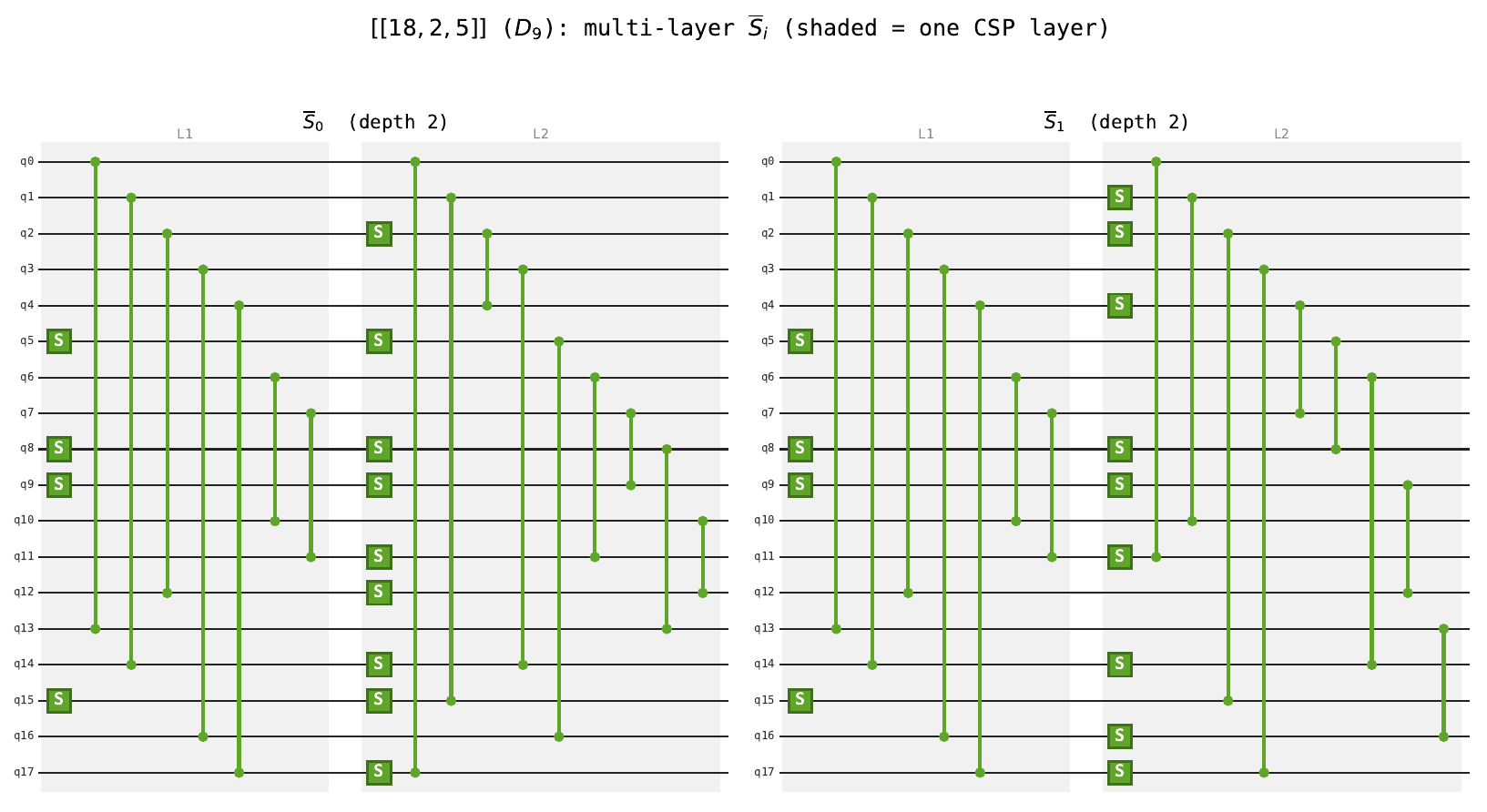}
\caption{Multi-layer $\mathrm{\overline{S}}_0$ and $\mathrm{\overline{S}}_1$ (each depth $2$) of the $[[18,2,5]]$ core ($G=D_9$). Each shaded band $\mathrm{L}j$ is one codespace-preserving $W\le2$ layer; the composition realizes the diagonal target in $\Sp(4,\mathbb{F}_2)$. The $D_9$ core attains $d=5$ at the cost of $\le2$ layers per generator (vs.\ depth-1 for $[[4,2,2]]$ and $[[20,2,6]]$).}
\label{fig:qlc-d1825}
\end{figure*}

\begin{figure}[t]
\centering
\includegraphics[width=0.74\linewidth,keepaspectratio]{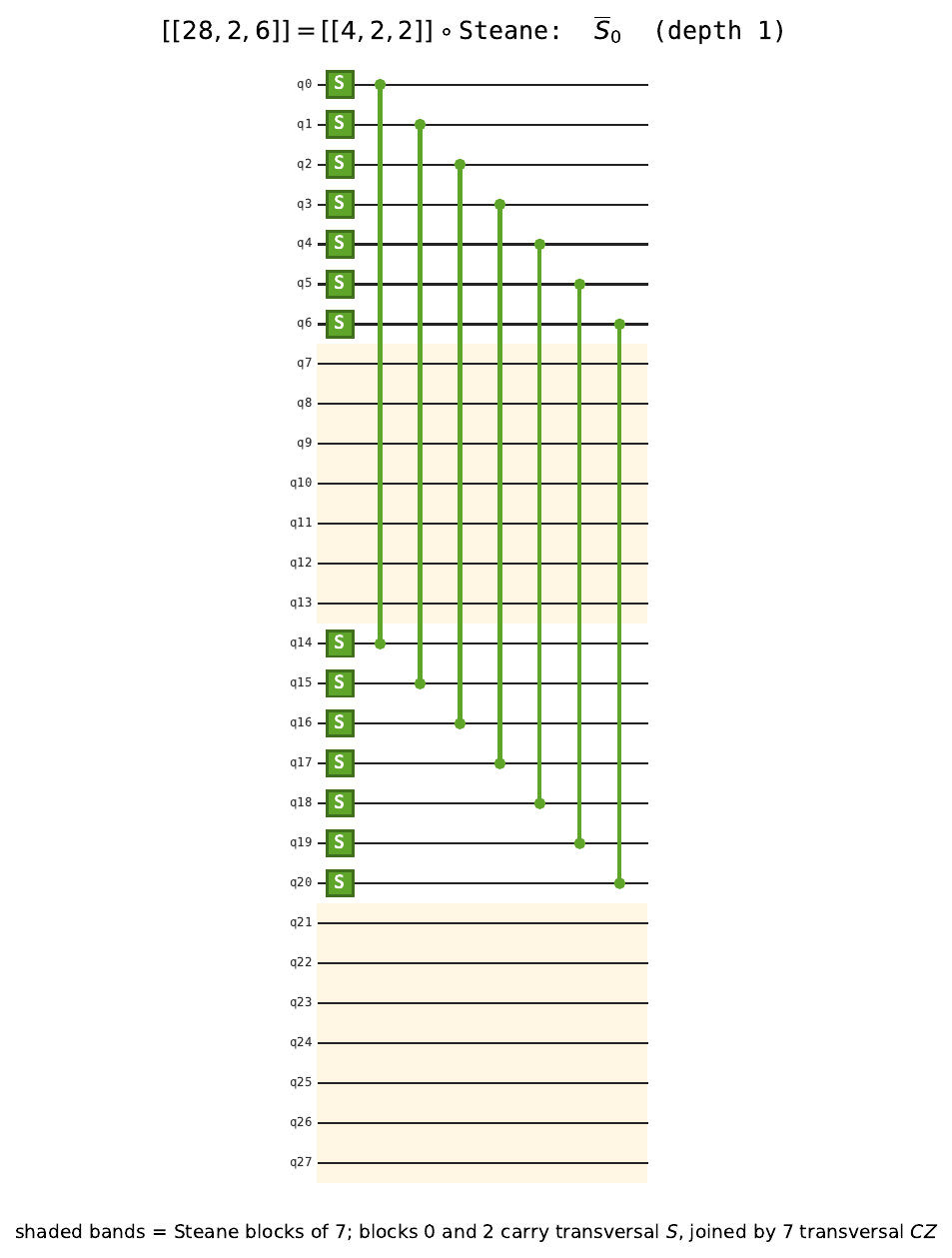}
\caption{The lifted $\mathrm{\overline{S}}_0$ on $[[28,2,6]] = [[4,2,2]]\circ\mathrm{Steane}$ as a single depth-1 codespace-preserving layer. Shaded bands are $[[7,1,3]]$ blocks of seven qubits; blocks $0$ and $2$ carry transversal $S$ (the lift of $S_0 S_2$), joined by seven transversal $CZ$ (the lift of $CZ_{0,2}$), exactly as in Fig.~\ref{fig:qlc-c422}'s $\mathrm{\overline{S}}_0$ but with the distance raised from $2$ to $6$.}
\label{fig:qlc-c2826}
\end{figure}

\end{document}